\documentclass[11pt,english]{article}
\usepackage[T1]{fontenc}
\usepackage[utf8]{luainputenc}
\usepackage{parskip}
\usepackage{babel}
\usepackage{enumitem}
\usepackage{amsmath}
\usepackage{amsthm}
\usepackage{amssymb}
\usepackage{graphicx}
\usepackage{geometry}
\usepackage{setspace}
\usepackage[authoryear]{natbib}
\usepackage[pdfusetitle,
 bookmarks=true,bookmarksnumbered=false,bookmarksopen=false,
 breaklinks=false,pdfborder={0 0 1},backref=false,colorlinks=false]
 {hyperref}

\makeatletter
\newenvironment{lyxcode}
	{\par\begin{list}{}{
		\setlength{\rightmargin}{\leftmargin}
		\setlength{\listparindent}{0pt}
		\raggedright
		\setlength{\itemsep}{0pt}
		\setlength{\parsep}{0pt}
		\normalfont\ttfamily}%
	 \item[]}
	{\end{list}}
\theoremstyle{definition}
 \newtheorem{example}{\protect\examplename}
\theoremstyle{plain}
\newtheorem{thm}{\protect\theoremname}
\theoremstyle{plain}
\newtheorem{lem}{\protect\lemmaname}
\theoremstyle{plain}
\newtheorem{prop}{\protect\propositionname}
\theoremstyle{remark}
\newtheorem{rem}{\protect\remarkname}
\theoremstyle{plain}
\newtheorem{fact}{\protect\factname}

\usepackage{booktabs}
\usepackage{float}
\usepackage{multirow, tabularx, graphicx}
\usepackage{makecell}

\usepackage{color}  
\definecolor{ucb}{RGB}{0, 50, 98}
\definecolor{booth}{RGB}{139, 0, 0}
\hypersetup{colorlinks=true,citecolor=booth,urlcolor=booth,anchorcolor=booth,linkcolor=booth}

\usepackage[toc,title,page]{appendix}

\makeatother

\providecommand{\examplename}{Example}
\providecommand{\factname}{Fact}
\providecommand{\lemmaname}{Lemma}
\providecommand{\propositionname}{Proposition}
\providecommand{\remarkname}{Remark}
\providecommand{\theoremname}{Theorem}

\begin{document}
\title{Comparisons of Experiments in Moral Hazard Problems}
\date{\today }
\author{Zizhe Xia\thanks{\protect\href{http://mailto:zizhe-xia@chicagobooth.edu}{zizhe-xia@chicagobooth.edu}
I am indebted to my committee members Doron Ravid, Emir Kamenica,
Lars Stole, and Alex Frankel for their invaluable guidance and advice.
I thank Eric Budish, Kailin Chen, Xiao Lin, Andrew McClellan, Daniel
Rappoport, Tyler Patterson, and Kun Zhang for their valuable comments
and suggestions. All errors are my own.}}\maketitle

\begin{abstract}
I use a novel approach to compare information in several classes of
moral hazard problems: implementability, cost under risk neutrality
and limited liability, and cost facing an agent with a general preference
for money. Incentives in moral hazard problems are determined by the
agent's state-dependent utility. Motivated by this observation, I
define three nested geometric orders on information: the column space,
the conic span, and the zonotope orders. Each order admits four equivalent
characterizations: (i) inclusion of feasible state-dependent utility
sets, (ii) dominance in the corresponding class of moral hazard problems,
(iii) matrix factorizations, (iv) posterior belief distributions.
In particular, the orders apply to both the classic and the flexible
moral hazard problems, providing a unified framework to compare information.

\end{abstract}

\thispagestyle{empty}

\newpage{}
\begin{lyxcode}
\setcounter{page}{1}
\end{lyxcode}

\section{Introduction}

In a canonical moral hazard (MH) setting, the principal relies on
a performance measure to contract with the agent and provide incentives.
In practice, there can be multiple such measures available. The same
underlying outcome is summarized differently by these measures, and
these differences affect how cheaply incentives can be provided. This
raises a central question: which performance measure is better? The
classic informativeness principle by \citet{holmstrom1979moral} gives
one answer: if one information structure augments another with additional
signals that contain new information about the agent's action, then
it strictly reduces the agency cost.\footnote{More precisely, the informativeness principle compares two information
structures that have an inclusive relation, that is, one contains
all of the other's signals plus some additional signals. In this case,
the information structure with additional signals always yields a
strictly lower agency cost if and only if the existing signals are
not sufficient statistics of the additional signals.} Beyond the informativeness principle, the literature offers a rich
array of likelihood-ratio-based comparisons.\footnote{See, for example, \citet{gjesdal1982information}, \citet{kim1995efficiency},
\citet{dewatripont1999economics}, \citet{demougin2001ranking}, and
\citet{chen2025experiments}. Section \ref{subsec:lit-review} provides
a summary of the papers that study the comparisons of information
in moral hazard problems.}

However, the likelihood ratio and informativeness approaches do not
provide a unified framework that identifies what information is more
valuable across different classes of moral hazard problems, nor do
they clarify how criteria for comparing information in moral hazard
settings relate to criteria in standard Bayesian decision problems
and the Blackwell \citeyearpar{blackwell1951comparison,blackwell1953equivalent}
order. As \citet{holmstrom1982moral} notes, ``whether the necessary
part of Blackwell's theorem, that information systems cannot be compared
unless one is sufficient for the other, is true in the agency framework
is still an open question.'' 

This paper studies the comparisons of information in general classes
of moral hazard problems with a geometric approach. The moral hazard
environment features a principal (she) who hires an agent (he) to
take a costly hidden action that stochastically affects the state.
The principal has access to a performance measure, modeled as some
noisy information about the state that can be contracted upon. One
can think of this performance measure as some accounting system or
evaluation scheme that maps the underlying outcome into a contractible
signal. To provide incentives, she offers a contract that pays the
agent based on the performance measure. The principal seeks to minimize
the expected cost of implementing a target action.\footnote{Following \citet{grossman1983implicit}, a moral hazard problem can
be solved in two steps: (i) cost minimization given a target action,
and (ii) profit maximization to select the best target action. The
main text focuses on (i), and (ii) is treated in Appendix \ref{appsec:profit-max}. } The only departure from classic moral hazard is that the principal
contracts on a performance measure rather than directly on the state. 

My analysis is driven by the following question: Under what conditions
does one performance measure yield a lower cost to the principal than
another, irrespective of the agent's cost function and the target
action?\footnote{The total cost to the principal decomposes into two parts: the first-best
cost, which is simply the production cost the principal would pay
if the agent's action were directly contractible, and the agency cost,
which is the extra expenditure required when the action is hidden.
Since the principal's information does not affect the first-best cost,
comparing total costs is equivalent to comparing agency costs alone.} Put differently, what performance measure makes the moral hazard
problem less severe?

I answer this question from a geometric perspective that focuses on
the agent’s state-dependent utility, the very object that governs
incentives. Formally, the state-dependent utility is the vector of
the agent's expected utilities across states. The agent's choice of
action is determined by his state-dependent utility. The principal's
problem then reduces to choosing a feasible state-dependent utility
for the agent. Her information completely specifies which state-dependent
utilities are feasible and at what cost. A larger feasible set is
always better. This insight motivates the comparisons of information
based on the set of feasible state-dependent utilities.

The main contribution of the paper is to compare information across
three classes of moral hazard problems. The first concerns implementability:
which target actions can be implemented. The second concerns costs
in moral hazard problems with risk neutrality and limited liability.
The third concerns costs across all moral hazard problems. I identify
three orders on experiments corresponding to the three classes of
problems: the column space, the conic span, and the zonotope orders.
Each gives sharp comparisons: one experiment dominates another in
that order if and only if it always yields better outcomes in the
corresponding class of problems. The three orders are nested. I also
provide equivalent characterizations of these orders in terms of matrix
factorizations and posterior beliefs. 

First, the column space order concerns the implementability problem.
The column space of an experiment is the set of state-dependent utilities
the principal can generate with unrestricted payments. One experiment
dominates another in this order if its column space contains the other's.
This order characterizes the comparisons of implementability in all
moral hazard problems. In other words, an experiment can always implement
more actions than another if and only if they are ranked by the column
space order. This is because a larger column space means the principal
can generate a richer variety of state-dependent utilities, and hence
implement more outcomes. This order does not speak to costs because
it only captures whether certain state-dependent utilities are feasible
without considering the cost.

Second, the conic span order captures cost comparisons under risk
neutrality and limited liability. The conic span of an experiment
is the set of state-dependent utilities the principal can generate
with non-negative payments. One experiment dominates another in this
order if its conic span contains the other's. This order characterizes
cost comparisons in moral hazard problems with a risk neutral agent
protected by limited liability. In other words, an experiment yields
a lower cost than another in all such problems if and only if they
are ranked by the conic span order. The reason mirrors the previous
case. A larger conic span provides the principal with a larger feasible
set in the cost minimization problem, weakly lowering the costs.

Third, cost comparisons more generally call for the zonotope order.
The zonotope of an experiment is the set of state-dependent utilities
the principal can generate with non-negative and bounded payments.
One experiment dominates another in this order if its zonotope contains
the other's.\footnote{In terms of the distribution posterior beliefs, the zonotope order
corresponds to dominance in the linear convex order, in contrast to
the convex order required by Blackwell (see Table \ref{tab:comparisons}).
This interpretation is not central to my approach, but it provides
an alternative perspective.} Again, by the larger-choice-set argument, this order naturally characterizes
cost comparisons in moral hazard problems where the agent is risk
neutral and protected by limited liability, and the principal faces
an ex post budget constraint. 

More importantly, the zonotope order also characterizes cost comparisons
in all moral hazard problems, regardless of the shape of the utility
from money. To gain intuition, suppose the agent is risk averse, since
non-concave utilities can be concavified with random payments. Concavity
effectively bounds the agent's utility from money on both sides: from
above, since the utility may satiate, and from below, since deeply
negative utility from large punishments would violate the participation
constraint. Thus, the zonotope of an experiment becomes the relevant
object for comparisons, and the larger-choice-set argument goes through.
The result also implies that comparing costs in all moral hazard problems
reduces to the case of risk neutrality with limited liability and
ex post budget.

The orders form a nested hierarchy, with the inclusions generally
strict. The column space order is the finest, followed by the conic
span order, then by the zonotope order, with Blackwell being the most
demanding. In special cases, some orders coincide. When the state
space is binary, zonotope coincides with Blackwell.\footnote{The equivalence between the zonotope order and the Blackwell order
when the state space is binary was first documented by \citet{blackwell1953equivalent}
without using the name “zonotope”. \citet{bertschinger2014blackwell}
later reintroduces it and provides an example to illustrate the difference
between the zonotope and the Blackwell orders with more than two states.} 

In addition, the conic span, zonotope, and Blackwell orders all coincide
when the experiments have no redundant realizations, that is, the
experiments have full column rank. This observation should not be
viewed as a limitation. Rather, it clarifies that the Blackwell order
is the appropriate order for moral hazard problems precisely when
experiments do not have redundancy.\footnote{For example, when there are more states than realizations, experiments
are generically non-redundant.} Once redundancy arises, the orders introduced in this paper continue
to provide meaningful comparisons while the Blackwell order becomes
less useful. 

\begin{table}[!t]
\caption{Comparisons of Experiments\label{tab:comparisons}}
\bigskip{}

\resizebox{\textwidth}{!}{
\begin{tabular}{ccccc}
\toprule
Order &
Set Inclusion &
Economic Problems &
Matrix Factorization &
Posterior Beliefs \\
\midrule\midrule

\makecell{Column space \\ $\geq_{\text{Col}}$} &
$\left\{ \mathcal{E}\boldsymbol{v} : \boldsymbol{v} \in \mathbb{R}^{M} \right\}$ &
\makecell{MH implementability} &
\makecell{$\mathcal{E}' = \mathcal{E}G$ \\ Any $G$} &
\makecell{Affine span \\ inclusion} \\
\midrule

\makecell{Conic span \\ $\geq_{\text{Cone}}$} &
$\left\{
  \begin{array}{c}
    \mathcal{E}\boldsymbol{v} : \boldsymbol{v} \in \mathbb{R}^{M}, \\
    \boldsymbol{v} \geq 0
  \end{array}
\right\}$ &
\makecell{MH with risk neutrality \\ and limited liability} &
\makecell{$\mathcal{E}' = \mathcal{E}G$ \\ $G \geq 0$} &
\makecell{Convex hull \\ inclusion} \\
\midrule

\makecell{Zonotope \\ $\geq_{\text{Zon}}$} &
$\left\{
  \begin{array}{c}
    \mathcal{E}\boldsymbol{v} : \boldsymbol{v} \in \mathbb{R}^{M}, \\
    0 \leq \boldsymbol{v} \leq 1
  \end{array}
\right\}$ &
\makecell{MH with risk neutrality, \\ limited liability and ex post budget; \\ MH with risk aversion} &
\makecell{The sum of any subset of \\ columns of $\mathcal{E}'$ lies in $\operatorname{Zon}\mathcal{E}$} &
Linear convex order \\
\midrule

\makecell{Blackwell \\ $\geq_{\text{B}}$} &
$\bigcup\limits_{K=1}^{\infty} \left\{
  \begin{array}{c}
    \mathcal{E} \Pi : \Pi \in \mathbb{R}^{M \times K}, \\
    \Pi \geq 0,\; \Pi\boldsymbol{1} \leq \boldsymbol{1}
  \end{array}
\right\}$ &
Any decision problem &
\makecell{$\mathcal{E}' = \mathcal{E}G$ \\ $G \geq 0,\; G\boldsymbol{1} = \boldsymbol{1}$} &
Convex order \\
\bottomrule
\end{tabular}
}

\bigskip{}

\noindent\begin{minipage}[t]{1\columnwidth}%
{\footnotesize Note: This table summarizes the results in the paper.
An experiment $\mathcal{E}$ is represented as a row stochastic matrix.
The first column lists the orders to compare experiments. The second
column presents the definitions of the orders based on set inclusion:
$\mathcal{E}$ dominates $\mathcal{E}'$ if the corresponding set
of $\mathcal{E}$ includes that of $\mathcal{E}'$. The third column
shows the matrix factorization condition for $\mathcal{E}$ to dominate
$\mathcal{E}'$.   The fourth column uses distributions of induced
posterior beliefs. For cumulative distribution functions $F,G:\mathbb{R}^{N}\to[0,1]$,
$F$ dominates $G$ in the linear convex order, denoted $F\geq_{\text{lcx}}G$,
if for any convex function $\phi:\mathbb{R}\to\mathbb{R}$ and any
vector $\boldsymbol{\beta}\in\mathbb{R}^{N}$, $\mathbb{E}_{\boldsymbol{x}\sim F}\left[\phi(\boldsymbol{\beta}\cdot\boldsymbol{x})\right]\geq\mathbb{E}_{\boldsymbol{x}\sim G}\left[\phi(\boldsymbol{\beta}\cdot\boldsymbol{x})\right]$,
where $\boldsymbol{\beta}\cdot\boldsymbol{x}$ represents the dot
product. $F$ dominates $G$ in the convex order, denoted $F\geq_{\text{cx}}G$,
if for any convex function $\phi:\mathbb{R}^{N}\to\mathbb{R}$, $\mathbb{E}_{\boldsymbol{x}\sim F}\left[\phi(\boldsymbol{x})\right]\geq\mathbb{E}_{\boldsymbol{x}\sim G}\left[\phi(\boldsymbol{x})\right]$.
The last column states the corresponding classes of moral hazard and
decision problems.}%
\end{minipage}
\end{table}

Lastly, I also provide easy-to-check equivalent characterizations
of each order in terms of matrix factorizations and posterior beliefs.
These alternative formulations link the geometric perspective to existing
notions in information economics. The results are summarized in Table
\ref{tab:comparisons}. The rest of the paper is organized as follows.
Section \ref{subsec:lit-review} discusses the related literature.
Section \ref{sec:MH-problems} introduces the moral hazard model.
Section \ref{sec:orders} studies the comparisons of information for
different classes of moral hazard problems. Section \ref{sec:extension-conclusion}
concludes.

\subsection{Related Literature \label{subsec:lit-review}}

This paper contributes to both the literature on moral hazard and
information ordering. 

The moral hazard literature begins with \citet{ross1973economic},
\citet{shavell1979risk}, \citet{holmstrom1979moral}, and \citet{mirrlees1999theory}.
Classic models typically take the agent's action space as a finite
or one-dimensional set. The principal observes the outcome and offers
contracts based on it. These models are usually analytically difficult:
the agent's optimization problem is generally non-convex, and the
first-order approach requires strong regularity conditions. More recently,
\citet{georgiadis2024flexible} propose a tractable variant where
the agent can flexibly choose a distribution over outcomes at a convex
and smooth cost. Flexibility convexifies the agent's problem, and
the first-order condition becomes both necessary and sufficient, though
the flexible setting does not nest the classic models. My model nests
both. I impose minimal restrictions on the agent's action space and
cost function. The comparison results developed below apply across
both settings. 

The comparisons of information start from Blackwell \citeyearpar{blackwell1951comparison,blackwell1953equivalent},
who proposes garbling as a way to compare the value of information
across all decision problems. Blackwell's comparison turns out to
be very restrictive. Many pairs of experiments are not comparable,
especially beyond the case of a binary state space. Follow-up works
explore ways to refine the Blackwell order by restricting attention
to monotone decision problems \citep{lehmann2011comparing,kim2023comparing},
problems with a single-crossing property \citep{persico2000information},
problems with the interval dominance property \citep{quah2009comparative},
binary decision problems \citep{chen2025experiments}, information
elicitation \citep{azrieli2022elicitability}, and information acquisition
and prediction \citep{xia2025expert}. \footnote{Apart from restricting the set of decision problems, there are other
ways to refine the Blackwell order. For example, \citet{brooks2024comparisons}
consider robustness to additional information that can be arbitrarily
correlated with the current information. \citet{moscarini2002law}
and \citet{mu2021blackwell} consider taking multiple independent
draws from the same experiments. } I follow the same route by restricting attention to different classes
of moral hazard problems. 

Since the celebrated informativeness principle by \citet{holmstrom1979moral},
economists have been studying what information is better at providing
incentives. \citet{grossman1983implicit} show that the Blackwell
order is sufficient for comparisons in moral hazard problems with
separable cost. \citet{gjesdal1982information} shows that even the
Blackwell order is not the correct order for comparisons if the action
taken affects the agent's risk attitude.\footnote{More concretely, \citet{gjesdal1982information} gives an example
where the principal has no information and the agent's risk aversion
depends on the action taken. In that example, a random incentive scheme
makes both the principal and the agent better off. } The results in my paper do require a notion of separability: the
action taken cannot affect the agent's risk attitude. Within this
separable setting, the literature sharpens the informativeness principle.
\citet{kim1995efficiency} extends the informativeness principle to
a mean-preserving-spread order on the distribution of likelihood ratios.
\citet{demougin2001ranking} provide an equivalent condition in the
integral form. \citet{dewatripont1999economics} consider this problem
in a career concern setting. More recently, \citet{chen2025experiments}
proposes using the linear convex order on the distribution of the
likelihood ratios. Whereas prior work focuses on likelihood ratio
properties, I develop the orders from the geometry of state-dependent
utilities. This allows me to study a general class of moral hazard
problems, encompassing both the classic and the flexible settings. 

Among the three orders studied in the paper, the column space order
first appears in \citet{azrieli2022elicitability} under a different
name. They study the problem of information elicitation, where a principal
contracts with an agent based on a noisy experiment to elicit the
agent's exogenously given belief. They compare experiments based on
what can be elicited from the agent and their order of elicitation
coincides with the column space order. I show that the same order
also characterizes the comparisons of the implementable actions in
moral hazard problems. The coincidence arises because, in both elicitation
and moral hazard, the agent's incentives are governed entirely by
his state-dependent utility. The column space order then exactly characterizes
the feasible set of state-dependent utilities the principal can provide. 

The conic span order appears in \citet{chatzikokolakis2020refinement}
in the context of differential privacy. They formalize it via induced
posterior beliefs. My contribution is to define it via sets of feasible
state-dependent utilities under non-negative payments, and to establish
its role in moral hazard problems. 

The zonotope order originates in \citet{blackwell1953equivalent},
who introduces it as the criterion to compare information in binary
decision problems, though without using the name “zonotope.” It is
later studied as a comparison of inequality (\citealt{koshevoy1995multivariate,koshevoy1996lorenz,koshevoy1997lorenz},
see \citealp{mosler2002multivariate} for a textbook treatment). \citet{bertschinger2014blackwell}
reintroduce it to information economics under the name ``zonotope''.
In a concurrent paper, \citet{chen2025experiments} refers to it as
the linear Blackwell order and illustrates how it sits between the
\citet{lehmann2011comparing} and the Blackwell orders. He also shows
that the zonotope order characterizes cost comparisons in flexible
moral hazard problems. My contribution, by contrast, is to shift to
a geometric perspective: incentives depend only on the set of state-dependent
utilities. This perspective allows me to work with a general model
that nests both the classic and the flexible settings. It also produces
the column space and the conic span orders, and yields a simple geometric
proof of the results.

\section{Moral Hazard Problems \label{sec:MH-problems}}

This section sets up a general model of moral hazard and defines several
classes of moral hazard problems considered in the analysis. 

\paragraph*{Notation}

Throughout the paper, I adopt the following notation conventions:
Matrices are denoted by uppercase letters, e.g., $\mathcal{E},G$;
vectors are denoted by boldface lowercase letters, e.g., $\boldsymbol{t}=(t_{1},t_{2},...,t_{M})$;
scalars are denoted by plain lowercase letters; inequality $\boldsymbol{t}\geq0$
means every entry of $\boldsymbol{t}$ weakly positive; boldface $\boldsymbol{1}$
and $\boldsymbol{0}$ are vectors of ones and zeros of conformable
shapes. 

\subsection{Principal-Agent Setup}

Fix a finite set of $N$ states $\Omega=\{\omega_{n}\}_{n=1}^{N}$.
The state represents some true performance outcome. 

A principal (she) hires an agent (he) to take a costly hidden action
that stochastically affects the state. The principal does not observe
the agent's action directly. Instead, she has access to a performance
measure, which is some noisy information about the realized state.
The performance measure turns the unobserved performance outcome into
a contractible signal. The principal can contract with the agent based
on the signal. Her goal is to maximize her profit, which is her payoff
from the agent's action minus the expected payments to the agent. 

The only departure from classic moral hazard models is that the principal
cannot directly contract on the state. Instead, she contracts on the
performance measure. Classic models are nested as a special case where
the performance measure fully reveals the state. In those models,
the state serves as both the payoff-relevant object and the information
to contract on. Separating the contractible information from the payoff-relevant
state allows me to compare information without otherwise changing
the underlying problem. 

The agent privately chooses a costly action from the set of available
actions. Let $A$ be the set of actions. Each action $a\in A$ induces
a distribution $\boldsymbol{\mu}_{a}$ over the state space $\Omega$.\footnote{If there are multiple actions leading to the same state distribution,
only the least costly one(s) are relevant.}  The agent's production technology is described by a cost function
$C:A\to\mathbb{R}_{+}$. The production cost $C$ is common knowledge.
I assume that $C$ is lower semi-continuous\footnote{Lower semi-continuity of $C$ is to be understood with respect to
the induced state distribution: identifying each action $a$ with
$\boldsymbol{\mu}_{a}$, I require $C$ to be lower semi-continuous
as a function on $\Delta\Omega$. This assumption ensures the agent's
problem admits a solution. It is generally assumed in both classic
and flexible moral hazard models and plays no other role in the analysis.} and is normalized to have a free option $\underline{a}$ with $C(\underline{a})=0$.
Let $\mathcal{C}$ be the set of all such cost functions. Let $\mathcal{A}$
be the collection of all possible action spaces.\footnote{Strictly speaking, there is no set of all sets. By $\mathcal{A}$
I mean a sufficiently rich action space, e.g., one indexed by $\Delta\Omega$
itself with $\boldsymbol{\mu}_{a}=a$. Any action space $A$ of interest
embeds into $\mathcal{A}$ by assigning $C(a)=+\infty$ for $a\notin A$. }

Beyond lower semi-continuity and normalization, the model imposes
no assumption on the agent's cost function $C$. In particular, $C$
need not be smooth or convex. This nests classic settings as well
as the flexible setting of \citet{georgiadis2024flexible}. Classic
settings with a discrete action space  or with a continuous effort
effort $a\in A=[0,1]$ directly map into my model. The flexible moral
hazard setting of \citet{georgiadis2024flexible} is also a special
case of my model. They assume that the agent directly chooses a state
distribution with a lower semi-continuous, convex, and smooth cost
function. Let $\mathcal{C}_{F}\subset\mathcal{C}$ denote the subclass
of convex and smooth cost functions defined on $\Delta\Omega$. The
flexible setting is therefore a special case of my model with $A=\Delta\Omega$
and $C\in\mathcal{C}_{F}$. The convexity of cost is without loss
in flexible problems, where the agent can randomize over the state
distributions. However, the cost function $C$ may fail to be convex
in classic settings.\footnote{For instance, in a continuous effort model where the agent chooses
effort $a\in[0,1]$, the cost is finite only at state distributions
$\left\{ \boldsymbol{\mu}_{a}\right\} _{a\in[0,1]}$, which need not
be a convex subset of $\Delta\Omega$.}

The agent's payoff is additively separable in his utility from money
and the production cost. His utility from money $u:\mathbb{R}\to\mathbb{R}$
is continuous, strictly increasing, unbounded, and normalized so that
$u(0)=0$. In the main text, I also assume that $u$ is concave for
ease of exposition.\footnote{With a concave utility, it suffices for the principal to consider
deterministic payments to the agent.} Appendix \ref{appsec:non-concave} shows that my analysis also applies
to more general utilities. Let $\mathcal{U}$ be the set of utilities
that satisfy all assumptions above. Later, I also consider the special
case where the agent is risk neutral with $u(t)=t$. The agent's payoff
when producing $a$ and receiving payment $t$ is $u(t)-C(a)$. He
maximizes his expected payoff and has an outside option $\underline{u}$.

The additive separable form $u(t)-C(a)$ is adopted for ease of exposition.
The analysis only requires that the agent's risk preference over money
is independent of his action. This is the notion of risk independence
in \citet{keeney1973risk}. This includes, for example, exponential
preferences of the form $-\exp\left[-\gamma\left(t-C(a)\right)\right]$
where $\gamma>0$ is the risk aversion parameter. But it rules out
settings where working harder changes the agent's risk attitude. Appendix
\ref{appsec:non-separable} formalizes the more general case. Without
the independence assumption, the comparisons developed in this paper
need not hold. For example, \citet{gjesdal1982information} provides
an example where adding pure noise can reduce agency costs. 

\subsection{Information}

Performance measures are modeled as finite (Blackwell) experiments
about the state. An experiment $\mathcal{E}:\Omega\to\Delta Y$ maps
each state to a distribution over realizations in $Y:=\{y_{m}\}_{m=1}^{M}$
. In applications, these realizations can be interpreted as the reports
generated by the accounting system or the scores produced by the evaluation
scheme. 

Any finite experiment $\mathcal{E}$ with $M$ realizations can be
represented as an $N\times M$ row stochastic matrix, that is, a matrix
whose entries are all weakly positive (formally, $\mathcal{E}\geq0$),
and rows sum to one (formally, $\mathcal{E}\boldsymbol{1}=\boldsymbol{1}$).
The $n$-th row of $\mathcal{E}$ represents the state $\omega_{n}$,
the $m$-th column of $\mathcal{E}$ represents the realization $y_{m}$,
and the $(n,m)$-th entry $\mathcal{E}_{n,m}:=\mathcal{E}(y_{m}\mid\omega_{n})$
is the conditional probability of realization $y_{m}$ in state $\omega_{n}$
with the following matrix form,
\[
\mathcal{E}:=\begin{array}{c}
\omega_{1}\\
\vdots\\
\omega_{N}
\end{array}\overset{\begin{array}{ccc}
y_{1} & \dots & y_{M}\end{array}}{\begin{bmatrix}\mathcal{E}_{1,1} & \dots & \mathcal{E}_{1,M}\\
\vdots & \ddots & \vdots\\
\mathcal{E}_{N,1} & \dots & \mathcal{E}_{N,M}
\end{bmatrix}}.
\]
I use $\mathcal{E}\left(\cdot\mid\boldsymbol{\mu}\right)\in\Delta Y$
to denote the induced distribution of realizations given the state
distribution $\boldsymbol{\mu}$, and $\pi\left(\mathcal{E};\boldsymbol{\nu}\right)\in\Delta\Delta\Omega$
to denote the distribution of posterior beliefs induced by $\mathcal{E}$
at prior $\boldsymbol{\nu}\in\Delta\Omega$. The posterior distribution
$\pi$ does not enter the moral hazard analysis directly. It appears
in alternative characterizations of the orders in Section \ref{sec:orders}.
Finally, let $E^{M}$ be the set of all experiments with $M$ realizations,
and $E=\cup_{M=1}^{\infty}E^{M}$ be the set of all finite experiments.

The principal's performance measure is an experiment $\mathcal{E}\in E$
with $M$ realizations in $Y=\{y_{m}\}_{m=1}^{M}$. It is commonly
known and its realization is contractible. 

\subsection{Contracting Problem}

The principal's objective is to choose a target action $a_{0}$ and
a contract to maximize her expected profit. This problem decomposes
into two steps, (i) cost minimization: for each candidate action $a_{0}$,
find the contract that implements $a_{0}$ at minimum cost; and (ii)
profit maximization: choose the action $a_{0}$ that maximizes the
principal's profit. Since the principal's information enters the problem
only through step (i), the main text focuses on the cost minimization
problem. Appendix \ref{appsec:profit-max} extends the analysis to
profit maximization. 

Following the decomposition above, I focus on the problem to implement
some target action $a_{0}$ at minimum cost. To provide incentives,
the principal offers a contract that ties the payment to the performance
measure. Formally, a contract under experiment $\mathcal{E}$ specifies
a payment rule $\boldsymbol{t}:Y\to\mathbb{R}$ which maps realizations
of $\mathcal{E}$ to payments. Since the experiment is finite, the
contract $\boldsymbol{t}$ can also be viewed as a vector in $\mathbb{R}^{M}$,
each component of which specifies the payment to the agent following
a realization of $\mathcal{E}$. It is without loss to assume that
the contract is deterministic when the agent is weakly risk averse.
Appendix \ref{appsec:non-concave} considers more general utility
functions and allows for random contracts.

Given a contract $\boldsymbol{t}$ under experiment $\mathcal{E}$,
the agent decides whether to accept it, and if so, what action to
choose. Formally, the agent's expected payoff from action $a$ is
given by 
\begin{equation}
U\left(a;\mathcal{E},\boldsymbol{t}\right):=\mathbb{E}_{y\sim\mathcal{E}\left(\cdot\mid\boldsymbol{\mu}_{a}\right)}\left[u\left(\boldsymbol{t}\left(y\right)\right)\right]-C\left(a\right).\label{eq:A-payoff}
\end{equation}
To implement action $a_{0}$, the contract $\boldsymbol{t}$ must
satisfy the incentive constraint (\ref{eq:IC}) and the participation
constraint (\ref{eq:IR}),

\begin{align}
a_{0} & \in\underset{a\in A}{\operatorname{argmax}}\:U\left(a;\mathcal{E},\boldsymbol{t}\right),\tag{IC}\label{eq:IC}\\
U\left(a_{0};\mathcal{E},\boldsymbol{t}\right) & \geq\underline{u},\tag{IR}\label{eq:IR}
\end{align}
so that the agent optimally chooses action $a_{0}$, and his payoff
is weakly above his outside option.\footnote{In case of indifference, I assume the agent always breaks ties in
favor of the principal. The incentive constraint is hence weak, in
the sense that it suffices for $a_{0}$ to be an optimal choice, rather
than the unique optimal choice.} 

A moral hazard problem is a tuple $P:=(A,a_{0},u,C,\mathcal{R})$
where $a_{0}\in A$ is the target action, $u\in\mathcal{U}$ is the
agent's utility for money, $C\in\mathcal{C}$ is the production cost,
and $\mathcal{R}\subseteq\{\text{LL},\text{B}\}$ denotes the set
of additional constraints on the contract.\footnote{The agent's outside option $\underline{u}$ is omitted for brevity.
The principal can always solve the moral hazard problem assuming $\underline{u}=0$,
and then provide the agent with a lump-sum payment that equals his
outside option. Since the outside option does not depend on the principal's
information, it does not affect the comparisons and is omitted from
the description of a moral hazard environment.} Depending on the application, the additional constraints may include
limited liability (\ref{eq:LL}), which requires the payments to be
non-negative, and ex post budget (\ref{eq:B}), which bounds the payments
from above by some constant $B$,
\begin{align}
\boldsymbol{t} & \geq0,\tag{LL}\label{eq:LL}\\
\boldsymbol{t} & \leq B.\tag{B}\label{eq:B}
\end{align}
Let $\mathcal{P}$ denote the set of all moral hazard problems. 

Given a problem $P:=(A,a_{0},u,C,\mathcal{R})$ and an experiment
$\mathcal{E}$, the principal solves
\begin{align}
\min_{\boldsymbol{t}} & \;\mathbb{E}_{y\sim\mathcal{E}\left(\cdot\mid\boldsymbol{\mu}_{0}\right)}\left[\boldsymbol{t}(y)\right],\text{ s.t. }\text{(\ref{eq:IC}), (\ref{eq:IR})},\mathcal{R},\tag{P}\label{eq:P-problem}
\end{align}
where $\boldsymbol{\mu}_{0}:=\boldsymbol{\mu}_{a_{0}}$ denotes the
state distribution induced by the target action $a_{0}$. Let $\mathcal{P}$
be the set of all problems $P$ with $A\in\mathcal{A},a_{0}\in A,u\in\mathcal{U},C\in\mathcal{C}$
and $\mathcal{R}\subseteq\{\text{LL},\text{B}\}$.

The formulation in (\ref{eq:P-problem}) assumes that the principal
is risk neutral over money. This is without loss: if the principal
has some general preference over money, problem (\ref{eq:P-problem})
can be rewritten in units of principal's utility and the analysis
goes through unaffected. Appendix \ref{appsec:non-concave} formalizes
this. 

\subsection{Comparing Experiments}

The main exercise of the paper is to compare experiments across different
classes of moral hazard problems. Let $\kappa(\mathcal{E};P)$ denote
the value of problem $P$. I write $\kappa(\mathcal{E};P)=+\infty$
when the problem is infeasible. This happens, for example, when $\mathcal{E}$
contains no information and the target action has a nonzero cost.
Comparisons of the form $\kappa(\mathcal{E};P)\leq\kappa(\mathcal{E}';P)$
are understood to hold when both sides are $+\infty$. 

I study two kinds of comparisons. The first is implementability, the
question of whether an action can be implemented. Given a set of actions
$A$, the utility $u$, and the production cost $C$, let 
\[
\mathcal{I}(\mathcal{E};A,u,C):=\left\{ a_{0}\in A:\kappa\left(\mathcal{E};(A,a_{0},u,C,\emptyset)\right)<\infty\right\} 
\]
be the set of actions that can be implemented. The comparisons of
implementability ask when the set is larger under some experiment
$\mathcal{E}$ than $\mathcal{E}'$ in set inclusion, regardless of
$u$ and $C$. 

The second is cost. The comparisons of cost ask when the cost $\kappa(\mathcal{E};P)$
is lower than $\kappa(\mathcal{E}';P)$ for any problem $P\in\mathcal{P}'$
within a certain class $\mathcal{P}'\subseteq\mathcal{P}$ of moral
hazard problems. I focus on the following classes. $\mathcal{P}_{1}:=\{P\in\mathcal{P}:u(t)=t,\mathcal{R}=\emptyset\}$
is the class of risk neutral problems with no contract constraint,
in which the principal can implement any feasible action at the first-best
cost and the comparisons of cost reduces to the comparisons of implementability.
Next, $\mathcal{P}_{2}:=\{P\in\mathcal{P}:u(t)=t,\mathcal{R}=\text{LL}\}$
is the class of risk neutral problems with limited liability, and
$\mathcal{P}_{3}:=\{P\in\mathcal{P}:u(t)=t,\mathcal{R}=\text{LL},\text{B}\}$
is the class with an additional ex post budget constraint. Risk neutrality
with limited liability is a workhorse assumption in contract theory
\citep{sappington1983limited,innes1990limited,poblete2012form,gottlieb2022simple}.
The case of bounded payments is further studied by \citet{jewitt2008moral}.
Lastly, $\mathcal{P}$ is the full class. 

I also define the class of flexible problems $\mathcal{P}_{F}:=\{P\in\mathcal{P}:C\in\mathcal{C}_{F}\}$
from \citet{georgiadis2024flexible}, which intersects with each of
$\mathcal{P}_{1},\mathcal{P}_{2},\mathcal{P}_{3},\mathcal{P}$ and
is technically convenient to solve.

\section{Comparisons of Experiments \label{sec:orders}}

In this section, I study the comparisons of experiments in different
classes of moral hazard problems. The results in this section are
summarized in Table \ref{tab:comparisons}. In Sections \ref{subsec:column-space}-\ref{subsec:zonotope},
I develop the orders to compare experiments for each class of moral
hazard problems. Section \ref{subsec:relations} discusses the relations
between these orders and their connections to the Blackwell order. 

The central object of moral hazard problems is the agent's state-dependent
utility. Formally, let $\boldsymbol{v}:=u(\boldsymbol{t})=\left[u(t_{m})\right]_{m=1}^{M}\in\mathbb{R}^{M}$
denote the vector of utilities associated with the contract $\boldsymbol{t}$
under experiment $\mathcal{E}$. The agent's state-dependent utility
is $\boldsymbol{u}:=\mathcal{E}\boldsymbol{v}\in\mathbb{R}^{N}$.
Its $n$-th component $u_{n}:=\sum_{m=1}^{M}\mathcal{E}_{n,m}u(t_{m})$
is the agent's expected utility from money in state $\omega_{n}$.
The vector $\boldsymbol{u}$ is central to both sides of the contracting
problem. On the agent's side, it determines incentives. Given $\boldsymbol{u}$,
the agent's problem is fully determined: 
\begin{equation}
\underset{a\in A}{\max}\:U\left(a;\mathcal{E},\boldsymbol{t}\right)=\boldsymbol{\mu}_{a}\cdot\boldsymbol{u}-C(a).\tag{A}\label{eq:A-problem}
\end{equation}
Two contracts that generate the same $\boldsymbol{u}$, possibly under
different experiments, yield the same optimal action, regardless of
how the experiments and the contracts differ otherwise. On the principal's
side, $\boldsymbol{u}$ is the only channel through which she can
affect the agent's action. Her contracting problem reduces to designing
$\boldsymbol{u}$ subject to feasibility, and which $\boldsymbol{u}$
is feasible depends on the experiment $\mathcal{E}$.

From this perspective, comparing experiments reduces to comparing
the sets of feasible state-dependent utilities they generate. A larger
feasible set is better, and this is the repeated theme of the results.
To see how this theme plays out, consider the following stylized example.
\begin{example}
\label{exa:main-example}A firm hires a worker to work on a project.
The project has two outcomes, left ($\omega_{\ell}$) and right ($\omega_{r}$).
The worker can exert private effort to affect the probability of the
outcomes. In particular, by choosing effort level $a\in[0,1]$ at
a cost of $C(a)=a^{2}$, he induces outcome $\omega_{r}$ with probability
$a$. His utility from money is given by some strictly increasing
and unbounded function $u:\mathbb{R}_{+}\to\mathbb{R}$. The firm
wants to implement $a=1$, but does not observe the outcomes directly.
This may be because the firm oversees numerous projects and cannot
monitor all of them closely. Instead, the firm has access to some
noisy performance measure, and can pay bonuses based on it. More concretely,
the firm considers the following performance measures:
\begin{enumerate}
\item [$\mathcal{E}_1$:] A binary experiment with realizations $Y_{1}=\{L,R\}$
that indicates either left ($L$) or right ($R$) and is correct with
probability 70\%.
\item [$\mathcal{E}_2$:] An experiment with realizations $Y_{2}=\{L,N,R\}$
that yields a completely uninformative signal ($N$) with probability
one half, and otherwise it correctly identifies $L$ or $R$ with
probability 80\%.
\item [$\mathcal{E}_3$:] An experiment with realizations $Y_{3}=\{sL,wL,wR,sR\}$
that indicates both whether its signal is strong or weak and which
outcome is more likely. The signal is equally likely to be strong
or weak. Strong signals $sL$ and $sR$ are correct with probability
80\%, and weak signals $wL$ and $wR$ are correct with probability
60\%. 
\end{enumerate}
Using the matrix notation introduced in Section \ref{sec:MH-problems},
the experiments can be represented as follows, 
\[
\mathcal{E}_{1}=\begin{bmatrix}0.7 & 0.3\\
0.3 & 0.7
\end{bmatrix};\;\mathcal{E}_{2}=\begin{bmatrix}0.4 & 0.5 & 0.1\\
0.1 & 0.5 & 0.4
\end{bmatrix};\;\mathcal{E}_{3}=\begin{bmatrix}0.4 & 0.3 & 0.2 & 0.1\\
0.1 & 0.2 & 0.3 & 0.4
\end{bmatrix}.
\]
The firm wishes to know how to optimally design the contract under
each experiment, and which experiment is better. 
\end{example}

\subsection{Implementability \label{subsec:column-space}}

I begin with implementability, the question of what can and cannot
be done with a given experiment. In the context of Example \ref{exa:main-example},
we want to know whether the firm can make it optimal for the worker
to produce $a=1$. Since the state space is binary, the worker's state-dependent
utility is $\boldsymbol{u}=(u_{\ell},u_{r})$ where $u_{\ell}$ is
his expected utility from money in state $\omega_{\ell}$, and $u_{r}$
in state $\omega_{r}$. His expected payoff is given by 
\[
(1-a)u_{\ell}+au_{r}-C(a).
\]
For optimality at $a=1$, the first-order condition requires that
the marginal benefit to increase the probability of $\omega_{r}$,
namely $u_{r}-u_{\ell}$, to exceed the marginal cost $C'(1)=2$.
Implementing $a=1$ therefore becomes a question of whether the firm
can generate some $\boldsymbol{u}$ with $u_{r}-u_{\ell}\geq2$. Under
each of the three experiments in Example \ref{exa:main-example},
this is possible with an appropriate bonus when the signal indicates
$\omega_{r}$ is more likely. This is not true for every experiment:
an uninformative experiment can only generate $\boldsymbol{u}$ with
$u_{r}=u_{\ell}$, and the firm cannot implement any costly action.

More generally, the comparisons of implementability ask whether an
experiment can implement a larger set of actions than another, regardless
of the agent's utility and cost. Formally, given experiments $\mathcal{E}$
and $\mathcal{E}'$, we want to know whether $\mathcal{I}(\mathcal{E};A,u,C)\supseteq\mathcal{I}(\mathcal{E}';A,u,C)$
for any action space $A$, utility $u\in\mathcal{U}$, and cost $C\in\mathcal{C}$.
This reduces to comparing the sets of feasible state-dependent utilities
$\boldsymbol{u}$ the experiments can generate, since the agent's
choice of action only depends on $\boldsymbol{u}$.

This motivates our first order, the column space order. Recall that
an experiment $\mathcal{E}$ is an $N\times M$ row stochastic matrix.
Its column space is $\operatorname{Col}\mathcal{E}:=\left\{ \mathcal{E}\boldsymbol{v}:\boldsymbol{v}\in\mathbb{R}^{M}\right\} $,
the set of all state-dependent utilities that can be generated with
experiment $\mathcal{E}$ using any contract.\footnote{The utility function $u$ is assumed to be unbounded. Therefore, any
$\boldsymbol{v}\in\mathbb{R}^{M}$ can be generated with some contract
$\boldsymbol{t}\in\mathbb{R}^{M}$.} 

The column space order is defined by set inclusion: say that $\mathcal{E}$
dominates $\mathcal{E}'$ in the column space order, denoted $\mathcal{E}\geq_{\text{Col}}\mathcal{E}'$,
if $\operatorname{Col}\mathcal{E}\supseteq\operatorname{Col}\mathcal{E}'$.
The column space order first appears in \citet{azrieli2022elicitability}
in the context of elicitation. Dominance in the column space order
means the principal has a larger set of feasible state-dependent utilities
to choose from. 

The column space order characterizes the comparisons of implementability:
$\mathcal{E}\geq_{\text{Col}}\mathcal{E}'$ if and only if any action
implementable under $\mathcal{E}'$ is also implementable under $\mathcal{E}$.
In Appendix \ref{appsec:non-concave}, I allow both the principal
and the agent to have general preferences over money, and show that
the column space order also characterizes the implementability comparisons
for a more general class of problems. 

The sufficiency part of the result is straightforward via a larger-feasible-set
argument. Given $(A,u,C)$, if $a_{0}$ is implementable under $\mathcal{E}'$
-- say, via a contract generating state-dependent utility $\boldsymbol{u}$
-- then it is also implementable under $\mathcal{E}$ since the same
$\boldsymbol{u}$ is feasible under $\mathcal{E}$ as well.

For necessity, I argue by contrapositive: if the column space inclusion
fails, I construct a moral hazard problem that violates the implementability
comparisons. The construction uses a problem in the flexible class
$\mathcal{P}_{F}$ introduced by \citet{georgiadis2024flexible},
where the first-order condition suffices for the agent's problem.
This allows precise control over the state-dependent utility required
for implementation. I pick a cost function and a target action so
that the required state-dependent utility for implementation is only
feasible under $\mathcal{E}'$, which completes the proof. Details
appear in Appendix \ref{appsec:proof-appendix}.

In Example \ref{exa:main-example}, all three experiments have the
same column space that spans the full $\mathbb{R}^{2}$. The firm
can therefore implement any action under each experiment, though,
as we will see, at different costs.\footnote{In fact, with a binary state space, any informative experiment can
implement all feasible actions with a finite marginal cost. The only
exception is an uninformative experiment, which cannot implement any
action with a strictly positive cost. If we move to larger state spaces,
the column space order becomes nontrivial, and informative experiments
may differ in which actions they can implement.}

The column space order compares implementability but not cost in general.
The reason is that it concerns the feasibility of state-dependent
utilities, rather than the scale or dispersion of transfers required
to generate them. One exception is the class $\mathcal{P}_{1}$ of
risk neutral problems without limited liability. In this case, every
implementable action can be implemented at the first-best cost, so
cost comparisons reduce to implementability comparisons.

Besides moral hazard problems, one can show that the column space
order admits two equivalent characterizations. In terms of matrix
factorization, $\mathcal{E}\geq_{\text{Col}}\mathcal{E}'$ requires
$\mathcal{E}'=\mathcal{E}G$ for some unconstrained matrix $G$, unlike
Blackwell where $G$ has to be a garbling. 

In terms of posterior beliefs, $\mathcal{E}\geq_{\text{Col}}\mathcal{E}'$
requires that the affine hull of the posteriors induced by $\mathcal{E}'$
lies inside that induced by $\mathcal{E}$. The column space order
thus compares the qualitative information contained in experiments:
the directions in which the posteriors are moved, not how far or how
often. This complements the earlier intuition that cost comparisons
usually do not hold for column space dominance. 

The next theorem summarizes the results on the column space order.

\begin{thm}
\label{thm:column-space} For any experiments $\mathcal{E}\in E^{M}$
and $\mathcal{E}'\in E^{M'}$, the following are equivalent:
\begin{enumerate}
\item [$(1)$] $\mathcal{E}\geq_{\text{Col}}\mathcal{E}'$,
\item [$(2)$] $\mathcal{I}(\mathcal{E}';A,u,C)\subseteq\mathcal{I}(\mathcal{E};A,u,C)$
for any $A\in\mathcal{A}$,$u\in\mathcal{U},C\in\mathcal{C}$.
\item [$(3)$] $\kappa(\mathcal{E};P)\leq\kappa(\mathcal{E}';P)$ for any
$P\in\mathcal{P}_{1}$,
\item [$(4)$] $\mathcal{E}'=\mathcal{E}G$ for some matrix $G$,
\item [$(5)$] $\operatorname{Aff}\operatorname{Supp}\pi\left(\mathcal{E};\boldsymbol{\nu}\right)\supseteq\operatorname{Aff}\operatorname{Supp}\pi\left(\mathcal{E}';\boldsymbol{\nu}\right)$
for any prior $\boldsymbol{\nu}\in\Delta\Omega$.\footnote{The affine span of a set $X\subseteq\mathbb{R}^{N}$ is defined as
$\operatorname{Aff}X=\left\{ \sum w_{i}\boldsymbol{x}_{i}:w_{i}\in\mathbb{R},\sum w_{i}=1,\boldsymbol{x}_{i}\in A,\forall i\right\} $,
which is the set of all affine combinations of elements in $X$. Condition
$(5)$ is equivalent to requiring the same inclusion to hold for some
interior prior $\boldsymbol{\nu}$. In Appendix \ref{appsec:proof-appendix},
I provide the proof for both statements.}
\end{enumerate}
\end{thm}

\subsection{Cost under Risk Neutrality \label{subsec:conic-span}}

Next, I turn to the comparisons of cost with a risk neutral agent
protected by limited liability -- the class $\mathcal{P}_{2}$.

Cost comparisons in $\mathcal{P}_{2}$ are immediate from the state-dependent
utility framework. The relevant object here is the conic span of an
experiment $\mathcal{E}$, defined as $\operatorname{Cone}\mathcal{E}:=\left\{ \mathcal{E}\boldsymbol{v}:\boldsymbol{v}\in\mathbb{R}^{M},\boldsymbol{v}\geq0\right\} $.
For any utility $u\in\mathcal{U}$, $\operatorname{Cone}\mathcal{E}$
is the set of all state-dependent utilities that can be generated
with experiment $\mathcal{E}$ under limited liability.\footnote{Recall that $u$ is strictly increasing and is normalized so that
$u(0)=0$. Therefore any $\boldsymbol{v}\geq0$ can be generated with
some contract $\boldsymbol{t}\geq0$.} 

The conic span order is again defined by set inclusion: say that $\mathcal{E}$
dominates $\mathcal{E}'$ in the conic span order, denoted $\mathcal{E}\geq_{\text{Cone}}\mathcal{E}'$,
if $\operatorname{Cone}\mathcal{E}\supseteq\operatorname{Cone}\mathcal{E}'$.
Dominance in the conic span order means the principal has a larger
set of feasible state-dependent utilities under limited liability.

The conic span order characterizes the cost comparisons in $\mathcal{P}_{2}$.\footnote{The conic span order also characterizes the cost comparisons in all
moral hazard problems with risk neutrality and an ex post budget constraint
that bounds the payments from above. See Appendix \ref{appsec:proof-appendix}.} The sufficiency also follows from a larger-choice-set argument: a
larger conic span provides a larger set of feasible state-dependent
utilities to minimize this cost. Its necessity is again proved by
constructing examples in the flexible class $\mathcal{P}_{F}$. Details
appear in Appendix \ref{appsec:proof-appendix}. 

To build intuition, let us now compare experiments $\mathcal{E}_{1}$
and $\mathcal{E}_{2}$ in Example \ref{exa:main-example}. $\mathcal{E}_{1}$
is always informative, while $\mathcal{E}_{2}$ is more accurate when
it is informative. They are not Blackwell ranked, but the cost under
$\mathcal{E}_{2}$ is lower than that under $\mathcal{E}_{1}$ for
any problem in $\mathcal{P}_{2}$. Consider implementing $a=1$. The
firm optimally provides no insurance to the worker due to risk neutrality.
Instead, all incentives are concentrated on a single bonus following
the $R$ signal since it has the largest likelihood ratio $\mathcal{E}(R\mid\omega_{r})/\mathcal{E}(R\mid\omega_{\ell})$.\footnote{That the principal pays the agent only after a single signal is a
feature of the binary state space. This holds for $\mathcal{E}_{2}$
even if the $N$ signal favors $\omega_{r}$ slightly. Under a larger
state space, the principal may optimally pay bonuses for several signals,
but the signals with strictly positive payments must be non-redundant.} The incentive provided $u_{r}-u_{\ell}$ is given by the difference
in signal $R$'s conditional probabilities across states, multiplied
by the payment, while the cost is given by signal $R$'s probability
in state $\omega_{r}$, multiplied by the payment. This means the
cost per unit of incentive is inversely related to the likelihood
ratio of $R$. Under $\mathcal{E}_{1}$, $R$'s likelihood ratio is
$0.7/0.3=2.33$, compared to $0.8/0.2=4$ under $\mathcal{E}_{2}$.
Therefore, to provide the same incentive, expected cost is lower under
$\mathcal{E}_{2}$. We can verify this by the feasibility of state-dependent
utilities: the optimal contract under $\mathcal{E}_{2}$ generates
$\left(u_{\ell},u_{r}\right)=(2/3,8/3)$ which lies in $\operatorname{Cone}\mathcal{E}$
but not $\operatorname{Cone}\mathcal{E}'$.\footnote{In both cases, the optimal contracts set $t_{L}=0$ so that limited
liability binds, and pick $t_{R}$ so that the marginal benefit of
effort equals its marginal cost at $a=1$, that is, $u_{r}-u_{\ell}=2$.
Under $\mathcal{E}_{1}$, $u_{r}-u_{\ell}=0.4t_{R}$, so $t_{R}=5$,
generating $\left(u_{\ell},u_{r}\right)=(1.5,3.5)$ at expected cost
$3.5$. Under $\mathcal{E}_{2}$, the principal optimally sets $t_{N}=0$,
since it merely raises $u_{\ell}$ and $u_{r}$ by the same amount
and provides no incentive. Here, $u_{r}-u_{\ell}=0.3t_{R}$, so $t_{R}=20/3$,
generating $\left(u_{\ell},u_{r}\right)=(2/3,8/3)$ at expected cost
$8/3$. The utility $\left(u_{\ell},u_{r}\right)=(2/3,8/3)$ cannot
be generated under $\mathcal{E}_{1}$ with non-negative payments,
since the only contract under $\mathcal{E}_{1}$ generating that is
$t_{L}=-5/6$ and $t_{R}=25/6$, violating limited liability.} Panel A.1 of Figure \ref{fig:conic-span-zonotope} illustrates the
conic spans and the corresponding utilities.

The conic span order generalizes this insight. The extreme rays of
$\operatorname{Cone}\mathcal{E}$ correspond to signals with extremal
likelihood ratios -- extreme points of the convex hull of the induced
likelihood ratios. The principal pays the agent only following these
extremal signals due to risk neutrality. A larger conic span means
the extremal likelihood ratios are further away from one, and thus
delivers a lower cost.\footnote{When there are more than two states, the likelihood ratio is a vector.
A larger conic span means the convex hull of the likelihood ratios
extends further outward in every direction, thus reducing the cost. } 

In Appendix \ref{appsec:non-concave}, I allow for general preferences
over money, and show that the conic span order characterizes cost
comparisons when the principal and the agent have the same risk attitude
towards money. 

Similar equivalent characterizations can be provided. In terms of
linear algebra, $\mathcal{E}\geq_{\text{Cone}}\mathcal{E}'$ requires
$\mathcal{E}'=\mathcal{E}G$ for some $G\geq0$. In terms of posterior
beliefs, $\mathcal{E}\geq_{\text{Cone}}\mathcal{E}'$ means the convex
hull of the support of the posteriors induced by $\mathcal{E}'$ lies
inside that induced by $\mathcal{E}$, formalizing the likelihood
ratio interpretation above. 

The next theorem summarizes the results on the conic span order. 

\begin{thm}
\label{thm:conic-span} For any experiments $\mathcal{E}\in E^{M}$
and $\mathcal{E}'\in E^{M'}$, the following are equivalent:
\begin{enumerate}
\item [$(1)$] $\mathcal{E}\geq_{\text{Cone}}\mathcal{E}'$,
\item [$(2)$] $\kappa(\mathcal{E};P)\leq\kappa(\mathcal{E}';P)$ for any
$P\in\mathcal{P}_{2}$,
\item [$(3)$]$\mathcal{E}'=\mathcal{E}G$ for some matrix $G\geq0$,
\item [$(4)$]$\operatorname{Co}\operatorname{Supp}\pi\left(\mathcal{E};\boldsymbol{\nu}\right)\supseteq\operatorname{Co}\operatorname{Supp}\pi\left(\mathcal{E}';\boldsymbol{\nu}\right)$
for any prior $\boldsymbol{\nu}\in\Delta\Omega$.\footnote{The convex hull of a set $X\subseteq\mathbb{R}^{N}$ is defined as
$\operatorname{Co}X=\left\{ \sum w_{i}\boldsymbol{x}_{i}:w_{i}\geq0,\sum w_{i}=1,\boldsymbol{x}_{i}\in A,\forall i\right\} $,
which is the set of all convex combinations of elements in $X$. Condition
$(4)$ is equivalent to requiring the same inclusion to hold for some
interior prior $\boldsymbol{\nu}$. In Appendix \ref{appsec:proof-appendix},
I provide the proof for both statements.}
\end{enumerate}
\end{thm}
\begin{figure}[th]
\begin{centering}
\caption{Comparisons of Experiments in Example \ref{exa:main-example} \label{fig:conic-span-zonotope}}
\par\end{centering}
\begin{centering}
\bigskip{}
\par\end{centering}
\begin{centering}
{\small{}%
\begin{minipage}[t]{0.5\columnwidth}%
\begin{center}
{\small Panel A: Comparisons of }$\mathcal{E}_{1}$ and $\mathcal{E}_{2}$
\par\end{center}
\begin{center}
\begin{minipage}[t]{0.5\columnwidth}%
\begin{center}
{\footnotesize A.1 Conic Spans}{\footnotesize\par}
\par\end{center}
\begin{center}
\includegraphics[width=1\columnwidth]{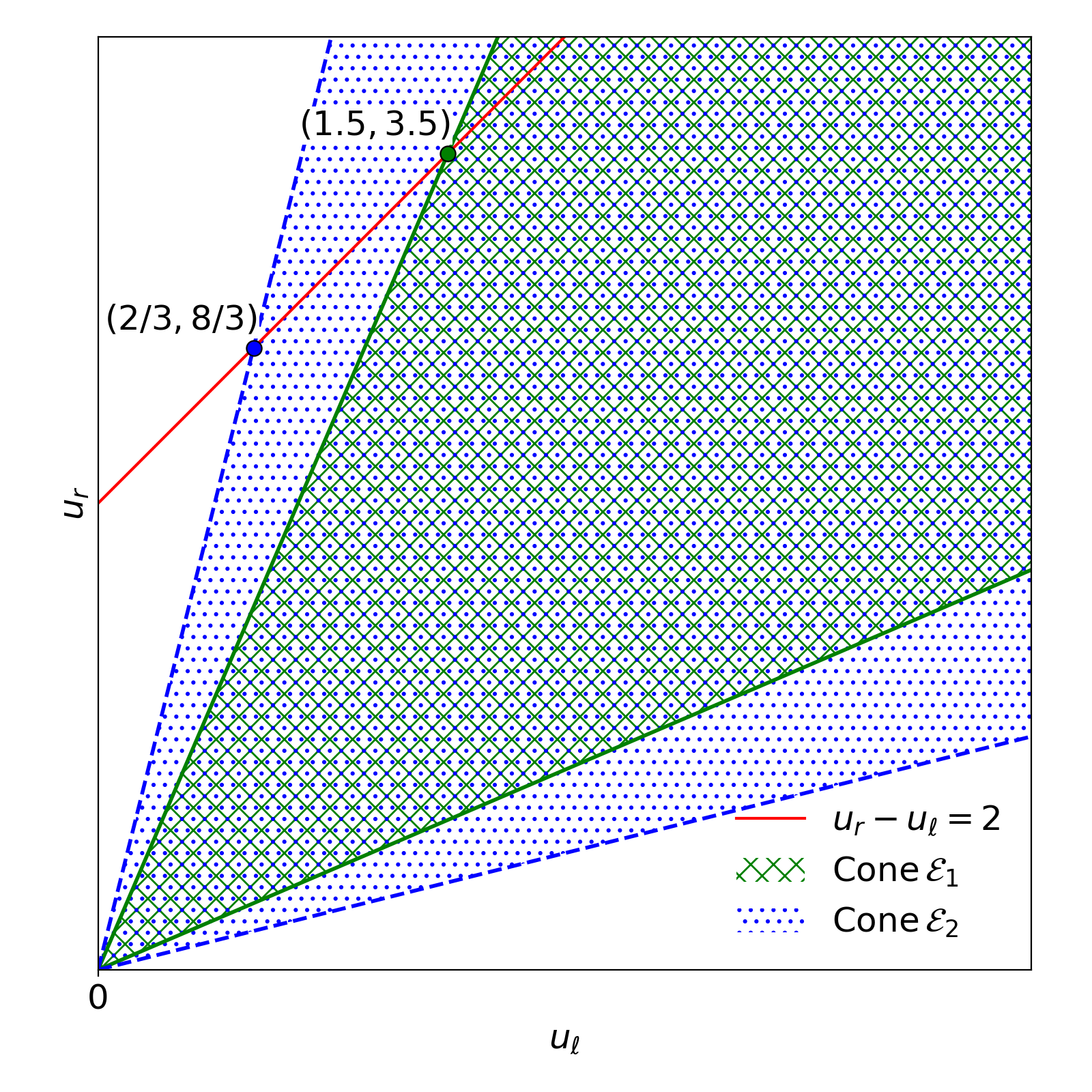}
\par\end{center}%
\end{minipage}\,%
\begin{minipage}[t]{0.5\columnwidth}%
\begin{center}
{\footnotesize A.2 Zonotopes}{\footnotesize\par}
\par\end{center}
\begin{center}
\includegraphics[width=1\columnwidth]{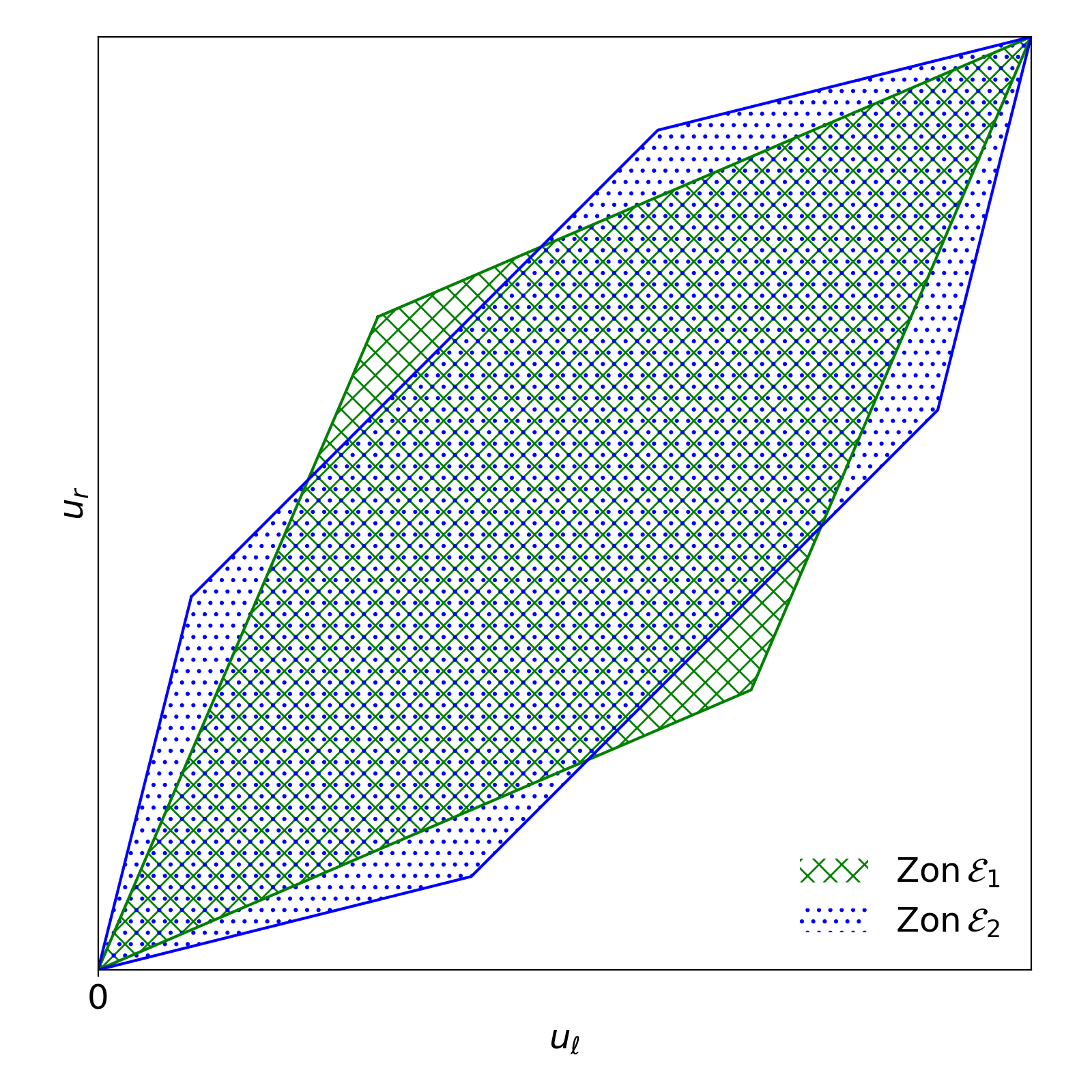}
\par\end{center}%
\end{minipage}
\par\end{center}%
\end{minipage}}\enspace{}%
\begin{minipage}[t]{0.5\columnwidth}%
\begin{center}
{\small Panel B: Comparisons of }$\mathcal{E}_{2}$ and $\mathcal{E}_{3}$
\par\end{center}
\begin{center}
\begin{minipage}[t]{0.5\columnwidth}%
\begin{center}
{\footnotesize B.1 Conic Spans}{\footnotesize\par}
\par\end{center}
\begin{center}
\includegraphics[width=1\columnwidth]{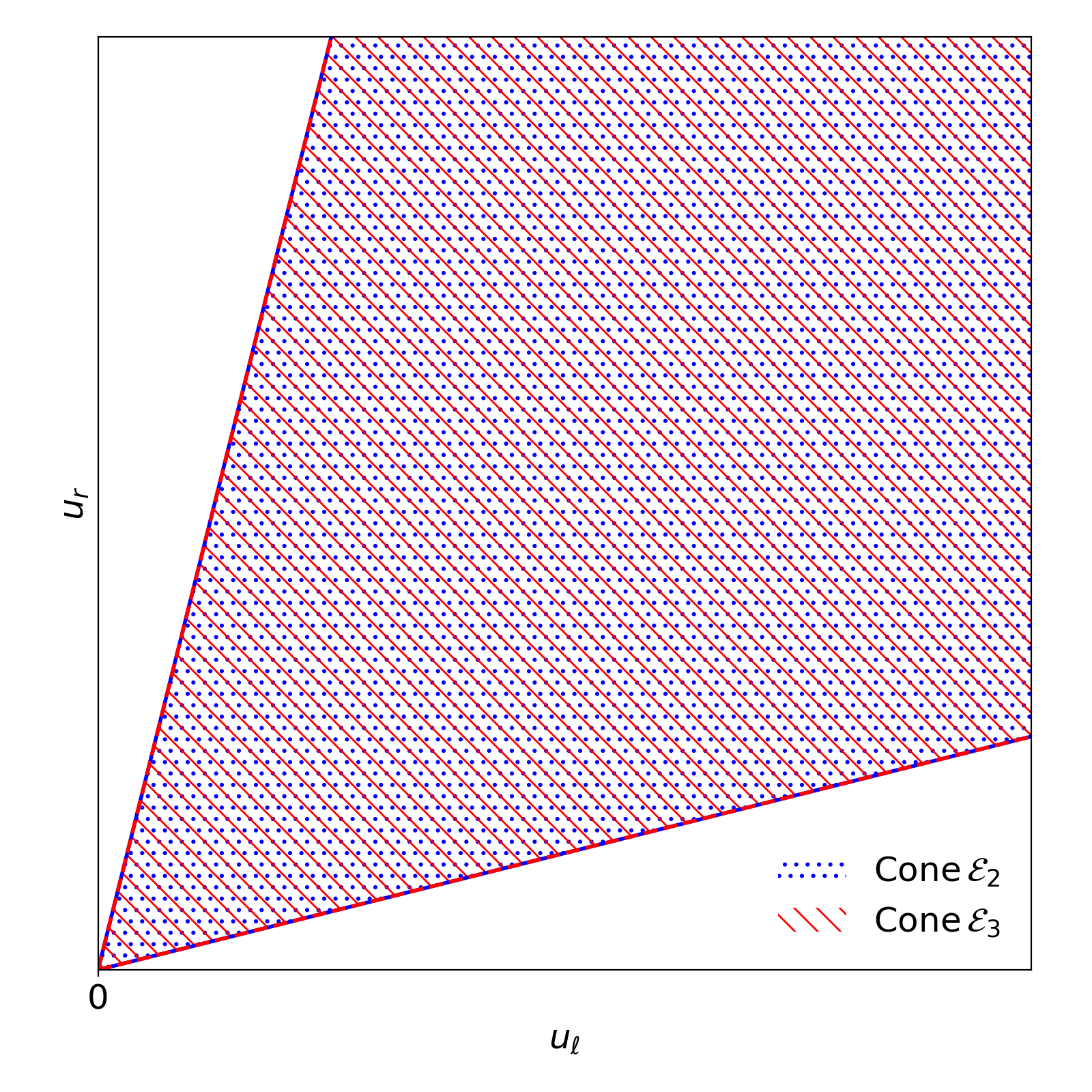}
\par\end{center}%
\end{minipage}\,%
\begin{minipage}[t]{0.5\columnwidth}%
\begin{center}
{\footnotesize B.2 Zonotopes}{\footnotesize\par}
\par\end{center}
\begin{center}
\includegraphics[width=1\columnwidth]{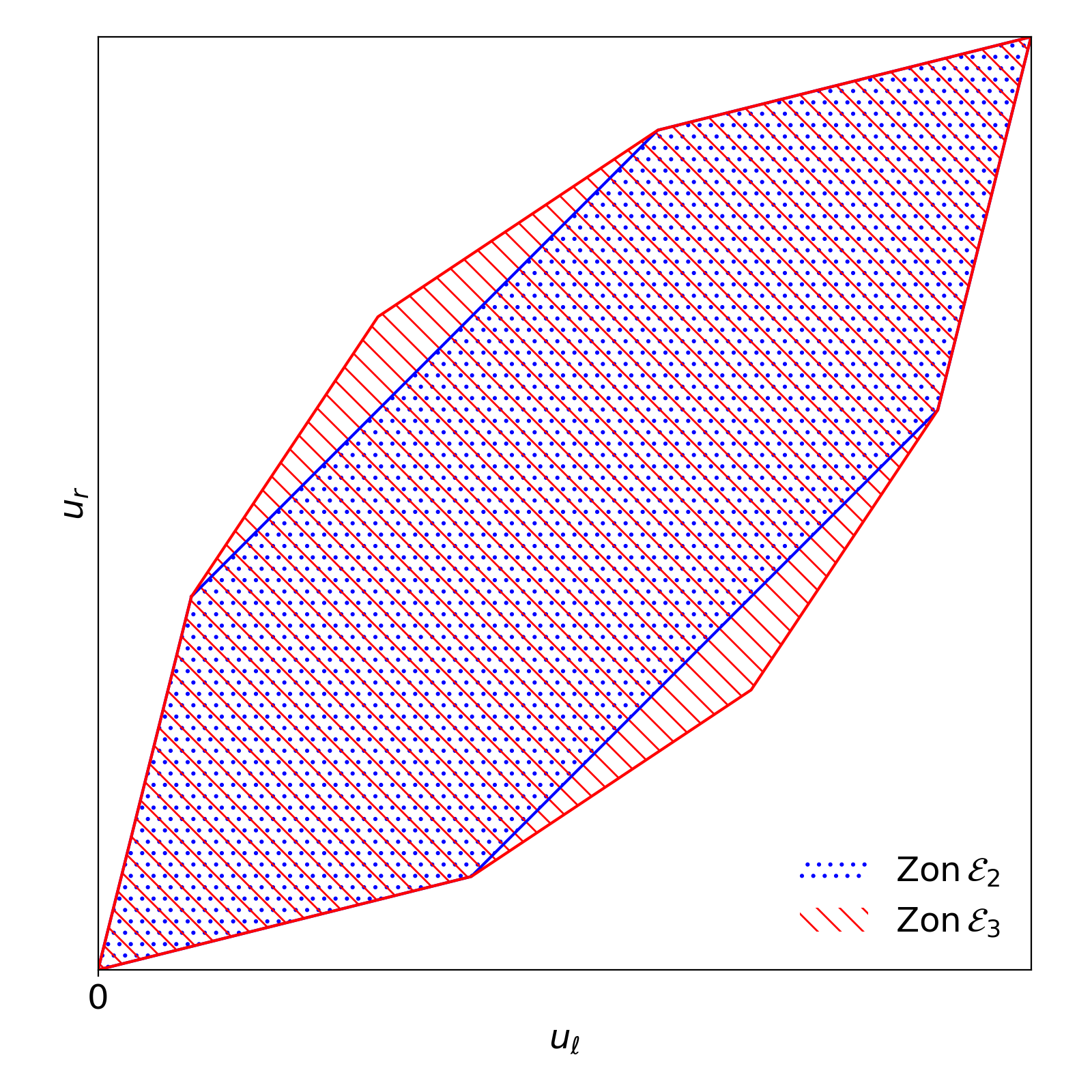}
\par\end{center}%
\end{minipage}
\par\end{center}%
\end{minipage}
\par\end{centering}
\centering{}%
\begin{minipage}[t]{0.9\columnwidth}%
\begin{singlespace}
{\scriptsize Notes: This figure plots the conic spans and the zonotopes
for experiments in Example \ref{exa:main-example}. Since the state
space is binary, the conic span and the zonotope are subsets of $\mathbb{R}^{2}$,
representing the agent's expected utility $u_{\ell}$ in states $\omega_{\ell}$
and $u_{r}$ in state $\omega_{r}$. Panel A.1 also plots the utilities
such that $u_{r}-u_{\ell}=2$, and marks the state-dependent utilities
from the optimal contract under $\mathcal{E}_{1}$ and $\mathcal{E}_{2}$
when the agent is risk neutral. }{\scriptsize\par}
\end{singlespace}

\end{minipage}
\end{figure}

\subsection{Cost under General Risk Preferences \label{subsec:zonotope}}

I now turn to the general case and study cost comparisons in class
$\mathcal{P}$, the set of all moral hazard problems. It suffices
to focus on the case with a risk averse agent, since non-concave utilities
can be reduced to concave utilities with random payments. The details
appear in Appendix \ref{appsec:non-concave}

The conic span order is no longer sufficient because it fails to account
for how risk aversion affects the cost of incentives. It only compares
the extremal signals, without referring to how often they are induced.
Under risk aversion, for example, paying a large bonus with a small
probability is costly. 

In Example \ref{exa:main-example}, suppose the worker's utility becomes
$u(t)=\sqrt{t}$. It now costs more to use $\mathcal{E}_{2}$ than
$\mathcal{E}_{1}$.\footnote{Under $\mathcal{E}_{1}$, the firm optimally sets $t_{L}=0$ and $t_{R}=25$,
yielding a cost of $17.5$. Under $\mathcal{E}_{2}$, the firm optimally
sets $t_{L}=t_{N}=0$ and $t_{R}=400/9$, which comes at a cost of
$160/9>17.5$. \label{fn:compute-E2-risk-averse}} In fact, neither experiment guarantees a lower cost despite conic
span dominance. To restore the cost comparisons, we need a finer
order to compare the sets of feasible state-dependent utilities for
any concave utility function. 

The key observation is that risk aversion effectively bounds utilities
from both above and below. Arbitrary concavity creates an upper bound,
since the utility may satiate after some point, and the participation
constraint creates a lower bound, since the utility from negative
payments can be so negative that the agent no longer participates.
This suggests comparing experiments based on feasible state-dependent
utilities from bounded payments. This is formalized by the zonotope
order. 

The zonotope of an experiment $\mathcal{E}$ is defined as $\operatorname{Zon}\mathcal{E}:=\left\{ \mathcal{E}\boldsymbol{v}:\boldsymbol{v}\in\mathbb{R}^{M},0\leq\boldsymbol{v}\leq1\right\} $.
For any utility $u\in\mathcal{U}$, $\operatorname{Zon}\mathcal{E}$
is the set of all state-dependent utilities that can be generated
with experiment $\mathcal{E}$ using any non-negative and bounded
payments, up to a multiplicative constant.\footnote{For any upper bound $B>0$, the set of state-dependent utilities from
any $0\leq\boldsymbol{t}\leq B$ is given by $\left\{ \mathcal{E}\boldsymbol{v}:\boldsymbol{v}\in\mathbb{R}^{M},0\leq\boldsymbol{v}\leq u(B)\right\} =u(B)\cdot\operatorname{Zon}\mathcal{E}$,
which is the zonotope scaled by a constant $u(B)$. } This is the full set of state-dependent utilities the principal can
choose from when she is constrained both by limited liability and
ex post budget. 

The zonotope order is defined as the inclusion of the zonotope. Formally,
say that $\mathcal{E}$ dominates $\mathcal{E}'$ in the zonotope
order, denoted $\mathcal{E}\geq_{\text{Zon}}\mathcal{E}'$, if $\operatorname{Zon}\mathcal{E}\supseteq\operatorname{Zon}\mathcal{E}'$.
It follows immediately that the zonotope order characterizes cost
comparisons in $\mathcal{P}_{3}$, risk neutral problems with limited
liability and ex post budget constraints. This is again because a
larger feasible set is always better for the principal.

The zonotope order also characterizes cost comparisons in $\mathcal{P}$,
the full class of moral hazard problems. The result is robust: it
is unchanged if we restrict attention to problems with limited liability
or ex post budget constraints or both; it is also unchanged if we
allow the principal and the agent to have general preferences over
money. Together with the previous result, this shows that, to compare
experiments in $\mathcal{P}$, it suffices to focus on problems in
$\mathcal{P}_{3}$.

I illustrate the result with Example \ref{exa:main-example}. First,
we know that $\mathcal{E}_{1}$ and $\mathcal{E}_{2}$ are zonotope
non-comparable because neither guarantees lower costs. Panel A.2 of
Figure \ref{fig:conic-span-zonotope} shows that neither zonotope
contains the other. 

Next, $\mathcal{E}_{3}$ dominates $\mathcal{E}_{2}$ in the zonotope
order, as illustrated by Panel B.2 of Figure \ref{fig:conic-span-zonotope}.\footnote{In fact, $\mathcal{E}_{3}$ dominates $\mathcal{E}_{2}$ in the Blackwell
order. When the state space is binary, the zonotope order and the
Blackwell order coincide. In Appendix \ref{appsec:example}, I provide
an example with three states where two experiments are ranked in the
zonotope order but non-comparable in the Blackwell order.} We can verify this by solving for the optimal contract to implement
$a=1$ with utility $u(t)=\sqrt{t}$. Under $\mathcal{E}_{2}$, the
cost is $160/9$ from the earlier computation. Under $\mathcal{E}_{3}$,
the firm should only pay the agent for signals $wR$ and $sR$. Let
$t_{wR}$ and $t_{sR}$ be the payments following the two signals.
The firm's problem is to minimize its cost, $0.3t_{wR}+0.4t_{sR}$,
subject to the incentive constraint that $0.1\sqrt{t_{wR}}+0.3\sqrt{t_{sR}}=2$.
This solves to $t_{sR}\approx33.7$ and $t_{wR}\approx6.7$ with a
cost of $15.5$, which is smaller than the cost under $\mathcal{E}_{2}$.
The cost saving comes from the informative weak signal of $\mathcal{E}_{3}$.
This allows the firm to insure the worker with a small payment $t_{wR}$
if the signal is weak, which is not possible under $\mathcal{E}_{2}$
because signal $N$ is unable to provide incentive.

The intuition is again a larger-choice-set argument, the recurring
theme across all three orders: a larger zonotope provides the principal
with a larger choice set of feasible state-dependent utilities, weakly
lowering cost. The proof, however, is more involved because a general
risk preference breaks the connection between the agent's utility
and the principal's cost. I sketch it after Theorem \ref{thm:zonotope}.

As for the equivalent characterizations, algebraically, the zonotope
order requires the sum of any subset of columns of $\mathcal{E}'$
to lie in $\operatorname{Zon}\mathcal{E}$, that is, $\mathcal{E}'\boldsymbol{v}\in\operatorname{Zon}\mathcal{E}$
for any binary vector $\boldsymbol{v}\in\{0,1\}^{M'}$. In matrix
form, this can be summarized as $\mathcal{E}'B=\mathcal{E}D$ for
some $0\leq D\leq1$, where $B\in\mathbb{R}^{M'\times2^{M'}}$ is
the matrix whose columns consist of all binary vectors in $\left\{ 0,1\right\} ^{M'}$.
To see this, $\mathcal{E}'B$ lists all partial sums of columns of
$\mathcal{E}'$, which are the extreme points of $\operatorname{Zon}\mathcal{E}'$.
Since the zonotope is a convex set, zonotope inclusion requires that
all these extreme points lie in $\operatorname{Zon}\mathcal{E}$.
This is a stronger condition than $\mathcal{E}'=\mathcal{E}G$ for
some $0\leq G\leq1$ because the latter only asks every column of
$\mathcal{E}'$ to lie in $\operatorname{Zon}\mathcal{E}$.

In terms of posterior beliefs, the zonotope order requires dominance
in linear convex order between the induced posterior distributions,
in contrast to the convex order required by Blackwell. Although this
interpretation is not central to my approach, I include the definition
for completeness. For cumulative distribution functions $F,G:\mathbb{R}^{N}\to[0,1]$,
$F$ dominates $G$ in the linear convex order, denoted $F\geq_{\text{lcx}}G$,
if for any convex function $\phi:\mathbb{R}\to\mathbb{R}$ and any
vector $\boldsymbol{\beta}\in\mathbb{R}^{N}$, $\mathbb{E}_{\boldsymbol{x}\sim F}\left[\phi(\boldsymbol{\beta}\cdot\boldsymbol{x})\right]\geq\mathbb{E}_{\boldsymbol{x}\sim G}\left[\phi(\boldsymbol{\beta}\cdot\boldsymbol{x})\right]$,
where $\boldsymbol{\beta}\cdot\boldsymbol{x}$ represents the dot
product.

The next theorem summarizes the results on the zonotope order.  

\begin{thm}
\label{thm:zonotope} For any experiments $\mathcal{E}\in E^{M}$
and $\mathcal{E}'\in E^{M'}$, the following are equivalent:
\begin{enumerate}
\item [$(1)$] $\mathcal{E}\geq_{\text{Zon}}\mathcal{E}'$,
\item [$(2)$] $\kappa(\mathcal{E};P)\leq\kappa(\mathcal{E}';P)$ for any
$P\in\mathcal{P}_{3}$,
\item [$(3)$] $\kappa(\mathcal{E};P)\leq\kappa(\mathcal{E}';P)$ for any
$P\in\mathcal{P}$,
\item [$(4)$] The sum of any subset of columns of $\mathcal{E}'$ lies
in $\operatorname{Zon}\mathcal{E}$, or formally, $\mathcal{E}'\boldsymbol{v}\in\operatorname{Zon}\mathcal{E}$
for any $\boldsymbol{v}\in\{0,1\}^{M'}$; in matrix form, this is
$\mathcal{E}'B=\mathcal{E}D$ for some matrix $D\in\mathbb{R}^{M\times2^{M'}}$
with $0\leq D\leq1$, where $B:=\begin{bmatrix}\boldsymbol{v}^{(1)} & \boldsymbol{v}^{(2)} & \dots & \boldsymbol{v}^{(2^{M'})}\end{bmatrix}\in\mathbb{R}^{M'\times2^{M'}}$
with $\{\boldsymbol{v}^{(1)},\dots,\boldsymbol{v}^{(2^{M'})}\}=\{0,1\}^{M'}$
being the set of all binary vectors in $\mathbb{R}^{M'}$,
\item [$(5)$] $\pi\left(\mathcal{E};\boldsymbol{\nu}\right)\geq_{\text{lcx}}\pi\left(\mathcal{E}';\boldsymbol{\nu}\right)$
for any prior $\boldsymbol{\nu}\in\Delta\Omega$,
\end{enumerate}
\end{thm}

The equivalence between $(1)$, $(2)$, and $(3)$ is the most subtle
part of Theorem \ref{thm:zonotope}. The main difficulty is that,
without risk neutrality, lower utility no longer translates into lower
expected payments directly. 

The key idea of the proof is to reintroduce a link between utilities
and payments. The relevant object is the set of feasible utilities
given a budget. Formally, fix a reference state distribution $\boldsymbol{\mu}\in\Delta\Omega$,
a utility function $u\in\mathcal{U}$, and a budget $B$. Define 
\[
\mathcal{V}_{\boldsymbol{\mu},u,B}(\mathcal{E}):=\left\{ \mathcal{E}\boldsymbol{v}:\exists\boldsymbol{t}\in\mathbb{R}^{M}\text{ such that }\boldsymbol{v}=u(\boldsymbol{t}),\boldsymbol{\mu}\cdot\mathcal{E}\boldsymbol{t}\leq B\right\} ,
\]
that is, the set of state-dependent utilities that can be generated
under experiment $\mathcal{E}$ subject to an ex ante budget $B$.
For ease of exposition, ignore the possible limited liability and
ex post budget constraints for now. This set serves as an intermediate
object bridging zonotopes and moral hazard problems. 

The equivalence is established via an intermediate step $(6)$.
\begin{lem}
\label{lem:feasible-set-budget} The equivalent conditions in Theorem
\ref{thm:zonotope} are also equivalent to
\begin{enumerate}
\item [$(6)$] $\mathcal{V}_{\boldsymbol{\mu},u,B}(\mathcal{E})\supseteq\mathcal{V}_{\boldsymbol{\mu},u,B}(\mathcal{E}')$
for any $\boldsymbol{\mu}\in\Delta\Omega$, $u\in\mathcal{U}$, and
$B>0$.
\end{enumerate}
\end{lem}
Condition $(6)$ says that given any ex ante budget, $\mathcal{E}$
generates a larger set of utilities. Equivalently, it says that the
expected cost to generate any state-dependent utility $\boldsymbol{u}$
is lower under $\mathcal{E}$ than under $\mathcal{E}'$, regardless
of the state distribution used to evaluate the expectation. To see
this, if $\boldsymbol{u}$ is feasible under $\mathcal{E}'$ at budget
$B$, it is also feasible under $\mathcal{E}$ at the same budget,
so the minimum budget needed to generate $\boldsymbol{u}$ must be
weakly lower under $\mathcal{E}$. 

The key step of the proof is $(1)\Rightarrow(6)$. This step circumvents
the main difficulty -- lower utility does not imply lower cost --
via convex analysis. Since both $\operatorname{Zon}\mathcal{E}$ and
$\mathcal{V}_{\boldsymbol{\mu},u,B}(\mathcal{E})$ are convex, their
inclusions can each be characterized by a separating hyperplane condition.
The condition characterizing zonotope inclusion turns out to imply
the one characterizing $\mathcal{V}_{\boldsymbol{\mu},u,B}$-inclusion.

The remaining implications are straightforward. For $(6)\Rightarrow(3)$:
condition $(6)$ directly implies cost dominance, since any $\boldsymbol{u}$
generated under $\mathcal{E}'$ can be generated under $\mathcal{E}$
at a lower cost while preserving incentives. $(3)\Rightarrow(2)$
follows by choosing a specific utility function that mimics ex post
budget and limited liability.\footnote{That is, take a sequence $u_{i}\in\mathcal{U}$ that converges pointwise
to $u_{0}$ where $u_{0}(t)=\min\{t,B\}$ when $t\geq0$ and $u_{0}(t)=-\infty$
when $t<0$, which satiates after receiving a payment of $B$ to mimic
the ex post budget constraint, and drops to $-\infty$ for negative
payment to mimic the limited liability constraint. } $(2)\Rightarrow(1)$ again uses a constructive proof with a flexible
problem in $\mathcal{P}_{F}$. The details are provided in Appendix
\ref{appsec:proof-appendix}. 

\subsection{Relations between the Orders \label{subsec:relations}}

I now discuss the relations between the orders. First of all, they
are nested.
\begin{prop}
$\geq_{\text{Col}}\:\supseteq\:\geq_{\text{Cone}}\:\supseteq\:\geq_{\text{Zon}}\:\supseteq\:\geq_{\text{B}}$. 
\end{prop}
The inclusion is the easiest to see using the matrix factorization
condition. The inclusions can be strict, as shown by Example \ref{exa:main-example}.
All three experiments have the same column space as they all have
full row rank. But Panel A of Figure \ref{fig:conic-span-zonotope}
shows that $\mathcal{E}_{1}\not\geq_{\text{Cone}}\mathcal{E}_{2}$
and $\mathcal{E}_{2}\geq_{\text{Cone}}\mathcal{E}_{1}$. From the
figure, we also see that $\mathcal{E}_{1}$ and $\mathcal{E}_{2}$
are not ranked in the zonotope order. On the other hand, $\mathcal{E}_{2}$
and $\mathcal{E}_{3}$ have the same conic span, but $\mathcal{E}_{2}\not\geq_{\text{Zon}}\mathcal{E}_{3}$,
as depicted in Panel B of Figure \ref{fig:conic-span-zonotope}. 

The zonotope order coincides with Blackwell when the state space is
binary. This is a known result due to \citet{blackwell1953equivalent}
and \citet{bertschinger2014blackwell}. When the state space is larger,
the inclusion between the zonotope and the Blackwell orders is strict.
An example is first given by \citet{bertschinger2014blackwell} and
reproduced in Appendix \ref{appsec:example}.
\begin{prop}
\label{prop:binary-states} When $N=2$, $\geq_{\text{Zon}}\:=\:\geq_{\text{B}}$.
When $N>2$, $\geq_{\text{Zon}}\:\supsetneq\:\geq_{\text{B}}$.
\end{prop}
The full rank property may also make the orders coincide. The zonotope
order coincides with the Blackwell order under full rank, and the
conic span order coincides with the Blackwell order under full column
rank. Appendix \ref{appsec:example} provides an example to show that
the full rank condition is necessary. 
\begin{prop}
\label{prop:full-rank} When $\mathcal{E}$ has full row rank or full
column rank, $\mathcal{E}\geq_{\text{Zon}}\mathcal{E}'$ if and only
if $\mathcal{E}\geq_{\text{B}}\mathcal{E}'$. When $\mathcal{E}$
has full column rank, $\mathcal{E}\geq_{\text{Cone}}\mathcal{E}'$
if and only if $\mathcal{E}\geq_{\text{B}}\mathcal{E}'$.
\end{prop}
Lastly, I briefly discuss why the Blackwell order is not the correct
order to compare experiments in moral hazard problems. One way to
define the Blackwell order is to use the inclusion of the set of feasible
joint distribution of state-action pairs. Specifically, for any action
space $A=\{a_{k}\}_{k=1}^{K+1}$, the feasible set of joint distributions
over $\Omega\times A$ induced by $\mathcal{E}$ can be described
by 
\[
\mathcal{S}_{K}(\mathcal{E}):=\left\{ \mathcal{E}\Pi:\Pi\in\mathbb{R}^{M\times K},\Pi\geq0,\Pi\boldsymbol{1}\leq\boldsymbol{1}\right\} ,
\]
where $\Pi_{m,k}$ is the probability of taking action $a_{k}$ after
realization $y_{m}$, and $1-\sum_{k=1}^{K}\Pi_{m,k}$ is the probability
of taking action $a_{K+1}$ after realization $y_{m}$. $\mathcal{E}\geq_{\text{B}}\mathcal{E}'$
if and only if $\bigcup_{K=1}^{\infty}\mathcal{S}_{K}(\mathcal{E})\supseteq\bigcup_{K=1}^{\infty}\mathcal{S}_{K}(\mathcal{E}')$.
This works well for general decision problems. 

In moral hazard problems, however, the principal only chooses a one-dimensional
object, the payment, following each realization. As a result, only
$\mathcal{S}_{1}(\mathcal{E})$ and its subsets matter, and $\mathcal{S}_{K}(\mathcal{E})$
for $K\geq2$ are irrelevant. In fact, $\mathcal{S}_{1}(\mathcal{E})=\operatorname{Zon}\mathcal{E}$,
and the zonotope order also characterizes the value of information
comparisons for all decision problems with two actions.

\section{Concluding Remarks \label{sec:extension-conclusion}}

This paper proposes a unified framework based on state-dependent utilities
to analyze both the classic and the flexible moral hazard problems.
I identify three nested geometric orders on information based on the
feasible sets of state-dependent utilities. The column space order
characterizes the comparisons of implementability. The conic span
order characterizes the comparisons of cost under risk neutrality
and limited liability. The zonotope order characterizes the comparisons
of cost under a general preference for money. 

I end by making a few remarks and discussing potential future work.
The first remark is on the finiteness of experiments adopted in the
main text. The state-dependent utility framework can be easily extended
to infinite experiments with a continuum of realizations. In this
case, an experiment $\mathcal{E}$ specifies for each state $\omega\in\Omega$
a distribution $\mathcal{E}(\cdot\mid\omega)$ over some realization
space $Y\subseteq\mathbb{R}$ with density $f_{\mathcal{E}}(\cdot\mid\omega)$.
A contract $\boldsymbol{t}:Y\to\mathbb{R}$ specifies the payment
to the agent following each realization. We can analogously define
the column space, conic span, and zonotope orders by the inclusion
of the feasible state-dependent utilities. One has to be careful about
the existence of a solution because of the \citet{mirrlees1999theory}
example where the principal can get arbitrarily close to the first
best but an optimal solution does not exist. The idea is to impose
an arbitrarily large punishment with an arbitrarily small probability.
For example, if a very positive (negative) $y\in Y$ is unlikely to
happen but is very informative about the agent's effort, the principal
can give the agent a huge bonus (or punishment) following $y$, approaching
the first best. We have to assume the likelihood ratios $f_{\mathcal{E}}(\cdot\mid\omega_{n})/f_{\mathcal{E}}(\cdot\mid\omega_{n'})$
are bounded for any pair of states $(\omega_{n},\omega_{n'})$ to
avoid this.

The second remark is on strict orders. Strict orders can also be defined
in the usual way. That is, say that $\mathcal{E}$ dominates $\mathcal{E}'$
strictly in the column space (or, conic span, zonotope) order, denoted
$\mathcal{E}>_{\text{Col}}\mathcal{E}'$ (or $\mathcal{E}>_{\text{Cone}}\mathcal{E}'$
, $\mathcal{E}>_{\text{Zon}}\mathcal{E}'$ ), if $\mathcal{E}$ dominates
$\mathcal{E}'$ but $\mathcal{E}'$ does not dominate $\mathcal{E}$
in the weak order. One can show that the dominance in the strict order
is characterized by the weak comparisons of costs as in Theorems \ref{thm:column-space}-\ref{thm:zonotope}
and the existence of a strict cost comparisons in some moral hazard
problem. 

As for potential future work, a follow-up question is about the demand
for information. Suppose the principal has to incur a cost to build
a more accurate performance measure, and a better performance measure
costs more. One can study how to optimally build the performance measure
to trade off the information cost and the agency cost. A prerequisite
is to understand what properties the information cost should have.
A natural requirement is Blackwell monotonicity, that is, more accurate
information costs more. Here, however, we may need a stronger notion
of monotonicity: the orders in this paper are finer than Blackwell,
and one would expect the cost of information to be larger if the performance
measure reduces agency cost across moral hazard problems. This suggests
studying information costs that are monotone with respect to the orders
developed in this paper.

\pagebreak{}

\bibliographystyle{format/aea/aea}
\bibliography{paper}

@inproceedings{blackwell1951comparison,
  title={Comparison of experiments},
  author={Blackwell, David},
  booktitle={Proceedings of the second Berkeley symposium on mathematical statistics and probability},
  volume={1},
  number={93-102},
  pages={26},
  year={1951}
}

@article{holmstrom1979moral,
  title={Moral hazard and observability},
  author={Holmstr{\"o}m, Bengt},
  journal={The Bell journal of economics},
  pages={74--91},
  year={1979},
  publisher={JSTOR}
}

@article{mirrlees1999theory,
  title={The theory of moral hazard and unobservable behaviour: Part I},
  author={Mirrlees, James A},
  journal={The Review of Economic Studies},
  volume={66},
  number={1},
  pages={3--21},
  year={1999},
  publisher={Wiley-Blackwell}
}

@article{innes1990limited,
  title={Limited liability and incentive contracting with ex-ante action choices},
  author={Innes, Robert D},
  journal={Journal of economic theory},
  volume={52},
  number={1},
  pages={45--67},
  year={1990},
  publisher={Elsevier}
}

@article{georgiadis2024flexible,
  title={Flexible moral hazard problems},
  author={Georgiadis, George and Ravid, Doron and Szentes, Bal{\'a}zs},
  journal={Econometrica},
  volume={92},
  number={2},
  pages={387--409},
  year={2024},
  publisher={Wiley Online Library}
}

@inproceedings{bertschinger2014blackwell,
  title={The Blackwell relation defines no lattice},
  author={Bertschinger, Nils and Rauh, Johannes},
  booktitle={2014 IEEE International Symposium on Information Theory},
  pages={2479--2483},
  year={2014},
  organization={IEEE}
}

@article{brooks2024comparisons,
  title={Comparisons of signals},
  author={Brooks, Benjamin and Frankel, Alexander and Kamenica, Emir},
  journal={American Economic Review},
  volume={114},
  number={9},
  pages={2981--3006},
  year={2024},
  publisher={American Economic Association 2014 Broadway, Suite 305, Nashville, TN 37203}
}

@article{blackwell1953equivalent,
  title={Equivalent comparisons of experiments},
  author={Blackwell, David},
  journal={The annals of mathematical statistics},
  pages={265--272},
  year={1953},
  publisher={JSTOR}
}

@article{mu2021blackwell,
  title={From Blackwell dominance in large samples to R{\'e}nyi divergences and back again},
  author={Mu, Xiaosheng and Pomatto, Luciano and Strack, Philipp and Tamuz, Omer},
  journal={Econometrica},
  volume={89},
  number={1},
  pages={475--506},
  year={2021},
  publisher={Wiley Online Library}
}

@article{kim1995efficiency,
  title={Efficiency of an information system in an agency model},
  author={Kim, Son Ku},
  journal={Econometrica: Journal of the Econometric Society},
  pages={89--102},
  year={1995},
  publisher={JSTOR}
}

@article{gjesdal1982information,
  title={Information and incentives: The agency information problem},
  author={Gjesdal, Fr{\o}ystein},
  journal={The Review of Economic Studies},
  volume={49},
  number={3},
  pages={373--390},
  year={1982},
  publisher={Wiley-Blackwell}
}

@article{chen2025experiments,
  title={Experiments in the Linear Convex Order},
  author={Chen, Kailin},
  journal={arXiv preprint arXiv:2502.06530},
  year={2025}
}

@article{lehmann2011comparing,
 ISSN = {00905364, 21688966},
 URL = {http://www.jstor.org/stable/2241739},
 author = {E. L. Lehmann},
 journal = {The Annals of Statistics},
 number = {2},
 pages = {521--533},
 publisher = {Institute of Mathematical Statistics},
 title = {Comparing Location Experiments},
 urldate = {2025-06-16},
 volume = {16},
 year = {1988}
}

@article{koshevoy1995multivariate,
  title={Multivariate lorenz majorization},
  author={Koshevoy, Gleb},
  journal={Social Choice and Welfare},
  pages={93--102},
  year={1995},
  publisher={JSTOR}
}

@article{koshevoy1996lorenz,
  title={The Lorenz zonoid of a multivariate distribution},
  author={Koshevoy, Gleb and Mosler, Karl},
  journal={Journal of the American Statistical Association},
  volume={91},
  number={434},
  pages={873--882},
  year={1996},
  publisher={Taylor \& Francis}
}

@book{mosler2002multivariate,
  title={Multivariate dispersion, central regions, and depth: the lift zonoid approach},
  author={Mosler, Karl},
  volume={165},
  year={2002},
  publisher={Springer Science \& Business Media}
}

@article{koshevoy1997lorenz,
  title={The Lorenz zonotope and multivariate majorizations},
  author={Koshevoy, Gleb},
  journal={Social Choice and Welfare},
  volume={15},
  number={1},
  pages={1--14},
  year={1997},
  publisher={Springer}
}

@article{marshall2009inequalities,
  title={Inequalities: theory of majorization and its applications},
  author={Marshall, Albert W and Olkin, Ingram and Arnold, Barry C},
  year={2009},
  publisher={Springer}
}

@article{rado1952inequality,
  title={An inequality},
  author={Rado, Richard},
  journal={Journal of the London Mathematical Society},
  volume={1},
  number={1},
  pages={1--6},
  year={1952},
  publisher={Oxford University Press}
}

@article{persico2000information,
  title={Information acquisition in auctions},
  author={Persico, Nicola},
  journal={Econometrica},
  volume={68},
  number={1},
  pages={135--148},
  year={2000},
  publisher={JSTOR}
}

@article{quah2009comparative,
  title={Comparative statics, informativeness, and the interval dominance order},
  author={Quah, John K-H and Strulovici, Bruno},
  journal={Econometrica},
  volume={77},
  number={6},
  pages={1949--1992},
  year={2009},
  publisher={Wiley Online Library}
}

@article{kim2023comparing,
  title={Comparing information in general monotone decision problems},
  author={Kim, Yonggyun},
  journal={Journal of Economic Theory},
  volume={211},
  pages={105679},
  year={2023},
  publisher={Elsevier}
}

@article{ross1973economic,
  title={The economic theory of agency: The principal's problem},
  author={Ross, Stephen A},
  journal={The American economic review},
  volume={63},
  number={2},
  pages={134--139},
  year={1973},
  publisher={JSTOR}
}

@article{demougin2001ranking,
  title={Ranking of information systems in agency models: an integral condition},
  author={Demougin, Dominique and Fluet, Claude},
  journal={Economic Theory},
  volume={17},
  pages={489--496},
  year={2001},
  publisher={Springer}
}

@article{wu2023geometric,
  title={A geometric Blackwell’s order},
  author={Wu, Wenhao},
  journal={Economics Letters},
  volume={226},
  pages={111082},
  year={2023},
  publisher={Elsevier}
}

@book{jehle2001advanced,
  title={Advanced Microeconomic Theory},
  author={Jehle, Geoffrey Alexander and Reny, Philip J},
  year={2001},
  publisher={Pearson Education}
}

@article{moscarini2002law,
  title={The law of large demand for information},
  author={Moscarini, Giuseppe and Smith, Lones},
  journal={Econometrica},
  volume={70},
  number={6},
  pages={2351--2366},
  year={2002},
  publisher={Wiley Online Library}
}

@article{azrieli2022elicitability,
  title={Elicitability},
  author={Azrieli, Yaron and Chambers, Christopher and Healy, Paul and Lambert, Nicolas},
  journal={arXiv preprint arXiv:2510.00879},
  year={2025}
}

@article{xia2025expert,
  title={Expert Incentives under Partially Contractible States},
  author={Xia, Zizhe},
  journal={arXiv preprint arXiv:2508.10170},
  year={2025}
}

@article{dewatripont1999economics,
  title={The economics of career concerns, part I: Comparing information structures},
  author={Dewatripont, Mathias and Jewitt, Ian and Tirole, Jean},
  journal={The Review of Economic Studies},
  volume={66},
  number={1},
  pages={183--198},
  year={1999},
  publisher={Wiley-Blackwell}
}

@article{holmstrom1982moral,
  title={Moral hazard in teams},
  author={Holmstr{\"o}m, Bengt},
  journal={The Bell journal of economics},
  pages={324--340},
  year={1982},
  publisher={JSTOR}
}

@article{grossman1983implicit,
  title={Implicit contracts under asymmetric information},
  author={Grossman, Sanford J and Hart, Oliver D},
  journal={The quarterly journal of economics},
  pages={123--156},
  year={1983},
  publisher={JSTOR}
}

@book{hardy1934inequalities,
  title={Inequalities},
  author={Hardy, G.H. and Littlewood, J.E. and P{\'o}lya, G.},
  year={1934},
  publisher={Cambridge University Press}
}

@article{karamata1932inegalite,
  title={Sur une in{\'e}galit{\'e} relative aux fonctions convexes},
  author={Karamata, Jovan},
  journal={Publications de l'Institut mathematique},
  volume={1},
  number={1},
  pages={145--147},
  year={1932},
  publisher={Matemati{\v{c}}ki institut SANU}
}

@book{shaked2007stochastic,
  title={Stochastic orders},
  author={Shaked, Moshe and Shanthikumar, J George},
  year={2007},
  publisher={Springer}
}

@article{shavell1979risk,
  title={Risk sharing and incentives in the principal and agent relationship},
  author={Shavell, Steven},
  journal={The Bell Journal of Economics},
  pages={55--73},
  year={1979},
  publisher={JSTOR}
}

@article{chatzikokolakis2020refinement,
  title={Refinement orders for quantitative information flow and differential privacy},
  author={Chatzikokolakis, Konstantinos and Fernandes, Natasha and Palamidessi, Catuscia},
  journal={Journal of Cybersecurity and Privacy},
  volume={1},
  number={1},
  pages={40--77},
  year={2020},
  publisher={MDPI}
}

@article{keeney1973risk,
  title={Risk independence and multiattributed utility functions},
  author={Keeney, Ralph L},
  journal={Econometrica: Journal of the Econometric Society},
  pages={27--34},
  year={1973},
  publisher={JSTOR}
}

@article{sappington1983limited,
  title={Limited liability contracts between principal and agent},
  author={Sappington, David},
  journal={Journal of economic Theory},
  volume={29},
  number={1},
  pages={1--21},
  year={1983},
  publisher={Elsevier}
}

@article{poblete2012form,
  title={The form of incentive contracts: agency with moral hazard, risk neutrality, and limited liability},
  author={Poblete, Joaqu{\'\i}n and Spulber, Daniel},
  journal={The RAND Journal of Economics},
  volume={43},
  number={2},
  pages={215--234},
  year={2012},
  publisher={Wiley Online Library}
}

@article{gottlieb2022simple,
  title={Simple contracts with adverse selection and moral hazard},
  author={Gottlieb, Daniel and Moreira, Humberto},
  journal={Theoretical Economics},
  volume={17},
  number={3},
  pages={1357--1401},
  year={2022},
  publisher={Wiley Online Library}
}

@article{jewitt2008moral,
  title={Moral hazard with bounded payments},
  author={Jewitt, Ian and Kadan, Ohad and Swinkels, Jeroen M},
  journal={Journal of Economic Theory},
  volume={143},
  number={1},
  pages={59--82},
  year={2008},
  publisher={Elsevier}
}

\pagebreak{}

\appendix
\setcounter{page}{1}
\section{Flexible Moral Hazard Problems \label{appsec:MH-detail}}

This appendix provides details for the class $\mathcal{P}_{F}$ of
flexible moral hazard problems introduced in Section \ref{sec:MH-problems}
of the main text. Flexible problems impose additional assumptions
on the cost function that render the first-order condition both necessary
and sufficient for the agent's problem. This allows me to construct
instances of moral hazard problems to establish the necessity of the
orders in Section \ref{sec:orders}. Appendix Section \ref{appsubsec:agent-cost}
states the assumptions on cost function. Appendix Section \ref{appsubsec:A-IC}
derives the agent's incentive constraint. 

\subsection{Production Cost \label{appsubsec:agent-cost}}

In flexible moral hazard problems, the agent directly chooses a state
distribution to produce. The action space $A=\Delta\Omega$ with $\boldsymbol{\mu}_{a}=a$.
I therefore use $\boldsymbol{\mu}$ to refer to an action in flexible
problems. The cost of action $\boldsymbol{\mu}$ is $C(\boldsymbol{\mu})$
where $C:\Delta\Omega\to\mathbb{R}_{+}\cup\{+\infty\}$. The production
cost $C$ is assumed to be convex and differentiable, in addition
to being lower semi-continuous and normalized to have a free option
$\underline{\boldsymbol{\mu}}$ with $C(\underline{\boldsymbol{\mu}})=0$.
The two new assumptions warrant further discussion. 

Convexity is without loss in flexible models because the agent is
allowed to randomize. The cost of producing $\boldsymbol{\mu}$ is
the cheapest expected cost across all randomizations that generate
$\boldsymbol{\mu}$. The resulting cost function must be convex. 

Differentiability requires $C$ to be Gateaux differentiable on the
relative interior of its effective domain $A_{C}$. Formally, the
effective domain of $C$ is $A_{C}:=\left\{ \boldsymbol{\mu}\in\Delta\Omega:C(\boldsymbol{\mu})<+\infty\right\} $,
the set of actions with a finite cost. Gateaux differentiability is
only meaningful on the relative interior of the effective domain.
Let $\operatorname{ri}\left(A_{C}\right)$ denote this relative interior.
Gateaux differentiability of $C$ requires the existence of a function
$\boldsymbol{\nabla}C:\operatorname{ri}\left(A_{C}\right)\to\bar{\mathbb{R}}^{N}$
where $\bar{\mathbb{R}}:=\mathbb{R}\cup\{\pm\infty\}$ such that,
for any $\boldsymbol{\mu}\in\operatorname{ri}\left(A_{C}\right),\boldsymbol{\mu}'\in\Delta\Omega$,
\begin{equation}
\lim_{\epsilon\downarrow0}\dfrac{1}{\epsilon}\left[C\left(\boldsymbol{\mu}+\epsilon(\boldsymbol{\mu}'-\boldsymbol{\mu})\right)-C\left(\boldsymbol{\mu}\right)\right]=(\boldsymbol{\mu}'-\boldsymbol{\mu})\cdot\boldsymbol{\nabla}C(\boldsymbol{\mu}).\label{eq:differentiability}
\end{equation}
The function $\boldsymbol{\nabla}C$ is called a (Gateaux) derivative
of $C$. Equation (\ref{eq:differentiability}) says that the cost
of any marginal change in production can be priced linearly with weights
$\boldsymbol{\nabla}C(\boldsymbol{\mu})$. The $n$-th component of
$\boldsymbol{\nabla}C(\boldsymbol{\mu})$ is the marginal cost of
increasing the probability of state $\omega_{n}$. Up to a normalizing
constant, $\boldsymbol{\nabla}C(\boldsymbol{\mu})$ is the usual derivative
of $C$ viewed as a function from a subset of $\mathbb{R}^{N}$ to
$\mathbb{R}_{+}\cup\{+\infty\}$. 

As a technical note, Equation (\ref{eq:differentiability}) defines
the derivative only up to an additive constant. If $\boldsymbol{\nabla}C(\boldsymbol{\mu})$
satisfies Equation (\ref{eq:differentiability}), so does $\boldsymbol{\nabla}C(\boldsymbol{\mu})+b\boldsymbol{1}$
for any $b\in\mathbb{R}$, because the constant $b$ cancels out:
the perturbation $\boldsymbol{\mu}'-\boldsymbol{\mu}$ has components
summing to zero. Alternatively, if we view $C$ as a function defined
over $\Delta\Omega\subseteq\mathbb{R}^{N-1}$, its partial derivative
with respect to the $n$-th component is the marginal cost of shifting
probability from $\omega_{N}$ to $\omega_{n}$. The partial derivatives
become $\boldsymbol{\nabla}C(\boldsymbol{\mu})+b\boldsymbol{1}$ with
$b=-\boldsymbol{\nabla}C(\boldsymbol{\mu})_{N}$ to normalize its
$N$-th component to zero, since $\omega_{N}$ is picked as the residual
state: its probability is one minus the sum of the probabilities of
all other states. Any choice of normalizing constant $b$ leaves the
analysis unchanged.

\subsection{Incentive Constraint \label{appsubsec:A-IC}}

I present two lemmas on the agent's incentive constraint for flexible
problems. They give necessary and sufficient conditions for the agent's
optimal choice. Given the principal's information $\mathcal{E}\in E^{M}$
and contract $\boldsymbol{t}\in\mathbb{R}^{M}$, the next lemma characterizes
when $\boldsymbol{\mu}_{0}$ maximizes the agent's objective (\ref{eq:A-payoff}). 
\begin{lem}
\label{lem:A-optimality} $\boldsymbol{\mu}_{0}\in\underset{\boldsymbol{\mu}\in\Delta\Omega}{\operatorname{argmax}}\:U\left(\boldsymbol{\mu};\mathcal{E},\boldsymbol{t}\right)$
if and only if there exists some $\lambda\in\mathbb{R}$ and $\boldsymbol{\eta}\in\mathbb{R}_{+}^{N}$
such that 
\begin{align}
\mathcal{E}u(\boldsymbol{t}) & =\boldsymbol{\nabla}C(\boldsymbol{\mu}_{0})-\lambda\boldsymbol{1}-\boldsymbol{\eta},\label{eq:FOC}\\
\eta_{n}\boldsymbol{\mu}_{0}(\omega_{n}) & =0,\forall1\leq n\leq N,\label{eq:comp-slackness}
\end{align}
where $u(\boldsymbol{t}):=\left[u(t_{m})\right]_{m=1}^{M}\in\mathbb{R}^{M}$
and $\boldsymbol{\mu}_{0}(\omega_{n})$ denotes the probability of
$\omega_{n}$ under $\boldsymbol{\mu}_{0}$.
\end{lem}
\begin{proof}
Recall that the agent's objective can be rewritten as $U\left(\boldsymbol{\mu};\mathcal{E},\boldsymbol{t}\right)=\boldsymbol{\mu}\cdot\mathcal{E}u(\boldsymbol{t})-C\left(\boldsymbol{\mu}\right)$.
I first show the necessity. Indeed, this is a first-order condition.
Suppose $\boldsymbol{\mu}_{0}\in\underset{\boldsymbol{\mu}\in\Delta\Omega}{\operatorname{argmax}}\:U\left(\boldsymbol{\mu};\mathcal{E},\boldsymbol{t}\right)$.
View the agent's problem as choosing $\boldsymbol{\mu}\in\mathbb{R}^{N}$
subject to $\boldsymbol{\mu}\cdot\boldsymbol{1}=1$ and $\boldsymbol{\mu}\geq0$,
since the probabilities must add up to one and be non-negative.\footnote{The probabilities are also bounded above by one. But this is implied
by non-negativity and adding-up.} Let $\lambda\in\mathbb{R}$ be the multiplier on the adding-up constraint,
and $\boldsymbol{\eta}\in\mathbb{R}_{+}^{N}$ be the vector of multipliers
on the non-negativity constraint. The Lagrangian of the problem is
\[
\mathcal{L}(\boldsymbol{\mu},\lambda,\boldsymbol{\eta}):=\boldsymbol{\mu}\cdot\mathcal{E}u(\boldsymbol{t})-C\left(\boldsymbol{\mu}\right)+\lambda(\boldsymbol{\mu}\cdot\boldsymbol{1}-1)+\boldsymbol{\mu}\cdot\boldsymbol{\eta}.
\]
Optimality of $\boldsymbol{\mu}_{0}$ implies that no perturbation
of $\boldsymbol{\mu}_{0}$ can improve the value of $\mathcal{L}(\boldsymbol{\mu},\lambda,\boldsymbol{\eta})$.
Specifically, perturb $\boldsymbol{\mu}_{0}$ to any $\boldsymbol{\mu}=\boldsymbol{\mu}_{0}+\epsilon\boldsymbol{d}$
for some $\epsilon>0$ and $\boldsymbol{d}\in\mathbb{R}^{N}$ with
$\boldsymbol{d}\cdot\boldsymbol{1}=0$,\footnote{The condition $\boldsymbol{d}\cdot\boldsymbol{1}=0$ comes from requiring
$\boldsymbol{\mu}$ to satisfy adding-up.} we must have
\[
\lim_{\epsilon\downarrow0}\;\dfrac{1}{\epsilon}\left[C\left(\boldsymbol{\mu}_{0}+\epsilon\boldsymbol{d}\right)-C\left(\boldsymbol{\mu}_{0}\right)\right]-\boldsymbol{d}\cdot\mathcal{E}u(\boldsymbol{t})-\lambda\boldsymbol{d}\cdot\boldsymbol{1}-\boldsymbol{d}\cdot\boldsymbol{\eta}\geq0.
\]
Apply the definition of $\boldsymbol{\nabla}C$ and notice $\boldsymbol{d}\cdot\boldsymbol{1}=0$,
\[
\boldsymbol{d}\cdot\left[\boldsymbol{\nabla}C(\boldsymbol{\mu}_{0})-\mathcal{E}u(\boldsymbol{t})-\lambda\boldsymbol{1}-\boldsymbol{\eta}\right]\geq0.
\]
This must hold for any $\boldsymbol{d}\in\mathbb{R}^{N}$ with $\boldsymbol{d}\cdot\boldsymbol{1}=0$.
The bracket term has to be proportional to $\boldsymbol{1}$, which
absorbs into $\lambda$. Therefore, we must have (\ref{eq:FOC}),
and (\ref{eq:comp-slackness}) is the corresponding complementary
slackness condition. If $\boldsymbol{\mu}_{0}$ is interior, the condition
reduces to $\mathcal{E}u(\boldsymbol{t})=\boldsymbol{\nabla}C(\boldsymbol{\mu}_{0})-\lambda\boldsymbol{1}$
for some $\lambda\in\mathbb{R}$. This completes the proof for necessity.

For sufficiency, observe that the agent's problem is convex, because
in his objective (\ref{eq:A-payoff}), the expected utility from money
is linear in $\boldsymbol{\mu}$, and the production cost is differentiable
and convex. As a result, the first-order condition is sufficient for
global optimality. 
\end{proof}
By Lemma \ref{lem:A-optimality}, the principal's problem reduces
to choosing the agent's state-dependent utility $\mathcal{E}u(\boldsymbol{t})$
to minimize her expected cost, subject to the agent's first-order
condition (\ref{eq:FOC}) and (\ref{eq:comp-slackness}), as well
as his participation constraint (\ref{eq:IR}) and possibly the limited
liability (\ref{eq:LL}) and ex post budget (\ref{eq:B}) constraints.

The first-order condition simplifies further if $\boldsymbol{\mu}_{0}$
is interior, in which case $\boldsymbol{\eta}=\boldsymbol{0}$. The
next lemma summarizes this result. 

\begin{lem}
\label{lem:A-opt-interior} Suppose $\boldsymbol{\mu}_{0}\in\operatorname{Int}\Delta\Omega$.
$\boldsymbol{\mu}_{0}\in\underset{\boldsymbol{\mu}\in\Delta\Omega}{\operatorname{argmax}}\:U\left(\boldsymbol{\mu};\mathcal{E},\boldsymbol{t}\right)$
if and only if there exists some $\lambda\in\mathbb{R}$ such that
\begin{align}
\mathcal{E}u(\boldsymbol{t}) & =\boldsymbol{\nabla}C(\boldsymbol{\mu}_{0})-\lambda\boldsymbol{1},\label{eq:FOC-interior}
\end{align}
where $u(\boldsymbol{t}):=\left[u(t_{m})\right]_{m=1}^{M}\in\mathbb{R}^{M}$. 
\end{lem}
\begin{rem}
\label{rem:no-boundary} It suffices to focus on $\boldsymbol{\mu}_{0}\in\operatorname{Int}\Delta\Omega$
if the production cost $C$ satisfies $\lim_{\boldsymbol{\mu}\to\boldsymbol{\mu}'}\left\Vert \boldsymbol{\nabla}C(\boldsymbol{\mu})\right\Vert =+\infty$
for all $\boldsymbol{\mu}'\in\partial\Delta\Omega$, where $\partial\Delta\Omega$
denotes the boundary of $\Delta\Omega$, the set of all distributions
that are not fully supported. That is, the marginal cost to rule out
any state is infinite. In this case, it is infinitely costly to implement
any non-fully-supported state distribution. Only interior $\boldsymbol{\mu}_{0}$
can be implemented. The proofs in the next section repeatedly employ
such cost functions to construct instances of moral hazard problems.
\end{rem}

\section{Omitted Proofs \label{appsec:proof-appendix}}

This appendix collects the omitted proofs from the main text. 
\begin{proof}
[Proof of Theorem \ref{thm:column-space}] In addition to the equivalences
in Theorem \ref{thm:column-space}, I also prove the equivalence of
column space dominance to condition $(5')$ which is a variant of
condition $(5)$ with any fixed interior prior. 
\begin{itemize}
\item [$(5^\prime)$] $\operatorname{Aff}\operatorname{Supp}\pi\left(\mathcal{E};\boldsymbol{\nu}\right)\supseteq\operatorname{Aff}\operatorname{Supp}\pi\left(\mathcal{E}';\boldsymbol{\nu}\right)$
for some interior prior $\boldsymbol{\nu}\in\Delta\Omega$.
\end{itemize}
$(1)\Rightarrow(2)$. To implement some $a_{0}$ given some $A\in\mathcal{A},u\in\mathcal{U},C\in\mathcal{C}$,
it suffices to find some contract $\boldsymbol{t}$ to generate a
state-dependent utility such that $a_{0}$ is optimal for the agent's
problem (\ref{eq:A-problem}). If $\mathcal{E}\geq_{\text{Col}}\mathcal{E}'$,
then any state-dependent utility $\boldsymbol{u}$ that can be generated
under $\mathcal{E}'$ can also be generated under $\mathcal{E}$:
for any $\boldsymbol{u}:=\mathcal{E}'u(\boldsymbol{t}')$ for some
$\boldsymbol{t}'\in\mathbb{R}^{M'}$, one can always find $\boldsymbol{t}\in\mathbb{R}^{M}$
so that $u(\boldsymbol{t})=Gu(\boldsymbol{t}')$, and $\boldsymbol{t}$
generates the same state-dependent utility $\mathcal{E}u(\boldsymbol{t})=\mathcal{E}Gu(\boldsymbol{t}')=\mathcal{E}'u(\boldsymbol{t}')=\boldsymbol{u}$.
Therefore, if there exists a feasible $\boldsymbol{u}$ under $\mathcal{E}'$
such that $a_{0}$ solves the agent's problem (\ref{eq:A-problem}),
under $\mathcal{E}$ the principal can generate the same $\boldsymbol{u}$
so that $a_{0}$ remains optimal. 

$(2)\Rightarrow(3)$. It suffices to observe that when the agent is
risk neutral without limited liability, any $a_{0}\in\mathcal{I}(\mathcal{E};A,u,C)$
can be implemented at the first-best cost. To see this, suppose $a_{0}$
can be implemented by some contract $\boldsymbol{t}\in\mathbb{R}^{M}$.
The principal can simply take away a constant bonus from $\boldsymbol{t}$
until the agent earns no rent. The cost to the principal is always
$C(a_{0})+\underline{u}$, regardless of her information. The only
place information matters is to determine $\mathcal{I}(\mathcal{E};A,u,C)$
because if $a_{0}\not\in\mathcal{I}(\mathcal{E};A,u,C)$, the cost
is infinity. Therefore, the comparisons of cost are implied by the
comparisons of implementability.

$(2)\Rightarrow(1)$, $(3)\Rightarrow(1)$. I consider the contrapositive.
Suppose $\mathcal{E}\not\geq_{\text{Col}}\mathcal{E}'$. There exists
some $\boldsymbol{x}\in\operatorname{Col}\mathcal{E}'\setminus\operatorname{Col}\mathcal{E}$.
I construct a moral hazard problem in $\mathcal{P}_{F}\cap\mathcal{P}_{1}$
with infinite cost under $\mathcal{E}$ but feasible under $\mathcal{E}'$.
This proves $(2)\Rightarrow(1)$ and $(3)\Rightarrow(1)$ at the same
time. Consider a risk neutral agent with $u(t)=t$, and pick a flexible
cost function $C\in\mathcal{C}_{F}$ such that $\lim_{\boldsymbol{\mu}\to\boldsymbol{\mu}'}\left\Vert \boldsymbol{\nabla}C(\boldsymbol{\mu})\right\Vert =+\infty$
for all $\boldsymbol{\mu}'\in\partial\Delta\Omega$ where $\partial\Delta\Omega$
is the boundary of $\Delta\Omega$, that is, the set of non-fully-supported
distributions. Such cost functions make sure that the first-order
condition (\ref{eq:FOC}) does not involve multipliers $\boldsymbol{\eta}\in\mathbb{R}_{+}^{N}$
(see Remark \ref{rem:no-boundary}). I can choose some strictly convex
$C$ and action $\boldsymbol{\mu}_{0}$ so that $\boldsymbol{\nabla}C(\boldsymbol{\mu}_{0})=\boldsymbol{x}$,
which makes it implementable under $\mathcal{E}'$ but not $\mathcal{E}$.
Here I identify the target action with its induced state distribution
$\boldsymbol{\mu}_{0}\in\Delta\Omega$. Note that to implement $\boldsymbol{\mu}_{0}$,
the principal can pick any multiplier $\lambda\in\mathbb{R}$ to satisfy
(\ref{eq:FOC}). But if $\boldsymbol{x}\not\in\operatorname{Col}\mathcal{E}$,
the principal cannot satisfy (\ref{eq:FOC}) for any value of $\lambda$.
This is because, if so, $\boldsymbol{x}+\lambda\boldsymbol{1}\in\operatorname{Col}\mathcal{E}$
becomes a contradiction to $\boldsymbol{x}\not\in\operatorname{Col}\mathcal{E}$
since $\boldsymbol{1}\in\operatorname{Col}\mathcal{E}$.

$(1)\Rightarrow(4)$. Consider the $m$-th column of $\mathcal{E}'$,
denoted $\mathcal{E}_{m}'$. Due to the column space inclusion, $\mathcal{E}_{m}'\in\operatorname{Col}\mathcal{E}'\subseteq\operatorname{Col}\mathcal{E}$.
Therefore, for each $1\leq m\leq M'$, $\mathcal{E}_{m}'=\mathcal{E}\boldsymbol{v}_{m}$
for some $\boldsymbol{v}_{m}$. We have $\mathcal{E}'=\mathcal{E}G$
for $G=\begin{bmatrix}\boldsymbol{v}_{1} & \boldsymbol{v}_{2} & \dots & \boldsymbol{v}_{M'}\end{bmatrix}$.

$(4)\Rightarrow(1)$. Take any $\boldsymbol{x}\in\operatorname{Col}\mathcal{E}'$.
$\boldsymbol{x}=\mathcal{E}'\boldsymbol{v}$ for some $\boldsymbol{v}\in\mathbb{R}^{M'}$.
Thus, $\boldsymbol{x}=\mathcal{E}'\boldsymbol{v}=\mathcal{E}G\boldsymbol{v}\in\operatorname{Col}\mathcal{E}$.

$(4)\Rightarrow(5)$. Take any prior $\boldsymbol{\nu}$. Consider
$\mathcal{E}$ first. Apply the Bayes' rule, and the posterior $\boldsymbol{\mu}_{m}$
induced by realization $y_{m}$ is given by 
\begin{equation}
\boldsymbol{\mu}_{m}(\omega_{n})=\dfrac{\boldsymbol{\nu}(\omega_{n})\mathcal{E}_{n,m}}{\sum_{n'=1}^{N}\boldsymbol{\nu}(\omega_{n'})\mathcal{E}_{n',m}},\label{eq:Bayes-rule}
\end{equation}
where $\mathcal{E}_{n,m}$ is the probability of $y_{m}$ in state
$\omega_{n}$. The same formula applies to $\mathcal{E}'$ as well.
We can use $\mathcal{E}'=\mathcal{E}G$ to rewrite the posteriors
induced by $\mathcal{E}'$ as a linear combination of those induced
by $\mathcal{E}$. Specifically, the posterior $\boldsymbol{\mu}'_{m'}$
induced by realization $y'_{m'}$ is 
\begin{align}
\boldsymbol{\mu}'_{m'}(\omega_{n}) & =\frac{\boldsymbol{\nu}(\omega_{n})\mathcal{E}'_{n,m'}}{\sum_{n'=1}^{N}\boldsymbol{\nu}(\omega_{n'})\mathcal{E}'_{n',m'}}\nonumber \\
 & =\frac{\boldsymbol{\nu}(\omega_{n})\sum_{m=1}^{M}\mathcal{E}_{n,m}G_{m,m'}}{\sum_{n'=1}^{N}\boldsymbol{\nu}(\omega_{n'})\sum_{m=1}^{M}\mathcal{E}_{n',m}G_{m,m'}}\nonumber \\
 & =\frac{\sum_{m=1}^{M}G_{m,m'}\boldsymbol{\nu}(\omega_{n})\mathcal{E}_{n,m}}{\sum_{m=1}^{M}G_{m,m'}\sum_{n'=1}^{N}\boldsymbol{\nu}(\omega_{n'})\mathcal{E}_{n',m}}\nonumber \\
 & =\frac{\sum_{m=1}^{M}\boldsymbol{\mu}_{m}(\omega_{n})G_{m,m'}\sum_{n'=1}^{N}\boldsymbol{\nu}(\omega_{n'})\mathcal{E}_{n',m}}{\sum_{m=1}^{M}G_{m,m'}\sum_{n'=1}^{N}\boldsymbol{\nu}(\omega_{n'})\mathcal{E}_{n',m}}\nonumber \\
 & =\sum_{m=1}^{M}\boldsymbol{\mu}_{m}(\omega_{n})w_{m,m'}\label{eq:E-E'-posteriors}
\end{align}
with
\[
w_{m,m'}:=\frac{G_{m,m'}\sum_{n'=1}^{N}\boldsymbol{\nu}(\omega_{n'})\mathcal{E}_{n',m}}{\sum_{m=1}^{M}G_{m,m'}\sum_{n'=1}^{N}\boldsymbol{\nu}(\omega_{n'})\mathcal{E}_{n',m}},
\]
where the second equality uses $\mathcal{E}'_{n,m'}=\sum_{m=1}^{M}\mathcal{E}_{n,m}G_{m,m'}$,
the third equality changes the order of summation, and the fourth
equality plugs in (\ref{eq:Bayes-rule}). From the last expression
(\ref{eq:E-E'-posteriors}), $\boldsymbol{\mu}'_{m'}$ is an affine
combination of $\boldsymbol{\mu}_{m}$ with weights $w_{m,m'}$ since
the weights sum to one $\sum_{m=1}^{M}w_{m,m'}=1$ for any $m'$.
Therefore, $\operatorname{Aff}\operatorname{Supp}\pi\left(\mathcal{E};\boldsymbol{\nu}\right)\supseteq\operatorname{Supp}\pi\left(\mathcal{E}';\boldsymbol{\nu}\right)$
and this implies $(5)$ by taking affine spans on both sides. 

$(5)\Rightarrow(4)$. Take $\boldsymbol{\nu}$ to be the uniform prior.
In this case, 
\[
\boldsymbol{\mu}_{m}(\omega_{n})=\dfrac{\mathcal{E}_{n,m}}{\sum_{n'=1}^{N}\mathcal{E}_{n',m}}.
\]
The affine span inclusion thus implies that for any column $\mathcal{E}_{m}'$
of $\mathcal{E}'$, we have $\mathcal{E}_{m}'=\mathcal{E}\boldsymbol{v}_{m}$
for some $\boldsymbol{v}_{m}$, which implies the matrix factorization
condition. 

$(5)\Leftrightarrow(5^{\prime})$. The $\Rightarrow$ direction is
trivial. For $\Leftarrow$, take any interior prior $\boldsymbol{\nu}$
such that $(5')$ holds and consider any other prior $\boldsymbol{\nu}'$.
It suffices to show that $\operatorname{Aff}\operatorname{Supp}\pi\left(\mathcal{E};\boldsymbol{\nu}'\right)\supseteq\operatorname{Supp}\pi\left(\mathcal{E}';\boldsymbol{\nu}'\right)$,
since $(5)$ follows from taking affine spans on both sides. We first
express the posteriors in $\operatorname{Supp}\pi(\mathcal{E};\boldsymbol{\nu}')$
as some linear transformation of posteriors in $\operatorname{Supp}\pi(\mathcal{E};\boldsymbol{\nu})$,
and similarly for $\mathcal{E}'$. Specifically, let $\boldsymbol{\mu}_{m}$
and $\boldsymbol{\mu}_{m}'$ be the posteriors of $\mathcal{E}$ induced
by $y_{m}$ when the prior is $\boldsymbol{\nu}$ and $\boldsymbol{\nu}'$,
and let $\boldsymbol{\xi}_{m'}$ and $\boldsymbol{\xi}_{m'}'$ be
the posteriors of $\mathcal{E}'$ induced by $y'_{m'}$ when the prior
is $\boldsymbol{\nu}$ and $\boldsymbol{\nu}'$. From (\ref{eq:Bayes-rule}),
we have for $\boldsymbol{\mu}{}_{m}$ and $\boldsymbol{\mu}_{m}'$,
\begin{align*}
\boldsymbol{\mu}_{m}'(\omega_{n}) & =\dfrac{\boldsymbol{\nu}'(\omega_{n})\mathcal{E}_{n,m}}{\sum_{n'=1}^{N}\boldsymbol{\nu}'(\omega_{n'})\mathcal{E}_{n',m}}\\
 & =\boldsymbol{\mu}_{m}(\omega_{n})\cdot\dfrac{\boldsymbol{\nu}'(\omega_{n})}{\boldsymbol{\nu}(\omega_{n})}\cdot\dfrac{\sum_{n'=1}^{N}\boldsymbol{\nu}(\omega_{n'})\mathcal{E}_{n',m}}{\sum_{n'=1}^{N}\boldsymbol{\nu}'(\omega_{n'})\mathcal{E}_{n',m}}\\
 & =\boldsymbol{\mu}_{m}(\omega_{n})\cdot\alpha_{n}\cdot\beta_{m}
\end{align*}
where 
\[
\alpha_{n}:=\dfrac{\boldsymbol{\nu}'(\omega_{n})}{\boldsymbol{\nu}(\omega_{n})},\beta_{m}:=\dfrac{\sum_{n'=1}^{N}\boldsymbol{\nu}(\omega_{n'})\mathcal{E}_{n',m}}{\sum_{n'=1}^{N}\boldsymbol{\nu}'(\omega_{n'})\mathcal{E}_{n',m}}.
\]
The factor $\alpha_{n}$ depends only on the state, and $\beta_{m}$
only on the realization. Both are strictly positive since $\boldsymbol{\nu}$
is interior. The same factorization applies to posteriors $\boldsymbol{\xi}_{m'}$
and $\boldsymbol{\xi}_{m'}'$, that is, 
\[
\boldsymbol{\xi}_{m'}'(\omega_{n})=\boldsymbol{\xi}_{m'}(\omega_{n})\cdot\alpha_{n}\cdot\beta_{m'}',
\]
where 
\[
\beta_{m'}':=\dfrac{\sum_{n'=1}^{N}\boldsymbol{\nu}(\omega_{n'})\mathcal{E}'_{n',m'}}{\sum_{n'=1}^{N}\boldsymbol{\nu}'(\omega_{n'})\mathcal{E}'_{n',m'}}.
\]
We now show that $\operatorname{Aff}\operatorname{Supp}\pi\left(\mathcal{E};\boldsymbol{\nu}'\right)\supseteq\operatorname{Supp}\pi\left(\mathcal{E}';\boldsymbol{\nu}'\right)$
by proving $\boldsymbol{\xi}_{m'}'\in\operatorname{Aff}\operatorname{Supp}\pi\left(\mathcal{E};\boldsymbol{\nu}'\right)$
for any realization $y'_{m'}$ of $\mathcal{E}'$. From $(5')$ we
know that the corresponding posterior $\boldsymbol{\xi}_{m'}$ is
an affine combination of $\boldsymbol{\mu}_{m}$'s, 
\[
\boldsymbol{\xi}_{m'}=\sum_{m=1}^{M}\gamma_{m}\boldsymbol{\mu}_{m}
\]
with $\sum_{m=1}^{M}\gamma_{m}=1$. Therefore, we have 
\begin{align}
\boldsymbol{\xi}_{m'}'(\omega_{n}) & =\boldsymbol{\xi}_{m'}(\omega_{n})\cdot\alpha_{n}\cdot\beta_{m'}'\nonumber \\
 & =\sum_{m=1}^{M}\gamma_{m}\boldsymbol{\mu}_{m}(\omega_{n})\cdot\alpha_{n}\cdot\beta_{m'}'\nonumber \\
 & =\sum_{m=1}^{M}\dfrac{\gamma_{m}\beta_{m'}'}{\beta_{m}}\cdot\boldsymbol{\mu}_{m}'(\omega_{n}),\label{eq:post-comb}
\end{align}
where the equalities follow from plugging in the decompositions for
$\boldsymbol{\xi}_{m'}$ and $\boldsymbol{\mu}_{m}'$. Observe that
the coefficient $\gamma_{m}\beta_{m'}'/\beta_{m}$ before $\boldsymbol{\mu}_{m}'$
does not depend on the state. As a result, $\boldsymbol{\xi}_{m'}'$
is a linear combination of $\boldsymbol{\mu}_{m}'$. The coefficients
$\gamma_{m}\beta_{m'}'/\beta_{m}$ must sum to one: since both $\boldsymbol{\xi}_{m'}'$
and $\boldsymbol{\mu}_{m}'$ are probability distributions, taking
the dot product with $\boldsymbol{1}$ on both sides of Equation (\ref{eq:post-comb})
gives $\sum_{m=1}^{M}\gamma_{m}\beta_{m'}'/\beta_{m}=1$. Therefore,
$\boldsymbol{\xi}_{m'}'\in\operatorname{Aff}\operatorname{Supp}\pi\left(\mathcal{E};\boldsymbol{\nu}'\right)$.
This completes the proof. 
\end{proof}
\begin{proof}
[Proof of Theorem \ref{thm:conic-span}] In addition to the equivalences
in Theorem \ref{thm:conic-span}, I also prove the equivalence of
conic span dominance to condition $(2')$ which says lower cost in
all risk neutral problems with an ex post budget constraint, and condition
$(4')$ which is a variant of condition $(4)$ with any fixed interior
prior. 
\begin{itemize}
\item [$(2^\prime)$] $\kappa(\mathcal{E};P)\leq\kappa(\mathcal{E}';P)$
for any $P\in\mathcal{P}_{2}':=\{P\in\mathcal{P}:u(t)=t,\mathcal{R}=\text{B}\}$.
\item [$(4^\prime)$] $\operatorname{Co}\operatorname{Supp}\pi\left(\mathcal{E};\boldsymbol{\nu}\right)\supseteq\operatorname{Co}\operatorname{Supp}\pi\left(\mathcal{E}';\boldsymbol{\nu}\right)$
for some interior prior $\boldsymbol{\nu}\in\Delta\Omega$.
\end{itemize}
$(1)\Rightarrow(2)$. Given any problem $P\in\mathcal{P}_{2}$, suppose
the principal wants to implement $a_{0}$. If $a_{0}$ is not implementable
under $\mathcal{E}'$, $\kappa(\mathcal{E}';P)=+\infty$ and there
is nothing to prove. If $a_{0}$ is implementable under $\mathcal{E}'$,
let $\boldsymbol{u}$ be the agent's state-dependent utility generated
by the optimal contract under $\mathcal{E}'$. Due to conic span inclusion,
$\boldsymbol{u}$ is also feasible under $\mathcal{E}$. If the principal
generates the same $\boldsymbol{u}$ under $\mathcal{E}$, the agent
still optimally chooses $a_{0}$ and the principal incurs the same
expected cost. Therefore, the cost to implement $a_{0}$ under $\mathcal{E}$
is no higher than that under $\mathcal{E}'$, or $\kappa(\mathcal{E};P)\leq\kappa(\mathcal{E}';P)$. 

$(2)\Rightarrow(1)$. Suppose $\operatorname{Cone}\mathcal{E}\not\supseteq\operatorname{Cone}\mathcal{E}'$.
There exists some $\boldsymbol{x}\in\operatorname{Cone}\mathcal{E}'\setminus\operatorname{Cone}\mathcal{E}$.
We construct a problem in $\mathcal{P}_{F}\cap\mathcal{P}_{2}$ that
violates the cost comparisons. Focus on flexible cost functions $C\in\mathcal{C}_{F}$
such that $\lim_{\boldsymbol{\mu}\to\boldsymbol{\mu}'}\left\Vert \boldsymbol{\nabla}C(\boldsymbol{\mu})\right\Vert =+\infty$
for all $\boldsymbol{\mu}'\in\partial\Delta\Omega$, that is, the
marginal cost to rule out any state is infinite. This condition ensures
that $\boldsymbol{\eta}$ does not show up in the first-order condition.
Pick a cost function $C$ and action $\boldsymbol{\mu}_{0}$ so that
the marginal cost of $C$ at $\boldsymbol{\mu}_{0}$ is $\boldsymbol{x}$.
Here I identify the target action with its induced state distribution
$\boldsymbol{\mu}_{0}\in\Delta\Omega$. To optimally implement $\boldsymbol{\mu}_{0}$
under $\mathcal{E}'$, the principal already chooses the optimal state-dependent
utility $\boldsymbol{x}+\lambda\boldsymbol{1}\in\operatorname{Cone}\mathcal{E}'$
for some $\lambda\in\mathbb{R}$ according to Lemma \ref{lem:A-opt-interior}.
Under $\mathcal{E}$, the principal either cannot implement $\boldsymbol{\mu}_{0}$,
or she can at best pick some $\lambda'>\lambda$ because $\lambda'\leq\lambda$
would imply $\boldsymbol{x}\in\operatorname{Cone}\mathcal{E}$, a
contradiction. This is because $\boldsymbol{1}$ is a recessive direction
of $\operatorname{Cone}\mathcal{E}$ for any experiment $\mathcal{E}$
due to $\mathcal{E}\boldsymbol{1}=\boldsymbol{1}$.  As a result,
the cost is strictly higher under $\mathcal{E}$ than $\mathcal{E}'$. 

$(1)\Leftrightarrow(2')$. It suffices to realize, if the principal
is subject to only the ex post budget constraint $\boldsymbol{t}\leq B$,
the set of feasible state-dependent utilities is $\left\{ \mathcal{E}\boldsymbol{v}:\boldsymbol{v}\leq B\right\} $,
which reduces to $B(\boldsymbol{1}-\operatorname{Cone}\mathcal{E})$,
observing that $\mathcal{E}\boldsymbol{1}=\boldsymbol{1}$. The rest
of the proof follows the steps for $(1)\Leftrightarrow(2)$.

$(1)\Rightarrow(3)$. The $m$-th column of $\mathcal{E}'$ is in
the conic span, $\mathcal{E}_{m}'\in\operatorname{Cone}\mathcal{E}'\subseteq\operatorname{Cone}\mathcal{E}$.
Therefore, for each $1\leq m\leq M'$, $\mathcal{E}_{m}'=\mathcal{E}\boldsymbol{v}_{m}$
for some $\boldsymbol{v}_{m}\geq0$. We have $\mathcal{E}'=\mathcal{E}G$
for $G=\begin{bmatrix}\boldsymbol{v}_{1} & \boldsymbol{v}_{2} & \dots & \boldsymbol{v}_{M'}\end{bmatrix}\geq0$.

$(3)\Rightarrow(1)$. Take any $\boldsymbol{x}\in\operatorname{Cone}\mathcal{E}'$.
$\boldsymbol{x}=\mathcal{E}'\boldsymbol{v}$ for some $\boldsymbol{v}\geq0$.
Thus, $\boldsymbol{x}=\mathcal{E}'\boldsymbol{v}=\mathcal{E}G\boldsymbol{v}\in\operatorname{Cone}\mathcal{E}$
since $G\boldsymbol{v}\geq0$. 

$(3)\Rightarrow(4)$. Take any prior $\boldsymbol{\nu}$. From (\ref{eq:Bayes-rule}),
we know that $\operatorname{Co}\operatorname{Supp}\pi\left(\mathcal{E};\boldsymbol{\nu}\right)$
is simply the intersection of $\operatorname{Cone}\left\{ \boldsymbol{v}_{1},...,\boldsymbol{v}_{M}\right\} $
and the hyperplane $\boldsymbol{x}\cdot\boldsymbol{1}=1$, where $\left(\boldsymbol{v}_{m}\right)_{n}=\boldsymbol{\nu}(\omega_{n})\mathcal{E}_{n,m}$.
That is, $\boldsymbol{v}_{m}$ is the vector produced by taking entry
wise product of $\boldsymbol{\nu}$ and the $m$-th column of $\mathcal{E}$.
The conic span $\operatorname{Cone}\left\{ \boldsymbol{v}_{1},...,\boldsymbol{v}_{M}\right\} $
is therefore the conic span of matrix $\mathcal{E}$ weighted by the
prior $\boldsymbol{\nu}$. Similarly, $\operatorname{Co}\operatorname{Supp}\pi\left(\mathcal{E}';\boldsymbol{\nu}\right)$
is the intersection of $\operatorname{Cone}\left\{ \boldsymbol{v}_{1}',...,\boldsymbol{v}'_{M'}\right\} $
and the hyperplane $\boldsymbol{x}\cdot\boldsymbol{1}=1$ where $\left(\boldsymbol{v}_{m'}^{\prime}\right)_{n}=\boldsymbol{\nu}(\omega_{n})\mathcal{E}'_{n,m'}$.
The inclusion of the conic spans $\operatorname{Cone}\mathcal{E}\supseteq\operatorname{Cone}\mathcal{E}'$
immediately implies the inclusion of the weighted conic spans for
any prior $\boldsymbol{\nu}$, which implies $(4)$. 

$(4)\Rightarrow(3)$. Take $\boldsymbol{\nu}$ to be the uniform prior.
$(4)$ says that any posterior induced by $\mathcal{E}'$ is a convex
combination of posteriors induced by $\mathcal{E}$. From (\ref{eq:Bayes-rule}),
the posteriors induced by $\mathcal{E}$ and $\mathcal{E}'$ differ
from the columns of $\mathcal{E}$ and $\mathcal{E}'$ only by some
positive multiplicative constant. Therefore, columns of $\mathcal{E}'$
are positive combinations of columns of $\mathcal{E}$, that is, $(3)$
holds. 

$(4)\Leftrightarrow(4^{\prime})$. The proof is exactly the same as
that in Theorem \ref{thm:column-space}, except that we replace linear
combinations by convex combinations. 
\end{proof}
\begin{proof}
[Proof of Theorem \ref{thm:zonotope} and Lemma \ref{lem:feasible-set-budget}]
In addition to the equivalences in Theorem \ref{thm:zonotope}, I
also prove the equivalence of zonotope dominance to conditions $(3a)$-$(3d)$
which restrict attention to moral hazard problems with no additional
constraint, with limited liability, with an ex post budget, or with
both; condition $(5')$ which is a variant of $(5)$ with a fixed
interior prior; and condition $(6)$ that compares the feasible state-dependent
utilities given an ex ante budget. 
\begin{itemize}
\item [$(3a)$] $\kappa(\mathcal{E};P)\leq\kappa(\mathcal{E}';P)$ for any
$P\in\mathcal{P}_{\text{\ensuremath{\emptyset}}}:=\{P\in\mathcal{P}:\mathcal{R}=\emptyset\}$,
\item [$(3b)$] $\kappa(\mathcal{E};P)\leq\kappa(\mathcal{E}';P)$ for any
$P\in\mathcal{P}_{\text{LL}}:=\{P\in\mathcal{P}:\mathcal{R}=\text{LL}\}$,
\item [$(3c)$] $\kappa(\mathcal{E};P)\leq\kappa(\mathcal{E}';P)$ for any
$P\in\mathcal{P}_{\text{B}}:=\{P\in\mathcal{P}:\mathcal{R}=\text{B}\}$,
\item [$(3d)$] $\kappa(\mathcal{E};P)\leq\kappa(\mathcal{E}';P)$ for any
$P\in\mathcal{P}_{\text{LL,B}}:=\{P\in\mathcal{P}:\mathcal{R}=\text{LL},\text{B}\}$,
\item [$(5^\prime)$] $\pi\left(\mathcal{E};\boldsymbol{\nu}\right)\geq_{\text{lcx}}\pi\left(\mathcal{E}';\boldsymbol{\nu}\right)$
for some interior prior $\boldsymbol{\nu}\in\Delta\Omega$,
\item [$(6)$] $\mathcal{V}_{\boldsymbol{\mu},u,B}(\mathcal{E})\supseteq\mathcal{V}_{\boldsymbol{\mu},u,B}(\mathcal{E}')$
for any $\boldsymbol{\mu}\in\Delta\Omega$, $u\in\mathcal{U}$, and
$B>0$.
\end{itemize}
I first prove $(1)\Leftrightarrow(4),(5),(5')$.

$(1)\Rightarrow(4)$. Any partial sum of the columns in $\mathcal{E}'$
is still in the zonotope. Therefore, $\mathcal{E}'B=\mathcal{E}D$
with $0\leq D\leq1$.

$(4)\Rightarrow(1)$. To show that $\operatorname{Zon}\mathcal{E}'\subseteq\operatorname{Zon}\mathcal{E}$,
it suffices to show that all the extreme points of $\operatorname{Zon}\mathcal{E}'$
is included in $\operatorname{Zon}\mathcal{E}$. This is exactly what
$\mathcal{E}'B=\mathcal{E}D$ says. 

$(1)\Leftrightarrow(5),(5^{\prime})$. This follows directly from
\citet{koshevoy1996lorenz}. 

Next, I prove the more involved part $(1)\Leftrightarrow(6)\Leftrightarrow(3a)\Leftrightarrow(2)$.
I briefly sketch the proof before diving into the details. $(1)\Rightarrow(6)$
is the key step of the proof. It follows from observing that both
the zonotopes and the utility sets are convex, and the inclusion of
convex sets is equivalent to a separating hyperplane condition. Using
a majorization argument, the separating hyperplane version of $(1)$
implies that of $(6)$. The rest of the proof is simpler. $(6)\Rightarrow(3a)$
comes from enlarging the principal's choice set at any given budget
level. $(3a)\Rightarrow(2)$ comes from taking utility $u_{0}(t)=\min\{t,B\}$
when $t\geq0$ and $u_{0}(t)=-\infty$ when $t<0$, which satiates
after receiving a payment of $B$ to mimic the ex post budget constraint,
and drops to $-\infty$ for negative payment to mimic the limited
liability constraint. Technically speaking, the utility is required
to be strictly increasing with its range in $\mathbb{R}$. I need
to take a sequence $u_{i}\in\mathcal{U}$ that converges pointwise
to $u_{0}$. $(2)\Rightarrow(1)$ comes from considering the contrapositive
and constructing a moral hazard problem in $\mathcal{P}_{F}\cap\mathcal{P}_{3}$.

$(1)\Rightarrow(6)$. I need two additional math facts for this. The
first comes from an application of the supporting hyperplane theorem
to characterize the inclusion of convex sets (\citealp{rado1952inequality};
see also \citealp{marshall2009inequalities} Corollary 2.B.3).
\begin{fact}
\label{fact:SHT} For two convex sets $K,K'\subseteq\mathbb{R}^{N}$,
$K\supseteq K'$ if and only if, for any $\boldsymbol{\beta}\in\mathbb{R}^{N}$,
\[
\max_{\boldsymbol{v}\in K}\boldsymbol{\beta}\cdot\boldsymbol{v}\geq\max_{\boldsymbol{v}'\in K'}\boldsymbol{\beta}\cdot\boldsymbol{v}'.
\]
\end{fact}
The second is the majorization inequality due to \citet{karamata1932inegalite}
and \citet{hardy1934inequalities} (see \citealp{marshall2009inequalities}
Theorem 1.A.3 for the majorization vector form, and \citealp{shaked2007stochastic}
Theorem 3.A.1 for the distributional form).
\begin{fact}
\label{fact:majorization-ineq} Let $P$ and $Q$ be two probability
measures on $\mathbb{R}$ with equal means $\mathbb{E}_{P}[X]=\mathbb{E}_{Q}[X]$.
The following are equivalent: (i) for every $\alpha\in\mathbb{R}$,
\[
\mathbb{E}_{P}\left[\max\left\{ 0,X-\alpha\right\} \right]\geq\mathbb{E}_{Q}\left[\max\left\{ 0,X-\alpha\right\} \right].
\]
(ii) for every convex function $\phi:\mathbb{R}\to\mathbb{R}$,
\[
\mathbb{E}_{P}[\phi(X)]\geq\mathbb{E}_{Q}[\phi(X)].
\]
\end{fact}
This is the standard equivalence result for the convex order. For
finitely-supported distributions $P=\sum_{m}w_{m}\delta_{x_{m}}$
and $Q=\sum_{m'}w'_{m'}\delta_{x'_{m'}}$, this specializes to the
equivalence between (i) for every $\alpha\in\mathbb{R}$
\[
\sum_{m=1}^{M}w_{m}\max\{0,x_{m}-\alpha\}\geq\sum_{m'=1}^{M'}w'_{m'}\max\{0,x'_{m'}-\alpha\},
\]
and (ii) for every convex function $\phi:\mathbb{R}\to\mathbb{R}$,
\[
\sum_{m=1}^{M}w_{m}\phi(x_{m})\geq\sum_{m'=1}^{M'}w'_{m'}\phi(x'_{m'}).
\]

I now prove $(1)\Rightarrow(6)$. I first apply Fact \ref{fact:SHT}
to derive the separating hyperplane conditions for the set inclusion
conditions in both $(1)$ and $(6)$. I then apply Fact \ref{fact:SHT}
to show that the condition for $(1)$ implies that for $(6)$. I start
with condition $(1)$. Fix any $\boldsymbol{\mu}\in\Delta\Omega$.
For an experiment $\mathcal{E}$, define the lift zonotope of $\mathcal{E}$
as 
\[
\widehat{\operatorname{Zon}}\,\mathcal{E}:=\left\{ (\boldsymbol{u},\boldsymbol{\mu}\cdot\boldsymbol{u}):\boldsymbol{u}\in\operatorname{Zon}\mathcal{E}\right\} =\left\{ (\mathcal{E}\boldsymbol{v},\boldsymbol{\mu}\cdot\mathcal{E}\boldsymbol{v}):0\leq\boldsymbol{v}\leq1\right\} \subseteq\mathbb{R}^{N+1}.
\]
The lift zonotope pairs the state-dependent utility $\boldsymbol{u}$
with its expected value $\boldsymbol{\mu}\cdot\boldsymbol{u}$. Observe
that $\widehat{\operatorname{Zon}}\,\mathcal{E}$ is convex. The inclusion
$\operatorname{Zon}\mathcal{E}\supseteq\operatorname{Zon}\mathcal{E}'$
implies $\widehat{\operatorname{Zon}}\,\mathcal{E}\supseteq\widehat{\operatorname{Zon}}\,\mathcal{E}'$
because the lift map $\boldsymbol{u}\mapsto(\boldsymbol{u},\boldsymbol{\mu}\cdot\boldsymbol{u})$
does not depend on the experiment itself. 

Apply Fact \ref{fact:SHT} to the inclusion of the lift zonotope,
and we obtain that, for any $(\boldsymbol{\beta},-\alpha)\in\mathbb{R}^{N+1}$,
\begin{equation}
\max_{0\leq\boldsymbol{v}\leq1}\left[\boldsymbol{\beta}\cdot\mathcal{E}\boldsymbol{v}-\alpha\boldsymbol{\mu}\cdot\mathcal{E}\boldsymbol{v}\right]\geq\max_{0\leq\boldsymbol{v}'\leq1}\left[\boldsymbol{\beta}\cdot\mathcal{E}'\boldsymbol{v}'-\alpha\boldsymbol{\mu}\cdot\mathcal{E}'\boldsymbol{v}'\right].\label{eq:lift-zonoid-SHT}
\end{equation}
The maximum on each side of (\ref{eq:lift-zonoid-SHT}) decouples
across coordinates. Let $w_{m}:=(\boldsymbol{\mu}\cdot\mathcal{E})_{m}$
denote the probability of $y_{m}$ under prior $\boldsymbol{\mu}$.\footnote{Here and throughout, I use $\boldsymbol{\mu}\cdot\mathcal{E}$ to
denote the usual product $\boldsymbol{\mu}^{\top}\mathcal{E}$ without
a transpose symbol, with $m$-th entry $(\boldsymbol{\mu}\cdot\mathcal{E})_{m}=\sum_{n}\boldsymbol{\mu}(\omega_{n})\mathcal{E}_{n,m}$.
Similarly, $\boldsymbol{\beta}\cdot\mathcal{E}$.} We can rewrite the above as 
\begin{align*}
\max_{0\leq\boldsymbol{v}\leq1}\left[\boldsymbol{\beta}\cdot\mathcal{E}\boldsymbol{v}-\alpha\boldsymbol{\mu}\cdot\mathcal{E}\boldsymbol{v}\right] & =\sum_{m=1}^{M}\max\bigl\{0,(\boldsymbol{\beta}\cdot\mathcal{E})_{m}-\alpha w_{m}\bigr\},\\
 & =\sum_{m=1}^{M}w_{m}\max\!\left\{ 0,\frac{(\boldsymbol{\beta}\cdot\mathcal{E})_{m}}{w_{m}}-\alpha\right\} ,
\end{align*}
and similarly for $\mathcal{E}'$ with $w'_{m'}:=(\boldsymbol{\mu}\cdot\mathcal{E}')_{m'}$.\footnote{The pulled-out factor $w_{m}$ on the second line is purely expository
to make the connection to Fact \ref{fact:majorization-ineq} easier
to see. The argument does not require division by $w_{m}$, so $w_{m}=0$
is not an issue.} The inequality (\ref{eq:lift-zonoid-SHT}) thus becomes, for any
$(\boldsymbol{\beta},-\alpha)\in\mathbb{R}^{N+1}$, 
\begin{equation}
\sum_{m=1}^{M}w_{m}\max\!\left\{ 0,\frac{(\boldsymbol{\beta}\cdot\mathcal{E})_{m}}{w_{m}}-\alpha\right\} \geq\sum_{m'=1}^{M'}w'_{m'}\max\!\left\{ 0,\frac{(\boldsymbol{\beta}\cdot\mathcal{E}')_{m'}}{w'_{m'}}-\alpha\right\} .\label{eq:lift-zonoid-stop-loss}
\end{equation}

Apply Fact \ref{fact:majorization-ineq}, inequality (\ref{eq:lift-zonoid-stop-loss})
is equivalent to 
\begin{equation}
\sum_{m=1}^{M}w_{m}\phi\left(\frac{(\boldsymbol{\beta}\cdot\mathcal{E})_{m}}{w_{m}}\right)\;\geq\sum_{m'=1}^{M'}w'_{m'}\phi\left(\frac{(\boldsymbol{\beta}\cdot\mathcal{E}')_{m'}}{w'_{m'}}\right).\label{eq:lift-zonoid-cvx}
\end{equation}

Next, I apply Fact \ref{fact:SHT} to the inclusion of $\mathcal{V}_{\boldsymbol{\mu},u,B}$.
To do that, we have to check the convexity of $\mathcal{V}_{\boldsymbol{\mu},u,B}(\mathcal{E})$.
Take any $\boldsymbol{v}_{1},\boldsymbol{v}_{2}$ such that $\mathcal{E}\boldsymbol{v}_{1},\mathcal{E}\boldsymbol{v}_{2}\in\mathcal{V}_{\boldsymbol{\mu},u,B}(\mathcal{E})$.
Let $\boldsymbol{t}_{1},\boldsymbol{t}_{2}$ be the corresponding
contracts that generates $\boldsymbol{v}_{1},\boldsymbol{v}_{2}$
with $\boldsymbol{\mu}\cdot\mathcal{E}\boldsymbol{t}_{1}\leq B$ and
$\boldsymbol{\mu}\cdot\mathcal{E}\boldsymbol{t}_{2}\leq B$. Take
any $\alpha\in[0,1]$. Let $\boldsymbol{t}=\alpha\boldsymbol{t}_{1}+(1-\alpha)\boldsymbol{t}_{2}$
and $\boldsymbol{v}:=u(\boldsymbol{t})$. Since the budget constraint
is linear, we have $\mathcal{E}\boldsymbol{v}\in\mathcal{V}_{\boldsymbol{\mu},u,B}(\mathcal{E})$.
Due to the concavity of $u$, $\boldsymbol{v}\geq\alpha\boldsymbol{v}_{1}+(1-\alpha)\boldsymbol{v}_{2}$.
We can lower the payment to pick some $\boldsymbol{t}'\leq\boldsymbol{t}$
so that $\boldsymbol{v}':=u(\boldsymbol{t}')$ satisfies $\boldsymbol{v}'=\alpha\boldsymbol{v}_{1}+(1-\alpha)\boldsymbol{v}_{2}$.
It suffices to show $\boldsymbol{v}'\in\mathcal{V}_{\boldsymbol{\mu},u,B}(\mathcal{E})$,
which follows from $\boldsymbol{t}'\leq\boldsymbol{t}$ because $\boldsymbol{t}'$
pays less and must also satisfy the budget constraint.

Fact \ref{fact:SHT} translates the inclusion of $\mathcal{V}_{\boldsymbol{\mu},u,B}$
into the following inequality: for every $\boldsymbol{\beta}\in\mathbb{R}^{N}$,
$u\in\mathcal{U}$, and $\boldsymbol{\mu}\in\Delta\Omega$,
\begin{align*}
\max_{\boldsymbol{\mu}\cdot\mathcal{E}\boldsymbol{t}\leq1}\;\boldsymbol{\beta}\cdot\mathcal{E}u(\boldsymbol{t}) & \geq\max_{\boldsymbol{\mu}\cdot\mathcal{E}'\boldsymbol{t}\leq1}\;\boldsymbol{\beta}\cdot\mathcal{E}'u(\boldsymbol{t}),
\end{align*}
and this inequality holds if 
\begin{equation}
\max_{\boldsymbol{t}}\boldsymbol{\beta}\cdot\mathcal{E}u(\boldsymbol{t})-\boldsymbol{\mu}\cdot\mathcal{E}\boldsymbol{t}\geq\max_{\boldsymbol{t}}\boldsymbol{\beta}\cdot\mathcal{E}'u(\boldsymbol{t})-\boldsymbol{\mu}\cdot\mathcal{E}'\boldsymbol{t},\label{eq:SHT-lagrangian}
\end{equation}
where we incorporate the ex ante budget constraint with a multiplier.
Specifically, let $\eta\in\mathbb{R}_{+}$ be the multiplier on the
ex ante budget constraint $\boldsymbol{\mu}\cdot\mathcal{E}\boldsymbol{t}\leq1$.
We must have $\eta>0$ since the budget constraint must bind at the
optimum. Fix any $\left(\boldsymbol{\beta},u,\boldsymbol{\mu}\right)$,
the Lagrangian is $\mathcal{L}(\boldsymbol{t},\eta;\mathcal{E}):=\boldsymbol{\beta}\cdot\mathcal{E}u(\boldsymbol{t})-\eta\boldsymbol{\mu}\cdot\mathcal{E}\boldsymbol{t}+\eta$.
For the inequality to hold, it suffices to have $\mathcal{L}(\boldsymbol{t},\eta;\mathcal{E})\geq\mathcal{L}(\boldsymbol{t},\eta;\mathcal{E}')$
for any $\eta>0$. To obtain (\ref{eq:SHT-lagrangian}), we divide
$\eta>0$ on both sides sides and redefine $\boldsymbol{\beta}$ as
$\boldsymbol{\beta}/\eta$ since $\boldsymbol{\beta}$ can take any
value.

Lastly, observe that the maximum on each side again decouples across
coordinates, and inequality (\ref{eq:SHT-lagrangian}) can be rewritten
as 
\begin{equation}
\sum_{m=1}^{M}w_{m}\max_{t_{m}}\left[\frac{(\boldsymbol{\beta}\cdot\mathcal{E})_{m}}{w_{m}}\cdot u(t_{m})-t_{m}\right]\geq\sum_{m'=1}^{M'}w'_{m'}\max_{t'_{m'}}\left[\frac{(\boldsymbol{\beta}\cdot\mathcal{E}')_{m'}}{w'_{m'}}\,u(t'_{m'})-t'_{m'}\right]\label{eq:V-cvx}
\end{equation}
where the inner maximization is a convex function of $(\boldsymbol{\beta}\cdot\mathcal{E})_{m}/w_{m}$
since it is the maximum over affine functions. As a result, (\ref{eq:lift-zonoid-cvx})
implies (\ref{eq:V-cvx}) by taking $\phi$ to be this specific convex
function. Therefore, $(1)\Rightarrow(6)$.

$(6)\Rightarrow(3a)$. Condition $(6)$ is equivalent to saying that
the expected cost to generate any state-dependent utility $\boldsymbol{u}$
is lower under $\mathcal{E}$ than under $\mathcal{E}'$, regardless
of the state distribution used to evaluate the expectation. To see
this, if $\boldsymbol{u}$ is feasible under $\mathcal{E}'$ at budget
$B$, it is also feasible under $\mathcal{E}$ at the same budget,
so the minimum budget needed to generate $\boldsymbol{u}$ must be
weakly lower under $\mathcal{E}$. That is, the cost to generate $\boldsymbol{u}$
is lower under $\mathcal{E}$. 

From this observation, we can show that $\kappa(\mathcal{E};P)\leq\kappa(\mathcal{E}';P)$
for any $P\in\mathcal{P}_{\emptyset}$. Take any $P\in\mathcal{P}_{\emptyset}$
where the principal wants to implement $a_{0}$. If $a_{0}$ is not
implementable under $\mathcal{E}'$, $\kappa(\mathcal{E}';P)=+\infty$
and the inequality always holds. Otherwise, suppose the optimal contract
under $\mathcal{E}'$ generates state-dependent utility $\boldsymbol{u}$.
From the above observation, the principal can generate the same $\boldsymbol{u}$
under $\mathcal{E}$ at a weakly lower cost due to zonotope dominance.
The agent still optimally chooses $a_{0}$ since his state-dependent
utility is held the same. Therefore, the principal's cost under $\mathcal{E}$
must be weakly lower. 

$(3a)\Rightarrow(2)$. The idea is to take 
\[
u_{0}(t)=\begin{cases}
\min\{t,B\} & \text{if }t\geq0\\
-\infty & \text{if }t<0
\end{cases}
\]
which satiates after receiving a payment of $B$ to mimic the ex post
budget constraint, and drops to $-\infty$ with negative payments
to mimic limited liability. One has to be careful since we require
the utility $u\in\mathcal{U}$ to be strictly increasing and take
values in $\mathbb{R}$. To this end, let $u_{i}(t)=\min\left\{ it,t,\frac{1}{i}(t-B)+B\right\} $
with $i\in\mathbb{N}_{+}$. $u_{i}\in\mathcal{U}$ for all $i\geq1$
since it is strictly increasing, continuous, concave, unbounded, and
$u_{i}(0)=0$. Take any $A\in\mathcal{A}$, $a_{0}\in A$, and $C\in\mathcal{C}$.
Construct problems $P_{i}=\left(A,a_{0},u_{i},C,\emptyset\right)$
and $P_{0}=\left(A,a_{0},u_{0},C,\emptyset\right)$. We can obtain
two sequences of indirect costs $\left\{ \kappa(\mathcal{E};P_{i})\right\} _{i=1}^{\infty}$
and $\left\{ \kappa(\mathcal{E}';P_{i})\right\} _{i=1}^{\infty}$.
$(3a)$ implies $\kappa(\mathcal{E};P_{i})\leq\kappa(\mathcal{E}';P_{i})$
for all $i\geq1$, and we have to show that $\kappa(\mathcal{E};P_{0})\leq\kappa(\mathcal{E}';P_{0})$,
that is, the inequality is preserved after taking $i\to\infty$. 

Notice that the sequence $\left\{ \kappa(\mathcal{E};P_{i})\right\} _{i=1}^{\infty}$
increases in $i$ and $\kappa(\mathcal{E};P_{0})\geq\kappa(\mathcal{E};P_{i})$
for all $i\geq1$. The same holds for $\mathcal{E}'$. This is because
as $i\geq1$ increases, the utility $u_{i}(t)$ at every $t$ decreases
pointwise in $i$, and $u_{0}(t)\leq u_{i}(t)$ for all $i\geq1$.
The pointwise decrease in utilities means, to generate any state-dependent
utility, the principal's cost increases with $i$ and the cost is
the highest for $u_{0}$. As a result, the indirect cost must also
increase in $i$ and it is the highest for $u_{0}$. 

Given the increasingness of the sequences, we can prove the desired
inequality $\kappa(\mathcal{E};P_{0})\leq\kappa(\mathcal{E}';P_{0})$.
If $\kappa(\mathcal{E}';P_{0})=\infty$, the inequality always holds.
If $\kappa(\mathcal{E}';P_{0})<\infty$, we know that $\left\{ \kappa(\mathcal{E}';P_{i})\right\} _{i=1}^{\infty}$
is bounded. To show the inequality, it suffices to prove that 
\begin{equation}
\lim_{i\to\infty}\kappa(\mathcal{E};P_{i})=\kappa(\mathcal{E};P_{0})\label{eq:limit-cost}
\end{equation}
and similarly for $\mathcal{E}'$. I apply the theorem of the maximum
here (Theorem A2.21 in \citealp{jehle2001advanced}). It suffices
to check that for any $\boldsymbol{t}$ feasible in $\kappa(\mathcal{E};P_{0})$,
there exists a sequence of contracts $\boldsymbol{t}_{i}\to\boldsymbol{t}$
such that $\boldsymbol{t}_{i}$ is feasible in $\kappa(\mathcal{E};P_{i})$
for every $i$. This is obvious since we can always pick $\boldsymbol{t}_{i}=\boldsymbol{t}$.
To see that $\boldsymbol{t}$ is feasible in $\kappa(\mathcal{E};P_{i})$
for every $i$, note that $0\leq\boldsymbol{t}\leq B$ because $\boldsymbol{t}$
is feasible in $\kappa(\mathcal{E};P_{0})$, and utilities $u_{i}$
equal $u_{0}$ on $[0,B]$. This means $\boldsymbol{t}_{i}=\boldsymbol{t}$
generates the same state-dependent utility in $\kappa(\mathcal{E};P_{i})$.
Since the agent's incentive and participation constraints only concern
the state-dependent utility he obtains, when $\boldsymbol{t}$ satisfies
these constraints in $\kappa(\mathcal{E};P_{0})$, it should also
satisfy these constraints in $\kappa(\mathcal{E};P_{i})$. Therefore,
we can apply the theorem of maximum to conclude (\ref{eq:limit-cost})
and likewise for $\mathcal{E}'$, which completes the proof since
$\kappa(\mathcal{E};P_{i})\leq\kappa(\mathcal{E}';P_{i})$ for all
$i\geq1$. 

$(2)\Rightarrow(1)$. Suppose $\operatorname{Zon}\mathcal{E}\not\supseteq\operatorname{Zon}\mathcal{E}'$,
then there exists some $\boldsymbol{x}\in\operatorname{Zon}\mathcal{E}'\setminus\operatorname{Zon}\mathcal{E}$.
I construct a problem in $\mathcal{P}_{F}\cap\mathcal{P}_{3}$ to
violate the cost comparisons. Focus on cost functions $C\in\mathcal{C}_{F}$
such that $\lim_{\boldsymbol{\mu}\to\boldsymbol{\mu}'}\left\Vert \boldsymbol{\nabla}C(\boldsymbol{\mu})\right\Vert =+\infty$
for all $\boldsymbol{\mu}'\in\partial\Delta\Omega$. This condition
ensures that $\boldsymbol{\eta}$ does not show up in the first-order
condition. Pick a cost function $C$ and action $\boldsymbol{\mu}_{0}$
so that the marginal cost of $C$ at $\boldsymbol{\mu}_{0}$ is $\boldsymbol{x}$.
Here I identify the target action with its induced state distribution
$\boldsymbol{\mu}_{0}\in\Delta\Omega$. To optimally implement $\boldsymbol{\mu}_{0}$
under $\mathcal{E}'$, from Lemma \ref{lem:A-opt-interior}, the principal
chooses an optimal state-dependent utility $\boldsymbol{x}+\lambda\boldsymbol{1}\in\operatorname{Zon}\mathcal{E}'$
for some $\lambda\in\mathbb{R}$ so that it is on the boundary of
$\operatorname{Zon}\mathcal{E}'$. If $\boldsymbol{\mu}_{0}$ is not
implementable under $\mathcal{E}$, then we are done. Otherwise, the
principal optimally implements it at some $\boldsymbol{x}+\lambda'\boldsymbol{1}\in\operatorname{Zon}\mathcal{E}$
on the boundary of $\operatorname{Zon}\mathcal{E}$. If the cost is
lower under $\boldsymbol{x}+\lambda\boldsymbol{1}$, we are done.
Otherwise, since the zonotope is centrally symmetric, we can find
a symmetric $\boldsymbol{x}'$ so that the cost is lower under $\mathcal{E}'$
than $\mathcal{E}$. 

Next, I prove the equivalence $(1)\Leftrightarrow(3b),(3c),(3d)$.
In the proof for $(1)\Leftrightarrow(3a)$, nothing changes if I require
$\boldsymbol{t}\geq0$, $\boldsymbol{t}\leq1$, or $0\leq\boldsymbol{t}\leq1$.
In fact, the restriction on $\boldsymbol{t}$ does not play a role.
To make things concrete, I formally define the following sets of feasible
state-dependent utilities with an ex ante budget, but now imposing
additional constraints, 
\begin{align*}
\mathcal{V}_{\boldsymbol{\mu},u,B}^{\text{LL}}(\mathcal{E}) & :=\left\{ \mathcal{E}\boldsymbol{v}:\exists\boldsymbol{t}\in\mathbb{R}^{M}\text{ such that }\boldsymbol{v}=u(\boldsymbol{t}),\boldsymbol{t}\geq0,\boldsymbol{\mu}\cdot\mathcal{E}\boldsymbol{t}\leq B\right\} ,\\
\mathcal{V}_{\boldsymbol{\mu},u,B}^{\text{B}}(\mathcal{E}) & :=\left\{ \mathcal{E}\boldsymbol{v}:\exists\boldsymbol{t}\in\mathbb{R}^{M}\text{ such that }\boldsymbol{v}=u(\boldsymbol{t}),\boldsymbol{t}\leq1,\boldsymbol{\mu}\cdot\mathcal{E}\boldsymbol{t}\leq B\right\} ,\\
\mathcal{V}_{\boldsymbol{\mu},u,B}^{\text{LL},\text{B}}(\mathcal{E}) & :=\left\{ \mathcal{E}\boldsymbol{v}:\exists\boldsymbol{t}\in\mathbb{R}^{M}\text{ such that }\boldsymbol{v}=u(\boldsymbol{t}),0\leq\boldsymbol{t}\leq1,\boldsymbol{\mu}\cdot\mathcal{E}\boldsymbol{t}\leq B\right\} ,
\end{align*}
and similarly the following set inclusion conditions, 
\begin{enumerate}
\item [$(6b)$] $\mathcal{V}_{\boldsymbol{\mu},u,B}^{\text{LL}}(\mathcal{E})\supseteq\mathcal{V}_{\boldsymbol{\mu},u,B}^{\text{LL}}(\mathcal{E}')$
for any $\boldsymbol{\mu}\in\Delta\Omega$, $u\in\mathcal{U}$, and
$B>0$,
\item [$(6c)$] $\mathcal{V}_{\boldsymbol{\mu},u,B}^{\text{B}}(\mathcal{E})\supseteq\mathcal{V}_{\boldsymbol{\mu},u,B}^{\text{B}}(\mathcal{E}')$
for any $\boldsymbol{\mu}\in\Delta\Omega$, $u\in\mathcal{U}$, and
$B>0$,
\item [$(6d)$] $\mathcal{V}_{\boldsymbol{\mu},u,B}^{\text{LL},\text{B}}(\mathcal{E})\supseteq\mathcal{V}_{\boldsymbol{\mu},u,B}^{\text{LL},\text{B}}(\mathcal{E}')$
for any $\boldsymbol{\mu}\in\Delta\Omega$, $u\in\mathcal{U}$, and
$B>0$.
\end{enumerate}
I can use the exact same method for $(1)\Rightarrow(6)\Rightarrow(3a)$
to show $(1)\Rightarrow(6b),(6c),(6d)$, and $(6b),(6c),(6d)\Rightarrow(3b),(3c),(3d)$,
respectively. The implications $(3b),(3c),(3d)\Rightarrow(2)$ can
also reuse the proof for $(3a)\Rightarrow(2)$ and it is even easier
as the limited liability and the ex post budget constraints are built
into the statements of $(3b),(3c),(3d)$.

Lastly, I show the equivalence $(1)\Leftrightarrow(3)$. We have $(1)\Rightarrow(3)$
because $(1)\Rightarrow(3a),(3b),(3c),(3d)$, and conditions $(3a),(3b),(3c),(3d)$
collectively imply $(3)$. We have $(3)\Rightarrow(1)$ because $(3)$
automatically implies $(3a),(3b),(3c),(3d)$, any one of which implies
$(1)$. 
\end{proof}
\begin{proof}
[Proof of Proposition \ref{prop:full-rank}] The first part follows
from Theorem 2 in \citet{wu2023geometric} and the fact that $\mathcal{E}\geq_{\text{Zon}}\mathcal{E}'$
implies $\mathcal{E}\geq_{\text{Cone}}\mathcal{E}'$, which is equivalent
to $\operatorname{Co}\operatorname{Supp}\pi\left(\mathcal{E};\boldsymbol{\nu}\right)\supseteq\operatorname{Co}\operatorname{Supp}\pi\left(\mathcal{E}';\boldsymbol{\nu}\right)$
for any prior $\boldsymbol{\nu}\in\Delta\Omega$ by Theorem \ref{thm:conic-span}
in the main text. For the second part, suppose $\mathcal{E}\geq_{\text{Cone}}\mathcal{E}'$
and $\mathcal{E}$ has full column rank. $\mathcal{E}'=\mathcal{E}G$
for some $G\geq0$. Due to row stochasticity of $\mathcal{E}$ and
$\mathcal{E}'$, we have $\mathcal{E}\boldsymbol{1}=1=\mathcal{E}'\boldsymbol{1}=\mathcal{E}G\boldsymbol{1}$,
which implies $\mathcal{E}\left(G\boldsymbol{1}-\boldsymbol{1}\right)=\boldsymbol{0}$.
Since $\mathcal{E}$ has full column rank, this necessarily implies
$G\boldsymbol{1}=\boldsymbol{1}$.
\end{proof}

\section{Relations between the Orders: Examples \label{appsec:example}}

This appendix collects several examples on the relations between the
orders introduced in the main text. First, I reproduce the example
in \citet{bertschinger2014blackwell} that shows the zonotope order
is strictly implied by Blackwell. Note that such examples only exist
when $N>2$ since with binary states, Proposition \ref{prop:binary-states}
says the two orders coincide. 
\begin{example}
\label{exa:zonotope-Blackwell} Suppose $N=3$. Consider the following
experiments. 
\[
\mathcal{E}_{1}=\begin{bmatrix}0.5 & 0 & 0 & 0.5\\
0 & 0.5 & 0 & 0.5\\
0 & 0 & 0.5 & 0.5
\end{bmatrix};\;\mathcal{E}_{2}=\begin{bmatrix}0.5 & 0.5 & 0\\
0.5 & 0 & 0.5\\
0 & 0.5 & 0.5
\end{bmatrix}.
\]
We have $\mathcal{E}_{1}\not\geq_{\text{B}}\mathcal{E}_{2}$, but
$\mathcal{E}_{1}\geq_{\text{Zon}}\mathcal{E}_{2}$. To see $\mathcal{E}_{1}\not\geq_{\text{B}}\mathcal{E}_{2}$,
one way is to realize that the unique $G_{1}\geq0$ such that $\mathcal{E}_{2}=\mathcal{E}_{1}G_{1}$
is 
\[
G_{1}=\begin{bmatrix}1 & 1 & 0\\
1 & 0 & 1\\
0 & 1 & 1\\
0 & 0 & 0
\end{bmatrix},
\]
which is not row stochastic.\footnote{There are other $G$'s that satisfy $\mathcal{E}_{2}=\mathcal{E}_{1}G$
but those $G$'s are not weakly positive.} Alternatively, here is a decision problem where $\mathcal{E}_{2}$
outperforms $\mathcal{E}_{1}$. Suppose the decision maker has three
actions $A=\left\{ a_{1},a_{2},a_{3}\right\} $ and she gets a positive
utility when her action mismatches the state $\omega\in\Omega=\left\{ \omega_{1},\omega_{2},\omega_{3}\right\} $
and a negative utility otherwise. $\mathcal{E}_{2}$ allows her to
always mismatch the action and the state while $\mathcal{E}_{1}$
does not. 

To see that $\mathcal{E}_{1}\geq_{\text{Zon}}\mathcal{E}_{2}$, we
use the equivalence $(1)\Leftrightarrow(4)$ in Theorem \ref{thm:zonotope}.
We have to show that any partial sum of the columns in $\mathcal{E}_{2}$
lies in $\operatorname{Zon}\mathcal{E}_{1}$. It is obvious that each
column in $\mathcal{E}_{2}$ lies in $\operatorname{Zon}\mathcal{E}_{1}$.
So is any sum of two columns. For this, it suffices to observe that
the sum of the first two columns in $\mathcal{E}_{2}$ equals the
sum of the first and the last columns in $\mathcal{E}_{1}$. Other
cases are symmetric. The sum of all three columns in $\mathcal{E}_{2}$
is a vector of ones, which obviously lies in $\operatorname{Zon}\mathcal{E}_{1}$.
\end{example}
The same example also shows that the full rank conditions in Proposition
\ref{prop:full-rank} cannot be dispensed with. First, without full
rank, the zonotope and the Blackwell orders can differ. Example \ref{exa:zonotope-Blackwell}
shows that when $\mathcal{E}$ does not have full column rank, $\mathcal{E}\geq_{\text{Zon}}\mathcal{E}'$
does not imply $\mathcal{E}\geq_{\text{B}}\mathcal{E}'$. Second,
when $\mathcal{E}$ does not have full column rank, $\mathcal{E}\geq_{\text{Cone}}\mathcal{E}'$
does not imply $\mathcal{E}\geq_{\text{B}}\mathcal{E}'$. For this,
I can again reuse Example \ref{exa:zonotope-Blackwell}, since $\mathcal{E}_{1}\geq_{\text{Cone}}\mathcal{E}_{2}$
in Example \ref{exa:zonotope-Blackwell} as the extremal beliefs of
$\mathcal{E}_{1}$ reveal the state while $\mathcal{E}_{2}$ does
not. 

Lastly, I show that $\mathcal{E}'=\mathcal{E}G$ for some $0\leq G\leq1$
is not sufficient for $\mathcal{E}\geq_{\text{Zon}}\mathcal{E}'$.
I will use Example \ref{exa:main-example} again. There, $\mathcal{E}_{3}=\mathcal{E}_{2}G$
where 
\[
G=\begin{bmatrix}1 & 1/3 & 0 & 0\\
0 & 1/3 & 1/3 & 0\\
0 & 0 & 1/3 & 1
\end{bmatrix}
\]
 We have $0\leq G\leq1$ but $\mathcal{E}_{2}\not\geq_{\text{Zon}}\mathcal{E}_{3}$.

\section{Profit Maximization \label{appsec:profit-max}}

This appendix extends the results in Section \ref{sec:orders} of
the main text from cost minimization to profit maximization. I show
that the same orders also characterize profit comparisons.

The profit maximization problem differs from cost minimization in
one respect: the target action is endogenously chosen by the principal
given her preference. Following \citet{grossman1983implicit}, I decompose
every moral hazard problem into an inner problem to minimize the cost
of implementing every target action and an outer problem to optimize
over the target action. Recall that a cost minimization problem $P=(A,a_{0},u,C,\mathcal{R})$
is specified by an action space $A\in\mathcal{A}$, a target action
$a_{0}\in A$, the agent's utility from money $u\in\mathcal{U}$,
the cost function $C\in\mathcal{C}$, and constraints $\mathcal{R}$.
For profit maximization, the target action $a_{0}$ is endogenously
chosen by the principal, given her preference over the agent's action.
Let $V:A\to\mathbb{R}$ denote the principal's value from the agent's
action. A profit maximization problem is a tuple $Q:=(A,V,u,C,\mathcal{R})$,
where the principal solves
\begin{align}
\max_{\boldsymbol{t},a\in A} & \;V\left(a\right)-\mathbb{E}_{y\sim\mathcal{E}\left(\cdot\mid\boldsymbol{\mu}_{a}\right)}\left[\boldsymbol{t}(y)\right],\text{ s.t. }\text{(\ref{eq:IC}), (\ref{eq:IR})},\mathcal{R}.\tag{Q}\label{eq:P-profit-max}
\end{align}
Let $W(\mathcal{E};Q)$ denote the value of this problem. 

The problem classes can be defined analogously. Let $\mathcal{Q}$
be the set of all profit maximization problems with no restrictions
on the principal's preference, $\mathcal{Q}_{1}:=\left\{ Q\in\mathcal{Q}:u(t)=t,\mathcal{R}=\emptyset\right\} $
the class of risk neutral problems with no contract constraint, $\mathcal{Q}_{2}:=\{Q\in\mathcal{Q}:u(t)=t,\mathcal{R}=\text{LL}\}$
the class of risk neutral problems with limited liability, and $\mathcal{Q}_{3}:=\{Q\in\mathcal{Q}:u(t)=t,\mathcal{R}=\text{LL},\text{B}\}$
the class with an additional ex post budget constraint.

The next proposition says that the same orders characterize profit
comparisons as well. The key observation is that, since the principal's
preference over actions is unrestricted, profit comparisons encode
cost comparisons at every target action.
\begin{prop}
For any experiments $\mathcal{E}$ and $\mathcal{E}'$,
\begin{enumerate}
\item [$(1)$] $\mathcal{E}\geq_{\text{Col}}\mathcal{E}'$ if and only if
$W(\mathcal{E};Q)\geq W(\mathcal{E}';Q)$ for any $Q\in\mathcal{Q}_{1}$,
\item [$(2)$] $\mathcal{E}\geq_{\text{Cone}}\mathcal{E}'$ if and only
if $W(\mathcal{E};Q)\geq W(\mathcal{E}';Q)$ for any $Q\in\mathcal{Q}_{2}$,
\item [$(3)$] $\mathcal{E}\geq_{\text{Zon}}\mathcal{E}'$ if and only if
$W(\mathcal{E};Q)\geq W(\mathcal{E}';Q)$ for any $Q\in\mathcal{Q}_{3}$,
\item [$(4)$] $\mathcal{E}\geq_{\text{Zon}}\mathcal{E}'$ if and only if
$W(\mathcal{E};Q)\geq W(\mathcal{E}';Q)$ for any $Q\in\mathcal{Q}$.
\end{enumerate}
\end{prop}
\begin{proof}
For every equivalence, the sufficiency of the orders is immediate.
Suppose $\kappa(\mathcal{E};P)\leq\kappa(\mathcal{E}';P)$ for any
$P\in\mathcal{P}'$ with $\mathcal{P}'\in\left\{ \mathcal{P}_{1},\mathcal{P}_{2},\mathcal{P}_{3},\mathcal{P}\right\} $
being the corresponding classes of cost minimization problems. Take
any problem $P\in\mathcal{P}'$ and consider its profit maximization
counterpart $Q$ with the same $\left(A,u,C,\mathcal{R}\right)$ but
an arbitrary value function $V$. Regardless of which target action
is optimal for the principal under $\mathcal{E}'$, she can implement
the same action at a lower cost under $\mathcal{E}$, and may do even
better by implementing the optimal action under $\mathcal{E}$. Therefore,
her profit must be weakly higher. 

For the necessity part, suppose by contrapositive that the cost comparison
is reversed at some problem $P$. Let $a_{0}$ be the target action
in $P$. Let $Q=\left(A,V,u,C,\mathcal{R}\right)$ be the profit maximization
problem with the same $\left(A,u,C,\mathcal{R}\right)$ as in $P$,
and the value function $V$ assigns $V(a_{0})=V_{0}$ at $a_{0}$,
and $V(a)=0$ for all other action $a\in A$. As long as $V_{0}>\kappa(\mathcal{E}';P)$,
the principal optimally implements $a_{0}$ under $\mathcal{E}'$
in problem $Q$ with profit $V_{0}-\kappa(\mathcal{E}';P)>0$. Under
$\mathcal{E}$, the principal either implements a different action
$a$ yielding a non-positive profit since $V(a)=0$, or she implements
the same $a_{0}$ with the same value $V_{0}$ but a strictly higher
cost. Either way, her profit is strictly lower than under $\mathcal{E}'$,
contradicting the profit comparison.
\end{proof}

\section{General Utility from Money \label{appsec:non-concave}}

This appendix extends the analysis to general preferences over money
for both the principal and the agent. I show that the problem can
be reduced to one where the principal is risk neutral and an agent
has a general preference over money. The reduction is essentially
a change of units: re-denominating payments in units of the principal's
disutility. The results in Section \ref{sec:orders} continue to hold.
In particular, the implementability comparisons by the column space
order in Theorem \ref{thm:column-space} and the cost comparisons
by the zonotope order in Theorem \ref{thm:zonotope} hold for the
more general class of problems, and the cost comparisons by the conic
span order in Theorem \ref{thm:conic-span} hold for the class of
problems where the principal and the agent have the same risk attitude
towards money. 

Suppose the principal has a disutility function $v:\mathbb{R}\to\mathbb{R}$
over payments and the agent has additively separable payoff $u(t)-C(a)$
with utility $u:\mathbb{R}\to\mathbb{R}$. Both $u$ and $v$ are
continuous, strictly increasing, unbounded, and normalized so that
$u(0)=v(0)=0$. Let $\tilde{\mathcal{U}}$ be the set of all (dis)utility
functions that satisfy these assumptions. The main text focuses on
a special case where $v(t)=t$ is linear and $u$ is concave. 

The principal may now want to offer a random contract since utilities
may not be concave. A random contract $\tilde{\boldsymbol{t}}:Y\to\Delta\mathbb{R}$
maps realizations of $\mathcal{E}$ to a lottery of payments. 

Instead of working with payments directly, I work with the disutilities
to the principal. Given any random contract $\tilde{\boldsymbol{t}}$
under $\mathcal{E}$, I can rewrite it in units of disutility as some
random contract $\tilde{\boldsymbol{v}}:Y\to\Delta\mathbb{R}$ that
maps realizations of $\mathcal{E}$ to a lottery of disutilities to
the principal. The principal then pays the agent the amount that corresponds
to the disutility level drawn from the lottery. Let $\phi:=u\circ v^{-1}$
be the function that translates the principal's disutility level to
the corresponding utility from money for the agent. 

The random contract still affects the agent's incentives only through
the state-dependent utility. The agent solves the same problem as
before, except that the contract can be random and payments are denominated
in units of the principal's disutility. Given contract $\tilde{\boldsymbol{v}}$
under experiment $\mathcal{E}$, the agent's utility is 
\begin{equation}
\widetilde{U}\left(a;\mathcal{E},\tilde{\boldsymbol{v}}\right):=\mathbb{E}_{y\sim\mathcal{E}\left(\cdot\mid\boldsymbol{\mu}_{a}\right)}\left[\mathbb{E}_{w\sim\tilde{\boldsymbol{v}}\left(y\right)}\left[\phi\left(w\right)\right]\right]-C\left(a\right)=\boldsymbol{\mu}_{a}\cdot\boldsymbol{u}-C(a)\label{eq:A-payoff-random-t}
\end{equation}
from action $a$, where the state-dependent utility $\boldsymbol{u}$
is given by $\boldsymbol{u}:=\mathcal{E}\boldsymbol{v}$ with 
\[
\boldsymbol{v}:=\left(\mathbb{E}_{w\sim\tilde{\boldsymbol{v}}\left(y_{m}\right)}\left[\phi\left(w\right)\right]\right)_{m=1}^{M}
\]
being the vector of expected utilities from realizations of $\mathcal{E}$.
The random contract $\tilde{\boldsymbol{v}}$ affects the agent's
payoff and hence his action only through $\boldsymbol{u}$.

To implement some $a_{0}\in A$ with induced state distribution $\boldsymbol{\mu}_{0}:=\boldsymbol{\mu}_{a_{0}}$,
the principal's problem is to minimize her expected disutility from
payments by designing a random contract $\tilde{\boldsymbol{v}}$,
\begin{align}
\min_{\tilde{\boldsymbol{v}}} & \;\mathbb{E}_{y\sim\mathcal{E}\left(\cdot\mid\boldsymbol{\mu}_{0}\right)}\left[\mathbb{E}_{w\sim\tilde{\boldsymbol{v}}\left(y\right)}\left[w\right]\right],\tag{\ensuremath{\widetilde{\text{P}}}}\label{eq:P-problem-random-t}
\end{align}
subject to the same incentive and participation constraints, now stated
for random contracts, 
\begin{align}
a_{0} & \in\underset{a\in A}{\operatorname{argmax}}\:\widetilde{U}\left(a;\mathcal{E},\tilde{\boldsymbol{v}}\right),\tag{\ensuremath{\widetilde{\text{IC}}}}\label{eq:IC-random-t}\\
\widetilde{U}\left(a_{0};\mathcal{E},\tilde{\boldsymbol{v}}\right) & \geq\underline{u},\tag{\ensuremath{\widetilde{\text{IR}}}}\label{eq:IR-random-t}
\end{align}
as well as the optional limited liability and ex post budget constraints
$\widetilde{\mathcal{R}}\subseteq\{\widetilde{\text{LL}},\widetilde{\text{B}}\}$,
\begin{align}
\tilde{\boldsymbol{v}}\left(y\right) & \geq0\text{ a.s.},\forall y\in Y,\tag{\ensuremath{\widetilde{\text{LL}}}}\label{eq:LL-random-t}\\
\tilde{\boldsymbol{v}}\left(y\right) & \leq B\text{ a.s.},\forall y\in Y,\tag{\ensuremath{\widetilde{\text{B}}}}\label{eq:B-random-t}
\end{align}
The problem is therefore reduced to one where the principal is risk
neutral and the agent has a general utility from money $\phi:=u\circ v^{-1}$. 

A cost minimization problem is now described by a tuple $\widetilde{P}:=\left(A,a_{0},u,v,C,\widetilde{\mathcal{R}}\right)$.
Let $\tilde{\kappa}\left(\mathcal{E};\widetilde{P}\right)$ denote
the value of $\widetilde{P}$ under experiment $\mathcal{E}$. I define
the classes of problems analogously. Let $\widetilde{\mathcal{I}}(\mathcal{E};A,u,v,C)$
be the set of actions implementable under $\mathcal{E}$. Let $\widetilde{\mathcal{P}}$
be the class of all moral hazard problems. The class $\mathcal{P}$
in the main text is a subset of $\widetilde{\mathcal{P}}$ with $v(t)=t$
and $u$ weakly concave. Let $\widetilde{\mathcal{P}}_{1}:=\{\widetilde{P}\in\widetilde{\mathcal{P}}:u\circ v^{-1}\text{ is affine},\widetilde{\mathcal{R}}=\emptyset\}$
be the class of problems where the principal and the agent has the
same risk attitude towards money,\footnote{Suppose $u\circ v^{-1}(s)=\alpha s+\beta$ for some constants $\alpha>0,\beta\in\mathbb{R}$.
For any $s$, we can find some $t$ so that $s=v(t)$. Plug in to
the equation, we have $u(t)=\alpha v(t)+\beta$. That is, $u$ is
a positive affine transformation of $v$, and they must represent
the same expected utility preference. Note that the normalization
$u(0)=v(0)=0$ forces the constant $\beta=0$. } and $\widetilde{\mathcal{P}}_{2}:=\{\widetilde{P}\in\widetilde{\mathcal{P}}:u\circ v^{-1}\text{ is affine},\widetilde{\mathcal{R}}=\widetilde{\text{LL}}\}$
and $\widetilde{\mathcal{P}}_{3}:=\{\widetilde{P}\in\widetilde{\mathcal{P}}:u\circ v^{-1}\text{ is affine},\widetilde{\mathcal{R}}=\widetilde{\text{LL}},\widetilde{\text{B}}\}$
be the classes with additional constraints. 

The orders in Section \ref{sec:orders} also apply to the more general
classes of problems introduced above. The next proposition summarizes
the results. While all other results carry through with a straightforward
argument, the zonotope results warrant more discussion. The proof
of Theorem \ref{thm:zonotope} uses the concavity of utility exactly
once: to establish the convexity of the feasible state-dependent utility
set $\mathcal{V}_{\boldsymbol{\mu},u,B}(\mathcal{E})$ so that I can
apply the supporting-hyperplane-style result of Fact \ref{fact:SHT}.
Without concavity, this set is no longer convex under deterministic
contracts. However, convexity can be restored with random contracts,
and the same argument still goes through. 
\begin{prop}
For any experiments $\mathcal{E}$ and $\mathcal{E}'$,
\begin{enumerate}
\item [$(1)$] $\mathcal{E}\geq_{\text{Col}}\mathcal{E}'$ if and only if
$\widetilde{\mathcal{I}}(\mathcal{E}';A,u,v,C)\subseteq\widetilde{\mathcal{I}}(\mathcal{E};A,u,v,C)$
for any $A\in\mathcal{A}$, $u,v\in\tilde{\mathcal{U}},C\in\mathcal{C}$.
\item [$(2)$] $\mathcal{E}\geq_{\text{Col}}\mathcal{E}'$ if and only if
$\tilde{\kappa}\left(\mathcal{E};\widetilde{P}\right)\leq\tilde{\kappa}\left(\mathcal{E}';\widetilde{P}\right)$
for any $\widetilde{P}\in\widetilde{\mathcal{P}}_{1}$,
\item [$(3)$] $\mathcal{E}\geq_{\text{Cone}}\mathcal{E}'$ if and only
if $\tilde{\kappa}\left(\mathcal{E};\widetilde{P}\right)\leq\tilde{\kappa}\left(\mathcal{E}';\widetilde{P}\right)$
for any $\widetilde{P}\in\widetilde{\mathcal{P}}_{2}$,
\item [$(4)$] $\mathcal{E}\geq_{\text{Zon}}\mathcal{E}'$ if and only if
$\tilde{\kappa}\left(\mathcal{E};\widetilde{P}\right)\leq\tilde{\kappa}\left(\mathcal{E}';\widetilde{P}\right)$
for any $\widetilde{P}\in\widetilde{\mathcal{P}}_{3}$,
\item [$(5)$] $\mathcal{E}\geq_{\text{Zon}}\mathcal{E}'$ if and only if
$\tilde{\kappa}\left(\mathcal{E};\widetilde{P}\right)\leq\tilde{\kappa}\left(\mathcal{E}';\widetilde{P}\right)$
for any $\widetilde{P}\in\widetilde{\mathcal{P}}$.
\end{enumerate}
\end{prop}
\begin{proof}
The first three equivalences follow directly from Theorems \ref{thm:column-space}
and \ref{thm:conic-span}. The sufficiency of the orders comes from
the same larger-feasible-set argument. The necessity of the orders
comes from the same constructive argument using flexible problems
with $v(t)=t$ and $u(t)$ weakly concave. 

The last two equivalences require more care. Slightly abusing notation,
denote 
\begin{align*}
\mathbb{E}\left[\phi\left(\tilde{\boldsymbol{v}}\right)\right] & :=\left(\mathbb{E}_{w\sim\tilde{\boldsymbol{v}}\left(y_{m}\right)}\left[\phi\left(w\right)\right]\right)_{m=1}^{M},\\
\mathbb{E}\left[\tilde{\boldsymbol{v}}\right] & :=\left(\mathbb{E}_{w\sim\tilde{\boldsymbol{v}}\left(y_{m}\right)}\left[w\right]\right)_{m=1}^{M}.
\end{align*}
Following the proof for Theorem \ref{thm:zonotope}, define the set
of feasible state dependent utilities given experiment $\mathcal{E}$,
an ex ante budget $B$ in disutility units, the state distribution
$\boldsymbol{\mu}$ to evaluate expectation, and the transformed utility
$\phi$, 
\[
\widetilde{\mathcal{V}}_{\boldsymbol{\mu},\phi,B}(\mathcal{E}):=\left\{ \mathcal{E}\boldsymbol{v}:\exists\tilde{\boldsymbol{v}}\text{ such that }\boldsymbol{v}=\mathbb{E}\left[\phi\left(\tilde{\boldsymbol{v}}\right)\right],\boldsymbol{\mu}\cdot\mathcal{E}\mathbb{E}\left[\tilde{\boldsymbol{v}}\right]\leq B\right\} .
\]

I show the chain of implications: $\mathcal{E}\geq_{\text{Zon}}\mathcal{E}'$
$\Rightarrow$ the inclusion of $\widetilde{\mathcal{V}}_{\boldsymbol{\mu},\phi,B}$
$\Rightarrow$ cost comparisons in class $\widetilde{\mathcal{P}}$
$\Rightarrow$ cost comparisons in class $\widetilde{\mathcal{P}}_{3}$
$\Rightarrow$ $\mathcal{E}\geq_{\text{Zon}}\mathcal{E}'$. All implications
can be proved with the same steps as in the proof of Theorem \ref{thm:zonotope}:
the last three implications follow literally, and the first implication
follows as long as $\widetilde{\mathcal{V}}_{\boldsymbol{\mu},\phi,B}(\mathcal{E})$
is convex. The convexity is easy to see. Take any $\boldsymbol{u}_{1},\boldsymbol{u}_{2}\in\widetilde{\mathcal{V}}_{\boldsymbol{\mu},\phi,B}(\mathcal{E})$.
There exists contracts $\tilde{\boldsymbol{v}}_{i},i=1,2$ that produce
$\boldsymbol{u}_{1},\boldsymbol{u}_{2}$. That is, for $i=1,2$, we
have $\boldsymbol{u}_{i}=\mathcal{E}\boldsymbol{v}_{i}$ with $\boldsymbol{v}_{i}=\mathbb{E}\left[\phi\left(\tilde{\boldsymbol{v}}_{i}\right)\right]$
and $\boldsymbol{\mu}\cdot\mathcal{E}\mathbb{E}\left[\tilde{\boldsymbol{v}}_{i}\right]\leq B$.
Take any $\alpha\in[0,1]$. I have to show that $\boldsymbol{u}:=\alpha\boldsymbol{u}_{1}+(1-\alpha)\boldsymbol{u}_{2}$
also lies in $\widetilde{\mathcal{V}}_{\boldsymbol{\mu},\phi,B}(\mathcal{E})$.
Let $\tilde{\boldsymbol{v}}=\alpha\tilde{\boldsymbol{v}}_{1}+(1-\alpha)\tilde{\boldsymbol{v}}_{2}$
be the random contract that pays according to $\tilde{\boldsymbol{v}}_{1}$
with probability $\alpha$ and $\tilde{\boldsymbol{v}}_{2}$ with
probability $1-\alpha$. Due to the linearity of expectation, we automatically
have 
\[
\boldsymbol{u}=\alpha\mathcal{E}\boldsymbol{v}_{1}+(1-\alpha)\mathcal{E}\boldsymbol{v}_{2}=\mathcal{E}\mathbb{E}\left[\phi\left(\tilde{\boldsymbol{v}}\right)\right],
\]
and that 
\[
\boldsymbol{\mu}\cdot\mathcal{E}\mathbb{E}\left[\tilde{\boldsymbol{v}}\right]=\alpha\boldsymbol{\mu}\cdot\mathcal{E}\mathbb{E}\left[\tilde{\boldsymbol{v}}_{1}\right]+(1-\alpha)\boldsymbol{\mu}\cdot\mathcal{E}\mathbb{E}\left[\tilde{\boldsymbol{v}}_{2}\right]\leq B.
\]
This completes the proof. 
\end{proof}

\section{Non-Separable Preferences \label{appsec:non-separable}}

This appendix relaxes the additive separability assumption on the
agent's preference. Results in Section \ref{sec:orders} and Appendix
\ref{appsec:non-concave} continue to hold under a more general form
of additive-multiplicative payoff. Most notably, this includes the
case of exponential preferences. 

The only requirement is that the agent's risk attitude over money
does not depend on the action taken. That is, if the agent prefers
one lottery over money to another when taking some action, he ranks
them the same way when taking any other action. This is the risk independence
property introduced by \citet{keeney1973risk}. He shows that this
is equivalent to an additive-multiplicative representation, which
is later used in \citet{grossman1983implicit}: the agent's payoff
from taking action $a$ and receiving payment $t$ takes the form
\[
H(a)+K(a)u(t),
\]
where $u$ is the agent's utility from money, $H:A\to\mathbb{R}\cup\left\{ -\infty\right\} $
is analogous to the cost function with negative infinity encoding
infeasibilities, and $K:A\to\mathbb{R}_{++}$ is a strictly positive
multiplicative weight on the utility from money. I impose no further
assumptions on $H$, $K$, and $u$. Exponential preferences of the
form $-\exp\left[-\gamma\left(t-C(a)\right)\right]$ with $\gamma>0$
are a special case with $H(a)=0$, $K(a)=\exp\left[\gamma C(a)\right]$,
and $u(t)=-\exp\left(-\gamma t\right)$. The separable case in Section
\ref{sec:MH-problems} corresponds to $H(a)=-C(a)$ and $K(a)=1$. 

The additive-multiplicative form preserves the key property of the
baseline that the agent's problem is only parametrized by his state-dependent
utility. Given experiment $\mathcal{E}$ and contract $\boldsymbol{t}$,
the agent's expected payoff from action $a$ is 
\[
H(a)+K(a)\boldsymbol{\mu}_{a}\cdot\boldsymbol{u},
\]
where $\boldsymbol{u}:=\mathcal{E}u(\boldsymbol{t})$ is the agent's
state-dependent utility. Both the incentive and participation constraints
depend on the contract $\boldsymbol{t}$ only through state-dependent
utility $\boldsymbol{u}$. 

The orders in Section \ref{sec:orders} therefore carry through unchanged.
The sufficiency of the orders uses the same larger-feasible-set argument.
The necessity of the orders uses the same constructive proof.

\end{document}